\PassOptionsToPackage{dvipsnames}{xcolor} %
\documentclass{lmcs}
\pdfoutput=1
\usepackage[utf8]{inputenc}

\usepackage{lastpage}
\lmcsdoi{20}{4}{14}
\lmcsheading{}{\pageref{LastPage}}{}{}%
{Sep.~29,~2022}{Nov.~14,~2024}{}


\ACMCCS{[{\bf Theory of computation~Lambda calculus}]: 300; [{\bf Theory of computation~Type theory}]: 300; [{\bf Software and its engineering~Recursion}]: 300}

\usepackage{booktabs}   %
\usepackage{adjustbox}
\usepackage[inference]{semantic}
\usepackage{amsmath}\allowdisplaybreaks
\usepackage{mathpartir}
\usepackage{amsthm}
\usepackage{amssymb}
\usepackage{stmaryrd} %
\usepackage{xspace}
\usepackage{latexsym}
\usepackage{hyperref}
\usepackage{ifthen}
\usepackage{mathtools}
\usepackage{color}
\usepackage{listings}
\usepackage{verbatim}
\usepackage[colorinlistoftodos]{todonotes}
\usepackage{tikz}
\usetikzlibrary{positioning,shadows,arrows,calc,backgrounds,fit,shapes,shapes.multipart,decorations.pathreplacing,shapes.misc,patterns}
\usepackage{tikzscale}
\usepackage{xspace}
\usepackage[T1]{fontenc}
\usepackage{epigraph}
\usepackage[capitalize]{cleveref}
\usepackage{booktabs}
\usepackage{float}
\usepackage{etoolbox}
\usepackage{wrapfig}
\usepackage{scalerel}

\usepackage{bm}
\usepackage{nameref}
\usepackage{bussproofs}
\usepackage{textcomp}
\usepackage{textalpha}

\hypersetup{ pdfpagemode=UseOutlines, colorlinks=true, linkcolor=OliveGreen, citecolor=blue }


\newcommand{\citet}[1]{%
    \ifthenelse{\equal{#1}{morris_lambda-calculus_1968}}{Morris \cite{#1}}{%
    \ifthenelse{\equal{#1}{syn-con-rec-ty}}{Abadi and Fiore \cite{#1}}{%
    \ifthenelse{\equal{#1}{prta}}{Urzyczyn \cite{#1}}{%
    \ifthenelse{\equal{#1}{popl-journal}}{Devriese et al.\ \cite{#1}}{%
    \ifthenelse{\equal{#1}{isoequi-popl}}{Patrignani et al.\ \cite{#1}}{%
    \ifthenelse{\equal{#1}{Cai:2016:SFE:2914770.2837660}}{Cai et al.\ \cite{#1}}{%
    \ifthenelse{\equal{#1}{popl-journal,Hur:2011:KLR:1926385.1926402}}{Devriese et al.\ \cite{popl-journal}; Hur and Dreyer \cite{Hur:2011:KLR:1926385.1926402}}{%
    \ifthenelse{\equal{#1}{scsurvey}}{Patrignani et al.\ \cite{#1}}{%
    \ifthenelse{\equal{#1}{Devriese:2016:FCA:2837614.2837618}}{Devriese et al.\ \cite{#1}}{%
    \ifthenelse{\equal{#1}{Devriese:2016:FCA:2837614.2837618,popl-journal}}{Devriese et al.\ \cite{Devriese:2016:FCA:2837614.2837618,popl-journal}}{%
    \ifthenelse{\equal{#1}{exprPA}}{Parrow \cite{#1}}{%
    \ifthenelse{\equal{#1}{obs-pc-corr-trans}}{Schmidt-Schauß et al.\ \cite{#1}}{%
    \ifthenelse{\equal{#1}{max-embed}}{New et al.\ \cite{#1}}{%
    \ifthenelse{\equal{#1}{van_strydonck_linear_2019}}{Van Strydonck et al.\ \cite{#1}}{%
      Error case%
    }}}}}}}}}}}}}}%
}

\newcommand{\citep}[1]{%
  \cite{#1}%
}

\newcommand{\citeauthor}[1]{\hyperlink{cite.syn-con-rec-ty}{Abadi and Fiore}}

\newcommand{\mi}[1]{\ensuremath{\mathit{#1}}}

\newcommand{\mtt}[1]{\ensuremath{\mathtt{#1}}}
\newcommand{\mf}[1]{\ensuremath{\mathbf{#1}}}
\newcommand{\mk}[1]{\ensuremath{\mathfrak{#1}}}
\newcommand{\mc}[1]{\ensuremath{\mathcal{#1}}}
\newcommand{\ms}[1]{\ensuremath{\mathsf{#1}}}
\newcommand{\mb}[1]{\ensuremath{\mathbb{#1}}}

\newcommand{\isdef}[0]{\ensuremath{\mathrel{\overset{\makebox[0pt]{\mbox{\normalfont\tiny\sffamily def}}}{=}}}}

\newcommand{\relmiddle}[1]{\mathrel{}\middle#1\mathrel{}}

\newcommand\bnfdef{\ensuremath{\mathrel{::=}}}

\newcommand{\myset}[2]{\ensuremath{\left\{#1 ~\relmiddle|~ #2\right\}}}

\newcommand{\term}[0]{\ensuremath{\!\!\Downarrow}\xspace}
\newcommand{\termsl}[0]{\ensuremath{\src{\Downarrow}}\xspace}
\newcommand{\termt}[0]{\ensuremath{\trgb{\Downarrow}}\xspace}

\Crefname{lstlisting}{Listing}{Listings}

\Crefname{equation}{Rule}{Rules}

\newcommand{\genlang}[2]{\ensuremath{\lambda^{#1}_{#2}}}

\newcommand{\stlcu}[0]{\com{\genlang{u}{}}\xspace}
\newcommand{\stlcf}[0]{\src{\genlang{fx}{}}\xspace}
\newcommand{\stlcm}[0]{\stlcim}
\newcommand{\stlcim}[0]{\trg{\trgb{\lambda}^{\trgb{\mu}}_{I}}\xspace} %
\newcommand{\stlcem}[0]{\oth{\genlang{\mu}{E}}\xspace}

\newcommand{\stlcemind}[0]{\eqind{\genlang{\mu}{Ei}}\xspace}
\newcommand{\stlcemsimp}[0]{\eqind{\genlang{\mu}{Es}}\xspace}

\newcommand{\compskel}[3]{\ensuremath{\bl{\left\llbracket {#1} \right\rrbracket^{#2}_{#3}}}}
\newcommand{\comp}[1]{\compskel{\bl{#1}}{}{}}

\newcommand{\compstlc}[1]{\compskel{\src{#1}}{\stlcf}{\stlcm}}
\newcommand{\compstlce}[1]{\compskel{\src{#1}}{\stlcf}{\stlcem}}
\newcommand{\compstlcic}[1]{\compskel{\trg{#1}}{\stlcim}{\stlcem}}
\newcommand{\compstlcfi}[1]{\compstlc{#1}}
\newcommand{\compstlcfe}[1]{\compstlce{#1}}
\newcommand{\compstlcie}[1]{\compstlcic{#1}}

\newcommand{\funname}[1]{\mtt{#1}}
\newcommand{\fun}[2]{\ensuremath{{\bl{\funname{#1}\left(#2\right)}}}\xspace}

\newcommand{\ftv}[1]{\fun{ftv}{#1}}

\newcommand{\backtrskel}[3]{\ensuremath{\bl{\left\langle\!\left\langle {#1} \right\rangle\!\right\rangle^{#2}_{#3}}}}
\newcommand{\backtr}[1]{\backtrskel{#1}{}{}}

\newcommand{\backtrstlc}[1]{\backtrskel{\trg{#1}}{\stlcm}{\stlcf}}
\newcommand{\backtrstlce}[1]{\backtrskel{\oth{#1}}{\stlcem}{\stlcim}}

\newcommand{\backtrstlcef}[1]{\backtrskel{\oth{#1}}{\stlcem}{\stlcf}}

\newcommand{\backtrstlcei}[1]{\backtrstlce{#1}}
\newcommand{\backtrstlcif}[1]{\backtrstlc{#1}}

\newcommand{\langT}[0]{\ensuremath{L_{\com{trg}}}\xspace}
\newcommand{\langS}[0]{\ensuremath{L_{\com{src}}}\xspace}

\newcommand{\ctx}[0]{\ensuremath{\mk{C}}}
\newcommand{\ctxs}[0]{\src{\ctx}\xspace}
\newcommand{\ctxt}[0]{\trgb{\ctx}\xspace}%
\newcommand{\ctxo}[0]{\oth{\ctx}\xspace}%
\newcommand{\ctxc}[0]{\com{\ctx}\xspace}%

\newcommand{\ctxhs}[1]{\ctxs\src{\hole{#1}}\xspace}
\newcommand{\ctxht}[1]{\ctxt\trg{\hole{#1}}\xspace}%
\newcommand{\ctxho}[1]{\ctxo\oth{\hole{#1}}\xspace}%
\newcommand{\ctxhc}[1]{\ctxc\com{\hole{#1}}\xspace}%
\newcommand{\hole}[1]{\ensuremath{\left[#1\right]}}

\newcommand{\evalctx}[0]{\ensuremath{\mb{E}}}
\newcommand{\evalctxs}[0]{\src{\evalctx}\xspace}
\newcommand{\evalctxt}[0]{\trgb{\evalctx}\xspace}

\newcommand{\evalctxc}[0]{\com{\evalctx}\xspace}
\newcommand{\evalctxhs}[1]{\src{\evalctx\hole{#1}}\xspace}
\newcommand{\evalctxht}[1]{\trg{\evalctx\hole{#1}}\xspace}

\newcommand{\evalctxhc}[1]{\com{\evalctx\hole{#1}}\xspace}

\DeclareMathOperator\tobl{\bl{\to}}
\DeclareMathOperator\vdashbl{\bl{\vdash}}

\DeclareMathOperator\tot{\trgb{\to}}
\DeclareMathOperator\uplust{\trgb{\uplus}}
\DeclareMathOperator\timest{\trgb{\times}}

\newcommand{\Unit}[0]{Unit\xspace}
\newcommand{\Bool}[0]{Bool\xspace}

\newcommand{\Boolc}[0]{\com{{Bool}}\xspace}

\newcommand{\Bools}[0]{\src{{Bool}}\xspace}

\newcommand{\Boolt}[0]{\trg{{Bool}}\xspace}

\newcommand{\Boolo}[0]{\oth{{Bool}}\xspace}

\newcommand{\Unitc}[0]{\com{{Unit}}\xspace}
\newcommand{\Units}[0]{\src{{Unit}}\xspace}
\newcommand{\Unitt}[0]{\trg{{Unit}}\xspace}
\newcommand{\Unito}[0]{\oth{{Unit}}\xspace}
\newcommand{\truec}[0]{\com{{true}}\xspace}
\newcommand{\falsec}[0]{\com{{false}}\xspace}
\newcommand{\unitc}[0]{\com{{unit}}\xspace}
\newcommand{\trues}[0]{\src{{true}}\xspace}
\newcommand{\falses}[0]{\src{{false}}\xspace}
\newcommand{\units}[0]{\src{{unit}}\xspace}
\newcommand{\truet}[0]{\trg{{true}}\xspace}
\newcommand{\falset}[0]{\trg{false}\xspace}
\newcommand{\unitt}[0]{\trg{unit}\xspace}

\newcommand{\srce}[0]{\src{\emptyset}\xspace}
\newcommand{\trge}[0]{\trgb{\emptyset}\xspace}
\newcommand{\come}[0]{\com{\emptyset}\xspace}
\newcommand{\othe}[0]{\oth{\emptyset}\xspace}
\newcommand{\eqinde}[0]{\eqind{\emptyset}\xspace}

\newcommand{\neutcol}[0]{black}
\newcommand{\stlccol}[0]{RoyalBlue}
\newcommand{\ulccol}[0]{RedOrange}
\newcommand{\commoncol}[0]{black}    %
\newcommand{\othercol}[0]{CarnationPink}%
\newcommand{\orcol}[0]{Emerald}%

\definecolor{dark-gray}{gray}{0.35}
\definecolor{light-gray}{gray}{0.9}

\newcommand{\col}[2]{\ensuremath{{\color{#1}{#2}}}}

\newcommand{\src}[1]{\ms{\col{\stlccol}{#1}}}
\newcommand{\trgb}[1]{\ensuremath{\bm{\col{\ulccol }{#1}}}}
\newcommand{\trg}[1]{{\mf{\col{\ulccol }{#1}}}}
\newcommand{\oth}[1]{\mi{\col{\othercol }{#1}}}
\newcommand{\othb}[1]{\oth{#1}} %

\newcommand{\eqind}[1]{\mtt{\col{\orcol }{#1}}}

\newcommand{\bl}[1]{\col{\neutcol }{#1}}
\renewcommand{\com}[1]{\mi{\col{\commoncol }{#1}}}

\newcounter{typerule}
\crefname{typerule}{rule}{rules}

\newcommand{\typeruleInt}[5]{%
	\def\thetyperule{#1}%
	\refstepcounter{typerule}%
	\label{tr:#4}%
  \ensuremath{\begin{array}{c}#5 \inference{#2}{#3}\end{array}} 
}
\newcommand{\typerule}[4]{%
  \typeruleInt{#1}{#2}{#3}{#4}{\textsf{\scriptsize ({#1})} \\      }
}

\pgfdeclarelayer{background}
\pgfdeclarelayer{veryback}
\pgfdeclarelayer{veryback2}
\pgfdeclarelayer{veryback3}
\pgfdeclarelayer{back2}
\pgfdeclarelayer{foreground}
\pgfsetlayers{veryback3,veryback2,veryback,background,back2,main,foreground}

\newcommand{\tikzpic}[1]{
\begin{tikzpicture}[shorten >=1pt,auto,node distance=6mm]
\tikzstyle{state} =[fill=white,minimum size=4pt]
\tikzstyle{field} =[fill=gray!5,draw=black!70, rectangle, minimum width={width("whiskersfieldww")+2pt}]]
#1
\end{tikzpicture}
}

\newcommand{\mytoprule}[1]{\vspace{1mm}\noindent\hrulefill\ \raisebox{-0.5ex}{\fbox{\ensuremath{#1}}} \hrulefill\hrulefill\hrulefill\vspace{0.5mm}}

\def\botrule{\vspace{0mm}\hrule\vspace{2mm}}

\definecolor{mygreen}{rgb}{0,0.6,0}
\definecolor{mygray}{rgb}{0.5,0.5,0.5}
\definecolor{mymauve}{rgb}{0.58,0,0.82}

\lstdefinelanguage{Java} %
{morekeywords={abstract, all, and, as, assert, but, check, disj, else, exactly, extends, fact, for, fun, iden, if, iff, implies, in, Int, int, let, lone, module, no, none, not, one, open, or, part, pred, run, seq, set, sig, some, sum, then, univ, package, class, public, private, null, return, new, interface, extern, object, implements, System, static, super, try , catch, throw, throws, Unit, var, val, of, principal, trust},
sensitive=true,
keywordstyle=\bfseries\color{green!40!black},
commentstyle=\itshape\color{purple!40!black},
morecomment=[l][\small\itshape\color{purple!40!black}]{//},
identifierstyle=\color{blue},
stringstyle=\color{orange},
basicstyle=\small,
basicstyle={\small\ttfamily},
numbers=left,
numberstyle=\tiny\color{mygray},
tabsize=2,
numbersep=3pt,
breaklines=true,
lineskip=-2pt,
stepnumber=1,
captionpos=b,
breaklines=true,
breakatwhitespace=false,
showspaces=false,
showtabs=false,
float=!h,
columns=fullflexible,escapeinside={(*@}{@*)},
moredelim=**[is][\color{red!60}]{@}{@},
literate={->}{{$\to$}}1 {^}{{$\mspace{-3mu}\widehat{\quad}\mspace{-3mu}$}}1
{<}{$<$ }2 {>}{$>$ }2 {>=}{$\geq$ }2 {=<}{$\leq$ }2
{<:}{{$<\mspace{-3mu}:$}}2 {:>}{{$:\mspace{-3mu}>$}}2
{=>}{{$\Rightarrow$ }}2 {+}{$+$ }2 {++}{{$+\mspace{-8mu}+$ }}2
{<=>}{{$\Leftrightarrow$ }}2 {+}{$+$ }2 {++}{{$+\mspace{-8mu}+$ }}2
{\~}{{$\mspace{-3mu}\widetilde{\quad}\mspace{-3mu}$}}1
{!=}{$\neq$ }2 {*}{${}^{\ast}$}1 %
{\#}{$\#$}1
}
\DeclareMathOperator\ceq{\ensuremath{\mathrel{\simeq_{{ctx}}}}}

\DeclareMathOperator\ceqs{\src{\ceq}}
\DeclareMathOperator\ceqt{\trgb{\ceq}}
\DeclareMathOperator\ceqo{\oth{\ceq}}

\DeclareMathOperator\ceqc{\com{\ceq}}

\theoremstyle{plain}

\newtheorem{theorem}[thm]{Theorem}
\newtheorem{corollary}[thm]{Corollary}
\newtheorem{lemma}[thm]{Lemma}

\theoremstyle{definition}

\newtheorem{definition}[thm]{Definition}

\newtheorem{example}[thm]{Example}

\Crefname{corollary}{Corollary}{Corollaries}
\Crefname{informal}{Definition}{Definition}
\Crefname{assumption}{Assumption}{Assumptions}
\crefname{assumption}{Assumption}{Assumptions}
\Crefname{property}{Property}{Properties}
\crefname{property}{Property}{Properties}

\newcommand{\lam}[2]{\ensuremath{\lambda #1\ldotp #2}}
\newcommand{\lamt}[2]{\ensuremath{\trgb{\lambda} #1\ldotp #2}}

\newcommand{\pair}[1]{\ensuremath{\left\langle#1\right\rangle}}
\newcommand{\projone}[1]{\ensuremath{#1.1}}
\newcommand{\projtwo}[1]{\ensuremath{#1.2}}
\newcommand{\proji}[1]{\ensuremath{#1.i}}
\newcommand{\fix}[1]{\ensuremath{fix}_{#1}~ }
\newcommand{\case}{\ensuremath{{case}}}
\newcommand{\of}{\ensuremath{{of}}}
\newcommand{\caseof}[3]{\ensuremath{{case}~#1~{of}~\inl{x_1}\mapsto #2\mid\inr{x_2}\mapsto #3}}

\newcommand{\casefoldedc}[3]{
	\com{\case}~#1~\com{\of}\left|\begin{aligned}
			&
			\com{\inl{x_1}\mapsto} #2
			\\
			&
			\com{\inr{x_2}\mapsto} #3
		\end{aligned}\right.
}
\newcommand{\casefoldeds}[3]{
	\src{\case}~#1~\src{\of}\left|\begin{aligned}
			&
			\src{\inl{x_1}\mapsto} #2
			\\
			&
			\src{\inr{x_2}\mapsto} #3
		\end{aligned}\right.
}
\newcommand{\casefoldedt}[3]{
	\trg{\case}~#1~\trg{\of}\left|\begin{aligned}
			&
			\trg{\inl{x_1}\mapsto} #2
			\\
			&
			\trg{\inr{x_2}\mapsto} #3
		\end{aligned}\right.
}

\newcommand{\ifte}[3]{\ensuremath{\com{if}~#1~\com{then}~#2~\com{else}~#3}}
\newcommand{\iftes}[3]{\ensuremath{\src{if}~#1~\src{then}~#2~\src{else}~#3}}
\newcommand{\iftet}[3]{\ensuremath{\trg{if}~#1~\trg{then}~#2~\trg{else}~#3}}

\newcommand{\inl}[1]{\ensuremath{{inl}~#1}}
\newcommand{\inr}[1]{\ensuremath{{inr}~#1}}

\newcommand{\fold}[1]{\ensuremath{{fold}_{\trg{#1}}}}
\newcommand{\unfold}[1]{\ensuremath{{unfold}_{\trg{#1}}}}

\newcommand{\alpt}[0]{\trgb{\alpha}}

\newcommand{\tat}[0]{\trgb{\tau}}

\newcommand{\Gat}[0]{\trgb{\Gamma}}

\newcommand{\matgen}[2]{\ensuremath{\mu #1\ldotp#2}}
\newcommand{\mat}[0]{\matgen{\alpha}{\tau}}

\newcommand{\matgent}[2]{\ensuremath{\trgb{\mu} #1\ldotp#2}}
\newcommand{\matt}[0]{\matgent{\alpt}{\tat}}

\newcommand{\redgen}[1]{\ensuremath{ \hookrightarrow^{#1} }}

\newcommand{\red}[0]{\redgen{}}

\newcommand{\redp}[0]{\ensuremath{ \hookrightarrow_{p} }}

\AtBeginEnvironment{example}{\pushQED{\qed}}
\AtEndEnvironment{example}{\popQED}

\makeatletter
\xdef\@thefnmark{\@empty}

\newcommand{\Thmref}[1]{\Cref{#1}~(\nameref{#1})}
\makeatother

\newcommand{\subst}[2]{\ensuremath{\bl{\left[#1\bl{/}#2\right]}}} %
\newcommand{\subs}[2]{\subst{\src{#1}}{\src{#2}}}
\newcommand{\subt}[2]{\subst{\trg{#1}}{\trg{#2}}}
\newcommand{\subo}[2]{\subst{\oth{#1}}{\oth{#2}}}
\newcommand{\subi}[2]{\subst{\eqind{#1}}{\eqind{#2}}}
\newcommand{\subc}[2]{\subst{\com{#1}}{\com{#2}}}

\newcounter{hps}
\crefname{hps}{}{}

\newcommand{\genlogrelname}[2]{\ensuremath{\mi{LR}^{#1}_{#2}}\xspace}
\newcommand{\lrfi}[0]{\genlogrelname{\src{fx}}{\iso{\isob{\mu} I}}}
\newcommand{\lrfe}[0]{\genlogrelname{\src{fx}}{\equi{\mu E}}}
\newcommand{\lrie}[0]{\genlogrelname{\iso{\isob{\mu} I}}{\equi{\mu E}}}

\newcommand{\oftypefi}[1]{\fun{oftype^{\ftoitext}}{#1}}
\newcommand{\oftypefe}[1]{\fun{oftype^{\ftoetext}}{#1}}
\newcommand{\oftypeie}[1]{\fun{oftype^{\itoetext}}{#1}}
\newcommand{\oftypes}[1]{\src{oftype\left(#1\right)}}
\newcommand{\oftypet}[1]{\trg{oftype\left(#1\right)}}
\newcommand{\oftypeo}[1]{\oth{oftype\left(#1\right)}}

\newcommand{\logrelgen}[1]{\ensuremath{\operatorname{\approx}^{#1}}}
\DeclareMathOperator\logrel{\logrelgen{}}
\DeclareMathOperator\bothlogrel{\logrelgen{}}
\newcommand{\underapproxlogrelgen}[1]{\ensuremath{\operatorname{\lesssim}^{#1}}}
\DeclareMathOperator\underlogrel{\underapproxlogrelgen{}}
\newcommand{\overapproxlogrelgen}[1]{\ensuremath{\operatorname{\gtrsim}^{#1}}}
\DeclareMathOperator\overlogrel{\overapproxlogrelgen{}}

\newcommand{\arbsim}{\ensuremath{\triangledown}}%

\newcommand{\genlogrel}[0]{\arbsim}
\DeclareMathOperator\anylogrel{\genlogrel{}}

\newcommand{\underlogreln}[1]{\ensuremath{\mathrel{\lesssim_{#1}}}}
\newcommand{\overlogreln}[1]{\ensuremath{\mathrel{\gtrsim_{#1}}}}
\newcommand{\anylogreln}[1]{\ensuremath{\mathrel{\triangledown_{#1}}}}

\newcommand{\genrel}[4]{\ensuremath{\bl{\mc{#1}\left\llbracket{#2}\right\rrbracket^{#3}_{#4}}}}
\newcommand{\valrel}[1]{\genrel{V}{\src{#1}}{}{\anylogrel}}

\newcommand{\contrel}[1]{\genrel{K}{\src{#1}}{}{\anylogrel}}

\newcommand{\termrel}[1]{\genrel{E}{\src{#1}}{}{\anylogrel}}

\newcommand{\envrel}[1]{\genrel{G}{\src{#1}}{}{\anylogrel}}

\newcommand{\valrele}[1]{\genrel{V}{\trgb{#1}}{}{\anylogrel}}

\newcommand{\langsp}[1]{\ensuremath{\mi{#1}}}
\newcommand{\langspfun}[2]{\ensuremath{\langsp{#1}(#2)}}
\newcommand{\W}[0]{\langsp{W}\xspace}

\newcommand{\stepsfungen}[2]{\ensuremath{\langspfun{lev^{#1}}{#2}}}

\newcommand{\stepsfun}[1]{\stepsfungen{}{#1}}

\newcommand{\latergen}[1]{\ensuremath{\triangleright^{#1}}}
\DeclareMathOperator\later{\latergen{}}
\newcommand{\laterfun}[1]{\langspfun{\triangleright}{#1}}

\newcommand{\obswfungen}[3]{\langspfun{O^{#1}}{#2}_{#3}}

\newcommand{\obsfun}[2]{\obswfungen{}{#1}{#2}}

\newcommand{\wobsfun}[2]{\mc{W}({#1})_{#2}}

\newcommand{\futwgen}[1]{\ensuremath{\sqsupseteq^{#1}}}
\newcommand{\strfutwgen}[1]{\ensuremath{\sqsupset_{\later}^{#1}}}
\DeclareMathOperator\futw{\futwgen{}}

\DeclareMathOperator\strfutw{\strfutwgen{}}

\newcommand{\precise}{\ensuremath{\mtt{precise}}}
\newcommand{\imprecise}{\ensuremath{\mtt{imprecise}}}

\newcommand{\emuldvtext}{EmulT}
\newcommand{\uvaltext}{BtT}

\newcommand{\uval}[1]{\src{\uvaltext^{\ftoitext}_{#1}}}
\newcommand{\emuldv}[1]{\src{\emuldvtext^{\ftoitext}_{#1}}}
\newcommand{\psd}[1]{\src{\hat{#1}}}
\newcommand{\srctotrgty}[1]{\fun{fxToIs}{#1}}
\newcommand{\emtotau}[1]{\fun{repEmul^{\ftoitext}}{#1}}
\newcommand{\toemul}[2]{\fun{toEmul_{\src{#2}}}{#1}}

\newcommand{\uvalfe}[1]{\src{\uvaltext^{\ftoetext}_{#1}}}
\newcommand{\emuldvfe}[1]{\src{\emuldvtext^{\ftoetext}_{#1}}}
\newcommand{\psdfe}[1]{\src{\hat{#1}}}
\newcommand{\srctoothty}[1]{\fun{fxToEq}{#1}}
\newcommand{\emtotaufe}[1]{\fun{repEmul^{\ftoetext}}{#1}}
\newcommand{\toemulfe}[2]{\fun{toEmul^{\ftoetext}_{\src{#2}}}{#1}}

\newcommand{\uvalic}[1]{\trg{\uvaltext^{\itoetext}_{#1}}}
\newcommand{\emuldvic}[1]{\ensuremath{\trg{\emuldvtext^{\itoetext}_{#1}}}}
\newcommand{\psdic}[1]{\trg{\hat{#1}}}
\newcommand{\srctotrgtyic}[1]{\fun{isToEq}{#1}}
\newcommand{\emtotauic}[1]{\fun{repEmul^{\itoetext}}{#1}}

\newcommand{\ftoetext}[0]{fE}
\newcommand{\ftoitext}[0]{fI}
\newcommand{\itoetext}[0]{IE}

\newcommand{\myomega}[0]{omega}
\newcommand{\unk}[0]{unk}

\newcommand{\emulate}[2]{\src{emulate^{\ftoitext}_{#1}\left(\trg{#2}\right)}}

\newcommand{\upgrade}[1]{\src{upgrade^{\ftoitext}_{#1}~}}
\newcommand{\downgrade}[1]{\src{downgrade^{\ftoitext}_{#1}~}}

\newcommand{\casetag}[1]{\src{case^{\ftoitext}_{#1}~}}
\newcommand{\inject}[1]{\src{inject^{\ftoitext}_{#1}~}}
\newcommand{\extract}[1]{\src{extract^{\ftoitext}_{#1}~}}

\newcommand{\indn}[1]{\src{in\text{-}dn^{\ftoitext}_{#1}~}}
\newcommand{\caseup}[1]{\src{case\text{-}up^{\ftoitext}_{#1}~}}

\newcommand{\upgradepar}[2]{\src{upgrade^{\ftoitext}_{#1}\left(\src{#2}\right)}}
\newcommand{\downgradepar}[2]{\src{downgrade^{\ftoitext}_{#1}\left(\src{#2}\right)}}

\newcommand{\casetagpar}[2]{\src{case^{\ftoitext}_{#1}\left(\src{#2}\right)}}
\newcommand{\injectpar}[2]{\src{inject^{\ftoitext}_{#1}\left(\src{#2}\right)}}

\newcommand{\indnfe}[1]{\src{in\text{-}dn^{\ftoetext}_{#1}~}}
\newcommand{\caseupfe}[1]{\src{case\text{-}up^{\ftoetext}_{#1}~}}

\newcommand{\casetagfe}[1]{\src{case^{\ftoetext}_{#1}~}}

\newcommand{\injectef}[1]{\src{inject^{\ftoetext}_{#1}~}}
\newcommand{\extractef}[1]{\src{extract^{\ftoetext}_{#1}~}}
\newcommand{\emulateef}[2]{\src{emulate^{\ftoetext}_{#1}\left(\oth{#2}\right)}}

\newcommand{\injectfe}[1]{\injectef{#1}}
\newcommand{\extractfe}[1]{\extractef{#1}}

\newcommand{\upgradefe}[1]{\src{upgrade^{\ftoetext}_{#1}~}}
\newcommand{\downgradefe}[1]{\src{downgrade^{\ftoetext}_{#1}~}}

\newcommand{\injectic}[1]{\trg{inject^{\itoetext}_{#1}~}}
\newcommand{\extractic}[1]{\trg{extract^{\itoetext}_{#1}~}}
\newcommand{\emulateic}[2]{\trg{emulate^{\itoetext}_{#1}\left(\oth{#2}\right)}}

\newcommand{\casetagic}[1]{\trg{case^{\itoetext}_{#1}~}}

\newcommand{\indnic}[1]{\trg{in\text{-}dn^{\itoetext}_{#1}~}}
\newcommand{\caseupic}[1]{\trg{case\text{-}up^{\itoetext}_{#1}~}}

\newcommand{\upgradeic}[1]{\trg{upgrade^{\itoetext}_{#1}~}}
\newcommand{\downgradeic}[1]{\trg{downgrade^{\itoetext}_{#1}~}}

\newcommand{\fa}[0]{\mi{FA}\xspace}

\DeclareMathOperator\tyeqind{\bl{\bumpeq}}
\DeclareMathOperator\ntyeqind{\bl{\not\bumpeq}}
\newcommand{\tyeqindbin}[3]{\eqind{#2} \tyeqind \eqind{#3}}

\DeclareMathOperator\tyeq{\bl{\circeq}}

\newcommand{\tyeqbin}[2]{\ensuremath{\oth{#1} \tyeq \oth{#2}}}

\newcommand*{\lMuCount}{\fun{lMuCount}}

\newcommand{\Efrac}[2]{%
  \mathchoice
    {\ooalign{%
      $\genfrac{}{}{1.2pt}0{#1}{#2}$\cr%
      $\color{white}\genfrac{}{}{1pt}0{\phantom{#1}}{\phantom{#2}}$}}%
    {\ooalign{%
      $\genfrac{}{}{1.2pt}1{#1}{#2}$\cr%
      $\color{white}\genfrac{}{}{1pt}1{\phantom{#1}}{\phantom{#2}}$}}%
    {\ooalign{%
      $\genfrac{}{}{1.2pt}2{#1}{#2}$\cr%
      $\color{white}\genfrac{}{}{1pt}2{\phantom{#1}}{\phantom{#2}}$}}%
    {\ooalign{%
      $\genfrac{}{}{1.2pt}3{#1}{#2}$\cr%
      $\color{white}\genfrac{}{}{1pt}3{\phantom{#1}}{\phantom{#2}}$}}%
}

\newcommand{\cotyperule}[4]{
	\cotyperuleInt{#1}{#2}{#3}{#4}{\textsf{\scriptsize ({#1})}}
}
\newcommand{\cotyperuleInt}[5]{%
\def\thetyperule{#1}%
	\refstepcounter{typerule}%
	\label{tr:#4}%
	\ensuremath{%
		\begin{array}{c}
		#5
		\\
		\Efrac{%
			\begin{aligned}%
				#2
			\end{aligned}
		}{%
			\begin{aligned}%
				#3
			\end{aligned}
		}%
		\end{array}
	}
}

\newcommand{\size}[1]{\fun{size}{#1}}

\newcommand{\shrinksym}[0]{\lightning}

\newcommand{\shrinks}[1]{\ensuremath{\src{\shrinksym_{\!#1}}}\xspace}
\newcommand{\shrinkt}[1]{\ensuremath{\trgb{\shrinksym_{\!#1}}}\xspace}

\newcommand{\shrinkc}[1]{\ensuremath{\com{\shrinksym_{\!#1}}}\xspace}

\newcommand{\boundedtermsym}[0]{\Downarrow}

\newcommand{\bterms}[1]{\ensuremath{\src{\boundedtermsym_{#1}}}\xspace}
\newcommand{\btermt}[1]{\ensuremath{\trgb{\boundedtermsym_{#1}}}\xspace}

\newcommand{\btermc}[1]{\ensuremath{\com{\boundedtermsym_{#1}}}\xspace}

\newcommand{\newterm}[0]{\text{size-bound termination}\xspace}
\newcommand{\newterms}[0]{\text{size-bound terminates}\xspace}

\newcommand{\NewTerm}[0]{\text{Size-Bound Termination}\xspace}

\newcommand{\tprobs}[0]{\src{t}\xspace}
\newcommand{\tprobt}[0]{\iso{t}\xspace} %

 \newcommand{\iso}[1]{\trg{#1}}
\newcommand{\isob}[1]{\trgb{#1}}
\newcommand{\equi}[1]{\oth{#1}}

\renewcommand{\lMuCount}[1]{\fun{lmc}{#1}}

\DeclareUnicodeCharacter{03BB}{$\lambda$}

\AtBeginEnvironment{tikzpicture}{\catcode`$=3 } %

\colorlet{NAVYBLUE}{NavyBlue}

\crefname{lem}{Lemma}{Lemmas}
\crefname{thm}{Theorem}{Theorems}
\crefname{defi}{Definition}{Definitions}
\crefname{exa}{Example}{Examples}

\begin{document}

\title[On the Semantic Expressiveness of Iso- and Equi-Recursive Type]{On the Semantic Expressiveness \\ of Iso- and Equi-Recursive Types}         %
\titlecomment{{\lsuper*}extended version of the paper in POPL'21, now including a fixed and mechanized proof. More details are in \Cref{sec:comparison}.}

\author[D.~Devriese]{Dominique Devriese\lmcsorcid{0000-0002-3862-6856}}[a]	%
\address{DistriNet, KU~Leuven, Belgium}	%
\email{dominique.devriese@kuleuven.be}  %

\author[E.~Martin]{Eric M. Martin\lmcsorcid{0000-0002-5739-2186}}[b]	%
\address{ Jane Street Capital }	%
\email{emartin@janestreet.com}  %

\author[M.~Patrignani]{Marco Patrignani\lmcsorcid{0000-0003-3411-9678}}[c]  %
\address{University of Trento}  %
\email{marco.patrignani@unitn.it}  %

\begin{abstract}
Recursive types extend the simply-typed lambda calculus (STLC) with the additional expressive power to enable diverging computation and to encode recursive data-types (e.g., lists).
Two formulations of recursive types exist: iso-recursive and equi-recursive.
The relative advantages of iso- and equi-recursion are well-studied when it comes to their impact on type-inference.
However, the relative semantic expressiveness of the two formulations remains unclear so far.
This paper studies the semantic expressiveness of STLC with iso- and equi-recursive types, proving that these formulations are \emph{equally expressive}.
In fact, we prove that they are both as expressive as STLC with only term-level recursion.
We phrase these equi-expressiveness results in terms of full abstraction of three canonical compilers between these three languages (STLC with iso-, with equi-recursive types and with term-level recursion).
Our choice of languages allows us to study expressiveness when interacting over both a simply-typed and a recursively-typed interface.
The three proofs all rely on a typed version of a proof technique called approximate backtranslation.

Together, our results show that there is no difference in semantic expressiveness between STLCs with iso- and equi-recursive types.
In this paper, we focus on a simply-typed setting but we believe our results scale to more powerful type systems like System~F.
\end{abstract}

\maketitle

\begin{center}
  \small\it 
  To present notions more clearly, this paper uses syntax highlighting accessible to both colourblind and black \& white readers \citep{patrignani2020use}.
  For a better experience, please print or view this in colour.%
  \\
    Specifically, we use a \src{blue}, \src{sans\text{-}serif} font for \src{STLC} with the \src{fix} operator, a \iso{red}, \iso{bold} font for \iso{STLC} with \iso{iso}-\iso{recursive} types, and \equi{pink}, \equi{italics} font for \equi{STLC} with \equi{coinductive} \equi{equi}-\equi{recursive} types.
    Elements common to all languages are typeset in a \com{\commoncol}, \com{italic} font (to avoid repetition).
\end{center}

\section{Introduction}\label{sec:intro}
Recursive types were first proposed by \citet{morris_lambda-calculus_1968} as a way to recover divergence from the untyped lambda calculus in a simply-typed lambda calculus.
They also enable the definition of recursive data-types such as lists, trees, and Lisp S-expressions in typed languages.

Morris' original formulation was equi-recursive: a type \equi{\mat} was regarded as an infinite type and considered equal to its unfolding \equi{\tau\subo{\mat}{\alpha}}.
Subsequent formulations (e.g., \citet{syn-con-rec-ty}) use different type equality relations.
In this paper we will work with \stlcem{}: a standard simply-typed lambda calculus with coinductive equi-recursive types (e.g. \citep{Cai:2016:SFE:2914770.2837660}).

Years after Morris' formulation of recursive types, a different one appeared (e.g. \citep{harper_type_1993,gordon_edinburgh_1979}), where the two types are not considered equal, but \emph{isomorphic}: values can be converted from \iso{\matt} to \iso{\tat\subt{\matt}{\alpt}} and back using explicit \iso{\fold{}} and \iso{\unfold{}} annotations in terms.
These annotations are used to guide typechecking, but they also have a significance at runtime: an explicit reduction step is needed to cancel them out: \iso{\unfold{\matt}~(\fold{\matt}~v) \red v}.
In this paper, we work with a standard iso-recursive calculus \stlcim{}. 

The relation between these two formulations has been studied by \citet{syn-con-rec-ty} and \citet{prta} (the latter focusing on positive recursive types).
Specifically, they show that any term typable in one formulation can also be typed in the other, possibly by adding extra \iso{\unfold{}} or \iso{\fold{}} annotations.
Additionally, \citeauthor{syn-con-rec-ty} prove that for types considered equal in the equi-recursive system, there exist coercion functions in the iso-recursive formulation that are mutually inverse in the (axiomatised) program logic.
The isomorphism properties are proved in a logic for the iso-recursive language (which is only conjectured to be sound), and the authors do not consider an operational semantics.

The relative semantic expressiveness of the two formulations, however, has remained yet unexplored.
In principle, executions that are converging in the equi-recursive language may become diverging in the iso-recursive setting because of the extra fold-unfold reductions.
Because of this, it is unclear whether the two formulations of recursive types produce equally expressive languages.

Concretely, in this paper, we study the expressive power of \stlcim{} and \stlcem{} when interacting over two kinds of language interfaces.
The first is characterized by simply-typed lambda calculus types, which do not mention recursive types themselves.
We consider implementations of this interface in \stlcf{}, a simply typed lambda calculus with term-level recursion in the form of a primitive fixpoint operator.
We embed these implementations into both \stlcim{} and \stlcem{} using two so-called canonical compilers, i.e., compilers that map any construct of the source language into the same -- or the closest -- construct of the target.
We show that if two \stlcf{} terms cannot be distinguished by \stlcf{} contexts, then the same is true for both \stlcim{} and \stlcem{} contexts, i.e., the compiler is fully abstract.
Additionally, we consider STLC types that contain recursive types themselves as interfaces.
We take implementations of them in \stlcim and a canonical compiler for them into \stlcem.
We show that this compiler is also fully abstract.
These three fully-abstract compilation results establish the equi-expressiveness of \stlcim{}, \stlcem{}, and \stlcf{} contexts, interacting over simply-typed interfaces with and without recursive types.
Moreover, these three fully-abstract compilation results have been completely formalised in the Coq proof assistant.

Proving full abstraction for a compiler is notoriously hard, particularly in the preservation direction, i.e., showing that equivalent source terms get compiled to equivalent target terms.
Informally, it requires showing that any behaviour (e.g., termination) of target program contexts can be replicated by source program contexts.
Demonstrating such a claim is particularly complicated in our setting since \stlcem contexts have coinductive (and thus infinite) type equality derivations.
To be able to prove fully-abstract compilation, we adopt the approximate backtranslation proof technique of \citet{popl-journal}.
This technique relies on two key components: a cross-language approximation relation between source and target terms (and source and target program contexts) and a backtranslation function from target to source program contexts.
Intuitively, the approximation relation is used to tell when a source and a target term (or program context) equi-terminate; we use step-indexed logical relations to define this and rely on the step as the measure for the approximation.
The backtranslation is a function that takes a target program context and produces a source program context that approximates the target one.
This is particularly appropriate for backtranslating \stlcem program contexts, since we show that it is sufficient to approximate their coinductive derivations instead of replicating them precisely.

We construct three backtranslations: from \stlcim and \stlcem contexts respectively into \stlcf ones and from \stlcem contexts into \stlcim ones.
We do so by defining a family of types for backtranslated terms that is not just indexed by the approximation level but also by the target type of the backtranslated term.
To the best of our knowledge, this is a novel approach, since all existing work relies on a single type for backtranslated terms~\citep{popl-journal,max-embed}.

For proving the correctness of these backtranslations, we define a step-indexed logical relation to express when compiled and backtranslated terms approximate each other.
While the logical relation is largely the same for the different compilers and backtranslations, differences in the language semantics impose that we treat backtranslated \stlcim terms differently from \stlcem.
Like previous work~\citep{popl-journal,max-embed}, we use a step-indexed logical relation that relates terms (and values) across languages so long as they equi-terminate.
In previous work, the step-indexed logical relation approximates (or, relates) terms (and values) up to an index that is related to the amount of steps that are required for termination.
In this work, we change that approximation to also consider an additional bound on the size of terms encountered during termination.
To provide this new bound, we introduce a novel notion of termination, called \newterm, and state that terms are related when \newterm of one term implies termination of the other.
The need for an additional bound (and thus for \newterm) arose while mechanising these proofs in the Coq proof assistant, as this led to the discovery of a bug in the previous proofs (as we describe in \Cref{sec:comparison}).
The additional bound lets us reason explicitly about the finiteness of values encountered during reductions, and it lets us go through those cases that broke certain proofs (as we describe in detail in \Cref{ex:need-newterm} in \Cref{sec:injext}).

\subsection{Using Fully Abstract Compilation to Compare Language Expressiveness}\label{sec:fac-for-langexp}
To study language expressiveness meaningfully, it is important to phrase the question properly.
If we just consider programs that receive a natural number and return a boolean, then both languages will allow expressing the same set of algorithms, simply by their Turing completeness~\citep{Mitchell-expr-pow}.

The question of comparing language expressiveness is more interesting if we consider programs that interact over a richer interface.
Consider, for example, a term \com{t} from the simply-typed lambda calculus embedded into either the \stlcim or \stlcem calculus.  
An interesting question is whether there are ways in which \stlcem{} contexts (i.e., larger programs) can interact with \com{t} that contexts in \stlcim{} cannot.
The use of contexts in different languages interacting with a common term as a way of measuring language expressiveness has a long history~\citep{Felleisen-expr-pow,Mitchell-expr-pow}, mostly in the study of process calculi~\citep{exprPA}.
In this setting, equal expressiveness of programming languages is sometimes argued for by proving the existence of a fully-abstract compiler from one language to the other \citep{gorla-fa}.
Such a compiler translates contextually-equivalent terms in a source language (indicated as \langS) to contextually-equivalent terms in a target language (indicated as \langT)~\citep{DBLP:conf/icalp/Abadi98,scsurvey}.
That is, if contexts cannot distinguish two terms in \langS, they will also not be able to distinguish them after the compilation to \langT.

Let us now argue why the choice of fully-abstract compilation as a measure of the relative expressiveness of programming languages is the right one in our setting.
After all, several researchers have pointed out that the mere existence of a fully-abstract compilation is not in itself meaningful and only compilers that are sufficiently well-behaved should be considered~\citep{exprPA,gorla-fa}. %
The reason for this is that one can build a degenerate fully-abstract compiler that shows both languages having an equal amount (cardinality) of equivalence classes for terms.
This would indicate that the languages are equally-expressive, but unfortunately this is also trivial to satisfy \citep{exprPA}.
These degenerate examples, as such, clarify the necessity for well-behavedness of the compiler.
However, we have not found a clear argument explaining why well-behaved fully-abstract compilation implies equi-expressiveness of languages, so here it is. %

In our opinion (and we believe this point has not yet been made in the literature), the issue is that fully-abstract compilation results measure language expressiveness \emph{not} by verifying that they can express the same \emph{terms}, but that they can express the same \emph{contexts}.
Defining when a context in \langS{} is the same as a context in \langT{} is hard, and therefore fully-abstract compilation simply requires that \langT{} contexts can express the interaction of \langS contexts with any term that is shared between both languages.
The role of the compiler, the translation from \langS{} to \langT{}, is simply to obtain this common term against which expressiveness of contexts in both languages can be measured.

In other words, expressiveness of a programming language is only meaningful with respect to a certain interface and the role of the compiler is to map \langS implementations of this interface to \langT implementations.
In a sense, the \langS implementation of the interface should be seen as an expressiveness challenge for \langS contexts and the compiler translates it to the corresponding challenge in \langT.
As such, the compiler should be seen as part of the definition of equi-expressiveness and the well-behavedness requirement is there to make sure the \langS challenge is translated to ``the same'' challenge in \langT.
Fortunately, in this work we only rely on canonical compilers that provide the most intuitive translation for a term in our source languages into ``the same'' term in our target ones.
Thus, we believe that in our setting using fully-abstract compilation is the right tool to measure the relative expressiveness of programming languages.

\subsection{Contributions and Outline}
To summarize, the key contribution of this paper is the proof that iso- and coinductive equi-recursive typing are equally expressive.
This result is achieved via the following contributions (depicted in \Cref{fig:dia}).
\begin{itemize}
  \item An adaptation of the approximate backtranslation proof technique to operate on families of backtranslation types that are type-indexed on target types
  \item An adaptation of the proof technique to be more precise when relating terms cross-language by relying on the notion of \newterm;
  \item A proof that the compiler from \stlcf to \stlcim is fully abstract with an approximate backtranslation;
  \item A proof that the compiler from \stlcim to \stlcem is fully abstract with an approximate backtranslation;
  \item A proof that the compiler from \stlcf to \stlcem is fully abstract with an approximate backtranslation;
  \item The mechanisation of these three proofs in the Coq proof assistant.
\end{itemize}

\begin{figure}[!htb]
  \centering
  \tikzpic{
    \node[](lam-fix){\stlcf};
    \node[below left =2 and 2 of lam-fix](lam-mu-i){\stlcim};
    \node[below right =2 and 2 of lam-fix](lam-mu-e){\stlcem};

    \draw[->,thick] ([yshift=.3em]lam-mu-i.0) to node[midway, align=left, font = \footnotesize,above](){ \compstlcie{\cdot} } ([yshift=.3em]lam-mu-e.180);
    \draw[->,thick,dotted] ([yshift=-.3em]lam-mu-e.180) to node[midway, align=left, font = \footnotesize,below](){ \backtrstlcei{\cdot} } ([yshift=-.3em]lam-mu-i.0);

    \draw[->,thick] (lam-fix.180) to 
      node[midway, align=left, font = \footnotesize,above,sloped](){ \compstlcfi{\cdot} } (lam-mu-i.90);
    \draw[->,thick,dotted] (lam-mu-i.north east) to node[midway, align=left, font = \footnotesize,below,sloped](){ \backtrstlcif{\cdot} } (lam-fix.south west);

    \draw[->,thick] (lam-fix.0) to 
      node[midway, align=left, font = \footnotesize,above,sloped](){ \compstlcfe{\cdot} } (lam-mu-e.90);
    \draw[->,thick,dotted] (lam-mu-e.north west) to node[midway, align=left, font = \footnotesize,below,sloped](){ \backtrstlcef{\cdot} } (lam-fix.south east);

  }
\caption{
  Our contributions, visually.
  Full arrows indicate canonical embeddings \comp{\cdot} while dotted ones are (approximate) backtranslations \backtr{\cdot}.
  Translations' superscripts indicate input languages while their subscripts indicate output languages.
  \label{fig:dia}
}
\end{figure}
Note that technically, we can derive the compiler and backtranslation between \stlcf and \stlcem by composing the compilers and backtranslations through \stlcim.
We present this result as a stand-alone one because it offers insights on proofs of fully-abstract compilation for languages with coinductive notions.

The remainder of this paper is organised as follows.
We first formalise the languages we use (\stlcf, \stlcim and \stlcem) as well as the cross-language logical relations which express when two terms in those languages are semantically equivalent (\Cref{sec:langs}).
Next, we present fully-abstract compilation and describe our approximate backtranslation proof technique in detail (\Cref{sec:fac}).
Then we define the three compilers (from \stlcf to \stlcim, from \stlcf to \stlcem and from \stlcim to \stlcem) and prove that they are fully abstract using three approximate backtranslations (\Cref{sec:comp}).
These compilers and their fully-abstract compilation proofs are all formalised in Coq, so we also present the most useful insights into this formalisation (\Cref{sec:coq}).
After a discussion of the presented results (\Cref{sec:disc}), we present related work (\Cref{sec:rw}) and conclude (\Cref{sec:conc}).

For the sake of simplicity we omit some elements of the formalisation such as auxiliary lemmas and proofs.
The Coq mechanisation of this work is available at:
\begin{center}
	\url{https://github.com/dominiquedevriese/fixismu-coq}
\end{center}

\subsection{Comparison with the Previous Version}\label{sec:comparison}
This work extends the work of \citet{isoequi-popl} presented at POPL'21 in the following way:
\begin{itemize}
  \item We fix a bug in the original proof that broke \Cref{thm:inj-ext-sem-pres}.
  The bug is addressed by making the approximate logical relation rely on an additional bound on the size of terms encountered during reductions, as mentioned before.
  This, in turn changes the observation relation of the logical relation, i.e., the part that tells when two terms are related.
  Previously (as well as in related work~\citep{Devriese:2016:FCA:2837614.2837618}), a term \src{t} was related to another one \trg{t} at level \com{n} if termination of \src{t} in at most \com{n} steps implied termination of \trg{t} in some steps (and vice versa).
  Here, we introduce a new notion of bounded termination (called \newterm) that a term fulfils for some steps \com{m} if the term terminates in at most \com{m} steps and (roughly) terms encountered during this reduction have at most size \com{m} (in terms of the depth of the AST of the term).
  We rely on \newterm in the observation relation and state that a term \src{t} is related to another \trg{t} at level \com{j} if \newterm of \src{t} in at most \com{j} steps implies termination of \trg{t} in some steps (and vice versa).
  Intuitively, in the previous formulation the step index \com{n} imposes a bound on the amount of steps required for termination.
  Here instead, the step index \com{j} imposes a bound both on the steps required for termination and on the size of terms encountered during such termination.

  We explain in more detail the problem with the old formulation and how this new idea lets \Cref{thm:inj-ext-sem-pres} go through in \Cref{sec:injext}, where we discuss \Cref{ex:need-newterm}.

  \item We mechanise the three fully-abstract compilation proofs in the Coq proof assistant and report on the formalisation in \Cref{sec:coq}.
\end{itemize}

\section{Languages and Cross-Language Logical Relations}\label{sec:langs}
This section presents the simply-typed lambda calculus (\com{\lambda}) and its extensions with a typed fixpoint operator (\stlcf), with iso-recursive types (\stlcim) and with coinductive equi-recursive types (\stlcem).
We first define the syntax (\Cref{sec:syn}), then the static semantics (\Cref{sec:typ}) and then the operational semantics of these languages (\Cref{sec:sem}).
Finally, this section presents the cross-language logical relations used to reason about the expressiveness of terms in different languages (\Cref{sec:logrel-main}).
Note that these logical relations are partial, the key addition needed to attain fully-abstract compilation is presented in \Cref{sec:rel-bt} only after said addition is justified.

\subsection{Syntax}\label{sec:syn}
All languages include standard terms (\com{t}) and values (\com{v}) from the simply-typed lambda calculus: lambda abstractions, applications, pairs, projections, tagged unions, case destructors, booleans, branching, unit and sequencing.
Additionally, \stlcf has a \src{fix} operator providing general recursion, while \stlcim has \iso{\fold{}} and \iso{\unfold{}} annotations; \stlcem requires no additional syntactic construct.

Regarding types, both \stlcim and \stlcem add recursive types according to the same syntax.
In \stlcim and \stlcem, recursive types are syntactically constrained to be \emph{contractive}.
Note however that for simplicity of presentation we will indicate a type as \com{\tau} and simply report the contractiveness constraints when meaningful.
A recursive type \mat\ is contractive if, the use of the recursion variable \com{\alpha} in \com{\tau} occurs under a type constructor such as \com{\to} or \com{\times}~\citep{10.1145/800017.800528}.
Non-contractive types (e.g., \com{\matgen{\alpha}{\alpha}}) are not inhabited by any value, so it is reasonable to elide them. 
Moreover, they do not have an infinite unfolding and (without restrictions on the type equality relation) can be proven equivalent to any other type~\citep{10.1007/978-3-642-39212-2_28}, which is undesirable.
All languages have evaluation contexts (\evalctxc), which indicate where the next reduction will happen, and program contexts (\ctxc), which are larger programs to link terms with.
\begin{gather*}\small
  \begin{aligned}
  \com{\tau},\com{\sigma} \bnfdef 
    &\ 
    \Unitc \mid \Boolc \mid \com{\tau^s\to\tau^s} \mid \com{\tau^s\times\tau^s} \mid \com{\tau^s\uplus\tau^s} 
    \mid \iso{\matt} %
    \mid \equi{\mat} %
  \\
  \com{\tau^s} \bnfdef 
    &\
    \iso{\alpt} \mid \equi{\alpha} \mid \com{\tau}
  \\
  \com{\Gamma} \bnfdef
    &\ 
    \come \mid \com{\Gamma},\com{x}:\com{\tau}
  \\
  \com{v} \bnfdef
    &\ 
    \unitc \mid \truec \mid \falsec \mid \com{\lam{x:\tau}{t}} \mid \com{\pair{v,v}} \mid \com{\inl{v}} \mid\com{\inr{v}}
    \mid \iso{\fold{\matt}~v}
  \\
  \com{t} \bnfdef
    &\ 
    \unitc \mid \truec \mid \falsec \mid \com{\lam{ x:\tau}{ t}} \mid \com{x} \mid \com{t~t} \mid \com{\projone{t}} \mid \com{\projtwo{t}} \mid \com{\pair{t,t}} 
  \\
  \mid
    &\
    \com{\caseof{t}{t}{t}} 
    \mid
    \com{\inl{t}} \mid \com{\inr{t}} \mid
    \com{\ifte{t}{t}{t}} \mid \com{t;t} 
  \\
  \mid
    &\
    \src{\fix{\tau\to\tau} t} \mid \iso{\fold{\matt}~t} \mid \iso{\unfold{\matt}~t}
  \\
  \evalctxc \bnfdef
    &\ 
    \com{\hole{\cdot}} \mid \evalctxc~\com{t} \mid \com{v}~\evalctxc \mid \com{\projone{\evalctxc}} \mid \com{\projtwo{\evalctxc}} \mid \com{\pair{\evalctxc,t}} \mid \com{\pair{v,\evalctxc}}
    \mid \com{\caseof{\evalctxc}{t_1}{t_2}}
  \\
  \mid
    &\ 
    \com{\inl{\evalctxc}} \mid \com{\inr{\evalctxc}} \mid \com{\evalctxc;t} \mid \com{\ifte{\evalctxc}{t}{t}}
    \mid \src{\fix{\tau\to\tau} \evalctx} 
    \mid \iso{\fold{\matt}~\evalctxt} \mid \iso{\unfold{\matt}~\evalctxt}
  \\
  \ctxc \bnfdef
    &\ 
    \com{\hole{\cdot}} \mid \com{\lam{x:\tau}{\ctxc}} \mid \com{\ctxc~t} \mid \com{t~\ctxc} \mid \com{\projone{\ctxc}} \mid \com{\projtwo{\ctxc}} \mid \com{\pair{\ctxc,t}} \mid \com{\pair{t,\ctxc}} \mid 
    \com{\caseof{\ctxc}{t}{t}}
  \\
  \mid
    &\ 
    \com{\caseof{t}{\ctxc}{t}} 
    \mid \com{\caseof{t}{t}{\ctxc}}
  \\
  \mid
    &\
    \com{\inl{\ctxc}} \mid \com{\inr{\ctxc}} \mid \com{\ctxc;t} \mid \com{t;\ctxc}
  \mid
    \com{\ifte{\ctxc}{t}{t}} \mid \com{\ifte{t}{\ctxc}{t}} 
  \\
  \mid
    &\
    \com{\ifte{t}{t}{\ctxc}} \mid \src{\fix{\tau\to\tau} \ctx}
  \mid
  \iso{\fold{\matt}~\ctx} \mid \iso{\unfold{\matt}~\ctx}
  \end{aligned}
\end{gather*}

As mentioned in \Cref{sec:intro}, we need a measure to define \newterm as the new logical relation requires.
The measure we rely on is the size of a term \com{t}, which we calculate via function $\size{\cdot} : \com{t} \to n\in\mb{N}$.
Intuitively, the size of a measure counts the number of nodes in the term's AST, ignoring the bodies of lambdas.
As a result, apart from bodies of lambdas, any sub-term \com{t'} has size smaller than the super-term \com{t} that contains \com{t'}.
\begin{gather*}
\begin{aligned}
  \size{ \unitc } 
    &=
      1
      &
      \size{ \truec }
        &=
        1
          &
          \size{ \falsec}
            &=
            1
\end{aligned}
\\
\begin{aligned}
  \size{ \com{x}}
    &=
    1
      &
      \size{ \lam{x:\tau}{t}}
        &= 1
  \\
          \size{t~t'}
            &=
            \size{t} + \size{t'} + 1
  &
    \size{\projone{t}}
      &=
      \size{t}+1
  \\
        \size{\projtwo{t}}
          &=
          \size{t}+1
            &
            \size{\pair{t,t'}}
              &=
              \size{t}+\size{t'}\!+\!1
  \\
  \size{\inl{t}}
    &=
    \size{t}+1
      &
      \size{\inr{t}}
        &=
        \size{t}+1
  \\
            \size{{t;t'}}
              &=
              \size{t}+\size{t'} + 1
              &
  \size{\src{\fix{\tau\to\tau} t}}
    &=
    \size{\src{t}}+1
  \\
      \size{\iso{\fold{\matt}~t}}
        &=
        \size{\iso{t}}+1
        &
        \hspace*{-0.5em}\size{\iso{\unfold{\matt}~t}}
          &=
          \size{\iso{t}}+1
\end{aligned}
\\
\begin{aligned}
  \size{\com{\caseof{t}{t'}{t''}}}
    &=\
    \size{t} + \size{t'} + \size{t''} + 1
  \\
  \size{\ifte{t}{t'}{t''}}
    &=\
    \size{t} + \size{t'} + \size{t''} + 1
\end{aligned}
\end{gather*}

\subsection{Static Semantics}\label{sec:typ}
This section presents the (fairly standard) static semantics of our languages, we delay discussing alternative formulations of equi-recursive types to \Cref{sec:rw}.
The static semantics for terms follows the canonical judgement $\com{\Gamma}\vdash\com{t}:\com{\tau}$, which attributes type \com{\tau} to term \com{t} under environment \com{\Gamma} and occasionally relies on function $\ftv{\com{\tau}}$, which returns the free type variables of \com{\tau}.
The only difference in the typing rules regards \iso{\fold{}}/\iso{\unfold{}} terms (\Cref{tr:t-fold,tr:t-unfold}) and the introduction of the type equality ($\tyeq$ in \Cref{tr:t-eq}).

\pagebreak
\begin{center}
\mytoprule{\com{\Gamma}\vdash \com{t}:\com{\tau}}

  \typerule{Type-var}{
    \com{x:\tau}\in\com{\Gamma}
  }{
    \com{\Gamma}\vdash \com{x}:\com{\tau}
  }{t-var}
  \typerule{Type-unit}{}{
    \com{\Gamma}\vdash \unitc :\com{\Unitc}
  }{t-unit}
  \typerule{Type-true}{}{
    \com{\Gamma}\vdash \truec :\com{\Boolc}
  }{t-true}
  \typerule{Type-false}{}{
    \com{\Gamma}\vdash \falsec :\com{\Boolc}
  }{t-false}

  \medskip

  \typerule{Type-p1}{
    \com{\Gamma}\vdash \com{t}:\com{\tau\times\tau'}
  }{
    \com{\Gamma}\vdash \com{\projone{t}}:\com{\tau}
  }{t-p1}
  \typerule{Type-p2}{
    \com{\Gamma}\vdash \com{t}:\com{\tau'\times\tau}
  }{
    \com{\Gamma}\vdash \com{\projtwo{t}}:\com{\tau}
  }{t-p2}
  \typerule{Type-inl}{
    \com{\Gamma}\vdash \com{t}:\com{\tau}
  }{
    \com{\Gamma}\vdash \com{\inl{t}}:\com{\tau\uplus\tau'}
  }{t-inl}
  \typerule{Type-inr}{
    \com{\Gamma}\vdash \com{t}:\com{\tau'}
  }{
    \com{\Gamma}\vdash \com{\inr{t}}:\com{\tau\uplus\tau'}
  }{t-inr}

\medskip

  \typerule{Type-case}{
    \com{\Gamma}\vdash \com{t}:\com{\tau'\uplus\tau''}
    \\
    \com{\Gamma,x_1:\tau'}\vdash \com{t'}:\com{\tau}
    &
    \com{\Gamma,x_2:\tau''}\vdash \com{t''}:\com{\tau}
  }{
    \com{\Gamma}\vdash \com{\caseof{t}{t'}{t''}}:\com{\tau}
  }{t-case}
  \typerule{Type-if}{
    \com{\Gamma}\vdash \com{t}:\com{\Boolc}
    \\
    \com{\Gamma}\vdash \com{t'}:\com{\tau}
    &
    \com{\Gamma}\vdash \com{t''}:\com{\tau}
  }{
    \com{\Gamma}\vdash \com{\ifte{t}{t'}{t''}}:\com{\tau}
  }{t-if}

  \medskip

  \typerule{Type-seq}{
    \com{\Gamma}\vdash \com{t}:\com{\Unitc}
    \\
    \com{\Gamma}\vdash \com{t'}:\com{\tau}
  }{
    \com{\Gamma}\vdash \com{t;t'}:\com{\tau}
  }{t-seq}
  \typerule{Type-lam}{
    \com{\Gamma,x:\tau}\vdash \com{t}:\com{\tau'}
    \\
    \ftv{\com{\tau}} = \come
  }{
    \com{\Gamma}\vdash \com{\lam{x:\tau}{t}}:\com{\tau\to\tau'}
  }{t-lam}
  \typerule{Type-app}{
    \com{\Gamma}\vdash \com{t}:\com{\tau'\to\tau}
    \\
    \com{\Gamma}\vdash \com{t'}:\com{\tau'}
  }{
    \com{\Gamma}\vdash \com{t~t'}:\com{\tau}
  }{t-app}

  \medskip

  \typerule{Type-pair}{
    \com{\Gamma}\vdash \com{t}:\com{\tau}
    \\
    \com{\Gamma}\vdash \com{t'}:\com{\tau'}
  }{
    \com{\Gamma}\vdash \com{\pair{t,t'}}:\com{\tau\times\tau'}
  }{t-pair}
  \typerule{\stlcf-Type-fix}{
    \src{\Gamma}\vdash \src{t} : \src{(\tau_1 \to \tau_2) \to \tau_1 \to \tau_2}
  }{
    \src{\Gamma}\vdash\src{\fix{\tau_1 \to \tau_2}} \src{t} : \src{\tau_1 \to \tau_2}
  }{stlcf-fix}

  \medskip
  
  \typerule{\stlcim-Type-fold}{
    \isob{\Gamma}\vdash\iso{t}:\iso{\tat\subt{\matt}{\alpt}}
  }{
    \isob{\Gamma}\vdash\iso{\fold{\matt}~t}:\iso{\matt} 
  }{t-fold}
  \typerule{\stlcim-Type-unfold}{
    \isob{\Gamma}\vdash\iso{t}:\iso{\matt}  
  }{
    \isob{\Gamma}\vdash\iso{\unfold{\matt}~t}:\iso{\tat\subt{\matt}{\alpt}}
  }{t-unfold}
  \typerule{\stlcem-Type-eq}{
    \equi{\Gamma}\vdash\equi{t}:\equi{\tau}
        &
        \tyeqbin{\tau}{\sigma}
  }{
    \equi{\Gamma}\vdash\equi{t}:\equi{\sigma}
  }{t-eq}
  
  \smallskip

\botrule
\end{center}

Program contexts have an important role in fully-abstract compilation.
They follow the usual typing judgement ($ \com{\ctxc} \vdash \com{\Gamma},\com{\tau} \to \com{\Gamma'},\com{\tau'} $), i.e., program context \ctxc is well typed with a hole of type \com{\tau} that use free variables in \com{\Gamma}, and overall \ctxc returns a term of type \com{\tau'} and uses variables in \com{\Gamma'}.
\begin{center}
\mytoprule{\com{\ctxc} \vdash \com{\Gamma},\com{\tau} \to \com{\Gamma'},\com{\tau'}}

  \typerule{Type-Ctx-Hole}{}{
  	\vdash \cdot : \com{\Gamma},\com{\tau} \to \com{\Gamma},\com{\tau}
  }{stlcf-typectx-hole}
  \typerule{Type-Ctx-Lam}{
  	\vdash \ctxc : \com{\Gamma''},\com{\tau''} \to (\com{\Gamma},\com{x}:\com{\tau'}), \com{\tau}
  }{
  	\vdash \com{\lam{x:\tau'}{\ctxc}} : \com{\Gamma''},\com{\tau''} \to \com{\Gamma},\com{\tau' \to \tau}
  }{stlcf-typectx-lam}

  \medskip

  \typerule{Type-Ctx-Pair1}{
  	\vdash \ctxc : \com{\Gamma'},\com{\tau'} \to \com{\Gamma}, \com{\tau_1} 
  	\\
  	\com{\Gamma}\vdash \com{t_2} : \com{\tau_2}
  }{
  	\vdash \com{\pair{\ctxc,t_2}} : \com{\Gamma'},\com{\tau'}\to\com{\Gamma},\com{\tau_1\times\tau_2}
  }{stlcf-typectx-pair1}
  \typerule{Type-Ctx-Pair2}{
  	\com{\Gamma}\vdash \com{t_1} :\com{\tau_1} 
  	\\
  	\vdash \ctxc : \com{\Gamma'},\com{\tau'} \to \com{\Gamma}, \com{\tau_2}
  }{
  	\vdash \com{\pair{t_1,\ctxc}} : \com{\Gamma'},\com{\tau'}\to\com{\Gamma},\com{\tau_1\times\tau_2}
  }{stlcf-typectx-pair2}

  \medskip

  \typerule{Type-Ctx-Inl}{
  	\vdash \ctxc : \com{\Gamma''},\com{\tau''}\to \com{\Gamma} , \com{\tau}
  }{
  	\vdash \com{\inl{\ctxc}} : \com{\Gamma''},\com{\tau''}\to \com{\Gamma},\com{\tau\uplus\tau'}
  }{stlcf-typectx-inl}
  \typerule{Type-Ctx-Inr}{
  	\vdash \ctxc : \com{\Gamma''},\com{\tau''}\to \com{\Gamma} , \com{\tau}'
  }{
  	\vdash \com{\inr{\ctxc}} : \com{\Gamma''},\com{\tau''}\to \com{\Gamma},\com{\tau\uplus\tau'}
  }{stlcf-typectx-inr}

  \medskip
  
  \typerule{Type-Ctx-App1}{
  	\vdash \ctxc : \com{\Gamma'},\com{\tau'} \to \com{\Gamma},\com{\tau_1\to\tau_2} 
  	\\
  	\com{\Gamma} \vdash \com{t_2} : \com{\tau_1}
  }{
  	\vdash \com{\ctxc\ t_2} : \com{\Gamma'},\com{\tau'} \to \com{\Gamma}, \com{\tau_2}
  }{stlcf-typectx-app1}
  \typerule{Type-Ctx-App2}{
  	\com{\Gamma}\vdash \com{t_1} : \com{\tau_1\to\tau_2}
  	\\
  	\vdash \ctxc : \com{\Gamma'},\com{\tau'} \to \com{\Gamma}, \com{\tau_1}
  }{
  	\vdash \com{t_1\ \ctxc} : \com{\Gamma'},\com{\tau'} \to \com{\Gamma}, \com{\tau_2}
  }{stlcf-typectx-app2}

  \medskip
  
  \typerule{Type-Ctx-Proj1}{
  	\vdash \ctxc : \com{\Gamma'},\com{\tau'} \to \com{\Gamma}, \com{\tau_1\uplus\tau_2}
  }{
  	\vdash \com{\projone{\ctxc}} : \com{\Gamma'},\com{\tau'} \to \com{\Gamma}, \com{\tau_1}
  }{stlcf-typectx-proj1}
  \typerule{Type-Ctx-Proj2}{
  	\vdash \ctxc : \com{\Gamma'},\com{\tau'} \to \com{\Gamma}, \com{\tau_1\uplus\tau_2}
  }{
  	\vdash \com{\projtwo{\ctxc}} : \com{\Gamma'},\com{\tau'} \to \com{\Gamma}, \com{\tau_2}
  }{stlcf-typectx-proj2}

  \medskip
  
  \typerule{Type-Ctx-Case1}{
  	\vdash\ctxc:\com{\Gamma'},\com{\tau'} \to\com{\Gamma},\com{\tau_1\uplus\tau_2} 
  	\\
  	\com{\Gamma},\com{x_1}:\com{\tau_1} \vdash \com{t_1} :\com{\tau_3} 
  	&
  	\com{\Gamma},\com{x_2} :\com{\tau_2} \vdash \com{t_2} : \com{\tau_3}
  }{
  	\vdash \com{\caseof{\ctxc}{t_1}{t_2}} : \com{\Gamma'},\com{\tau'} \to \com{\Gamma}, \com{\tau_3}
  }{stlcf-typectx-case1}

  \medskip
  
  \typerule{Type-Ctx-Case2}{
  	\com{\Gamma}\vdash \com{t}:\com{\tau_1\uplus\tau_2} 
  	&
  	\vdash \ctxc :\com{\Gamma'},\com{\tau'} \to (\com{\Gamma},\com{x_1}:\com{\tau_1}),\com{\tau_3} 
  	&
  	\com{\Gamma},\com{x_2} :\com{\tau_2} \vdash \com{t_2} : \com{\tau_3}
  }{
  	\vdash \com{\caseof{t}{\ctxc}{t_2}} : \com{\Gamma'},\com{\tau'} \to \com{\Gamma}, \com{\tau_3}
  }{stlcf-typectx-case2}

  \medskip
  
  \typerule{Type-Ctx-Case3}{
  	\com{\Gamma}\vdash\com{t}:\com{\tau_1\uplus\tau_2} 
  	&
  	\com{\Gamma},\com{x_1}:\com{\tau_1} \vdash \com{t_1} :\com{\tau_3} 
  	&
  	\vdash \ctxc : \com{\Gamma'},\com{\tau'} \to (\com{\Gamma},\com{x_2} :\com{\tau_2}),\com{\tau_3}
  }{
  	\vdash \com{\caseof{t}{t_1}{\ctxc}} : \com{\Gamma'},\com{\tau'} \to \com{\Gamma}, \com{\tau_3}
  }{stlcf-typectx-case3}

  \medskip
  
  \typerule{Type-Ctx-If1}{
  	\vdash \ctxc : \com{\Gamma},\com{\tau}\to\com{\Gamma'},\Boolc
  	\\
  	\com{\Gamma'} \vdash \com{t_1} : \com{\tau'}				
  	&
  	\com{\Gamma'} \vdash \com{t_2} : \com{\tau'}
  }{
  	\vdash \com{\ifte{\ctxc}{t_1}{t_2}} : \com{\Gamma},\com{\tau}\to\com{\Gamma'},\com{\tau'}
  }{stlcf-typectx-if1}
  \typerule{Type-Ctx-If2}{
  	\com{\Gamma}\vdash\com{t}:\Boolc								
  	\\
  	\vdash \ctxc : \com{\Gamma},\com{\tau} \to\com{\Gamma'},\com{\tau'}
  	&
  	\com{\Gamma} \vdash \com{t_2} : \com{\tau'}
  }{
  	\vdash \com{\ifte{t}{\ctxc}{t_2}} : \com{\Gamma},\com{\tau}\to\com{\Gamma'},\com{\tau'}
  }{stlcf-typectx-if2}

  \medskip
  
  \typerule{Type-Ctx-If3}{
  	\com{\Gamma}\vdash\com{t}:\Boolc
  	&
  	\com{\Gamma} \vdash \com{t_1} : \com{\tau'}							
  	&
  	\vdash \ctxc : \com{\Gamma},\com{\tau} \to\com{\Gamma'},\com{\tau'}
  }{
  	\vdash \com{\ifte{t}{t_1}{\ctxc}} : \com{\Gamma},\com{\tau}\to\com{\Gamma'},\com{\tau'}
  }{stlcf-typectx-if3}
  \typerule{Type-Ctx-Seq1}{
  	\ctxc : \com{\Gamma},\com{\tau}\to\com{\Gamma'},\Unitc
  	\\
  	\com{\Gamma'}\vdash\com{t}:\com{\tau''}
  }{
  	\vdash \com{\ctxc;t} : \com{\Gamma},\com{\tau}\to\com{\Gamma'},\com{\tau''}
  }{stlcf-typectx-seq1}

  \medskip
  
  \typerule{Type-Ctx-Seq2}{
  	\com{\Gamma}\vdash\com{t}:\Unitc
  	\\
  	\vdash \ctxc : \com{\Gamma},\com{\tau}\to\com{\Gamma'},\com{\tau'}
  }{
  	\vdash \com{t;\ctxc} : \com{\Gamma},\com{\tau}\to\com{\Gamma'},\com{\tau'}
  }{stlcf-typectx-seq2}
  \typerule{\stlcf-Type-Ctx-Fix}{
  	\ctxs : \src{\Gamma'},\src{\tau'}\to\src{\Gamma},\src{\tau\to\tau}
  }{
  	\vdash \src{\fix{\tau\to\tau} \ctxs} : \src{\Gamma'},\src{\tau'}\to\src{\Gamma},\src{\tau}
  }{stlcf-typectx-fix}  

  \medskip
  
  \typerule{\stlcm-Type-Ctx-Fold}{
  	\vdash \ctxt : \trg{\Gamma'},\trgb{\tau'}\to\trg{\Gamma},\trgb{\tau\subt{\mat}{\alpha}}
  }{
  	\vdash \trg{\fold{\matt}~ \ctxt} : \trg{\Gamma'},\trgb{\tau'}\to\trg{\Gamma},\trgb{\mat}
  }{stlcm-typectx-fold}
  \typerule{\stlcm-Type-Ctx-Unfold}{
  	\vdash \ctxt : \trg{\Gamma'},\trgb{\tau'}\to\trg{\Gamma},\trgb{\mat}
  }{
  	\vdash \trg{\unfold{\matt}\ctxt} : \trg{\Gamma'},\trgb{\tau'}\to\trg{\Gamma},\trgb{\tau\subt{\mat}{\alpha}}
  }{stlcm-typectx-unfold}

  \medskip
  
  \typerule{\stlcem-Type-Eq}{
  	\vdash \ctxo : \oth{\Gamma'},\othb{\tau'}\to\oth{\Gamma},\othb{\tau}
    &
    \tyeqbin{\tau}{\sigma}
  }{
  	\vdash \ctxo : \oth{\Gamma'},\othb{\tau'}\to\oth{\Gamma},\othb{\sigma}
  }{stlce-typectx-eq}
  \smallskip

\botrule
\end{center}

We use the same coinductive type equality relation of \citet{Cai:2016:SFE:2914770.2837660}, with a cosmetic difference only.
Two types are equal if they are the same base type \oth{\iota} or variable (\Cref{tr:tyeq-prim,tr:tyeq-var}).
If the types are composed of two types, the connectors must be the same and each sub-type must be equivalent (\Cref{tr:tyeq-bin}).
If the left type starts with a \oth{\mu} (or if that does not but the right one does), then we unfold the type for checking the equality (\Cref{tr:tyeq-mu-l,tr:tyeq-mu-r}).
Note that these last two rules are defined in an asymmetric fashion to make equality derivation deterministic.
Finally, we make explicit the rules for reflexivity, symmetry and transitivity (\Cref{tr:tyeq-refl}, \Cref{tr:tyeq-symm,tr:tyeq-trans}) whose derivations we have proved from the other rules.
\begin{center}
\mytoprule{ \tyeqbin{\tau}{\tau'} }\\
    \cotyperule{$\tyeq$-prim}{
      \equi{\iota} = \equi{\Unito} ~\vee~
      \equi{\iota} = \equi{\Boolo}
    }{
        \tyeqbin{\iota}{\iota}
    }{tyeq-prim}
    \cotyperule{$\tyeq$-var}{
    }{
        \tyeqbin{\alpha}{\alpha}
    }{tyeq-var}
    \cotyperule{$\tyeq$-bin}{
        \equi{\star} \in \{\equi{\to}, \equi{\times}, \equi{\uplus}\}
      &&
        \tyeqbin{\tau_1}{\sigma_1}
        &&
        \tyeqbin{\tau_2}{\sigma_2}
    }{
      \tyeqbin{\tau_1 \star \tau_2}{\sigma_1 \star \sigma_2}
    }{tyeq-bin}
    \cotyperule{$\tyeq$-${\mu_l}$}{
        \tyeqbin{\tau\subo{\mat}{\alpha}}{\sigma}
    }{
      \tyeqbin{\mat}{\sigma}
    }{tyeq-mu-l}
    \cotyperule{$\tyeq$-${\mu_r}$}{
        \tyeqbin{\tau}{\sigma\subo{\matgen{\alpha}{\sigma}}{\alpha}}
    }{
      \tyeqbin{\tau}{\matgen{\alpha}{\sigma}}
    }{tyeq-mu-r}
    \cotyperule{$\tyeq$-refl}{
    }{
      \tyeqbin{\tau}{\tau}
    }{tyeq-refl}
    \cotyperule{$\tyeq$-symm}{
      \tyeqbin{\sigma}{\tau}
    }{
      \tyeqbin{\tau}{\sigma}
    }{tyeq-symm}
    \cotyperule{$\tyeq$-trans}{
      \tyeqbin{\tau}{\sigma}
      &&
      \tyeqbin{\sigma}{\tau'}
    }{
      \tyeqbin{\tau}{\tau'}
    }{tyeq-trans}

\botrule
\end{center}

To prove results about this equality relation, we will often induct on the ``leading-mu-count'' (\mtt{lmc}) measure.
Intuitively, that measure counts the amount of \equi{\mu}s that a \stlcem type has before a different connector is found.
This is almost the same as the number of times a type can be unfolded before it is no longer recursive at the top level (e.g. $\lMuCount{\Unito} = 0$, $\lMuCount{\equi{\matgen{\alpha}{\alpha \uplus \Unito}}} = 1$).
\begin{align*}
    \lMuCount{\equi{\tau}} &\isdef 
                             \begin{cases}
                               \lMuCount{\equi{\tau'}} + 1 
                               & 
                               \equi{\tau} = \equi{\matgen{\alpha}{\tau'}} 
                               \\
                               0 
                               & 
                               \text{otherwise}
                             \end{cases}
  \end{align*}%
Non-contractive types such as \equi{\matgen{\alpha}{\alpha}}, however, create problems here, for they always unfold into another top level recursive type.
This motivates our restriction to contractive types only: a contractive type \equi{\tau} can be unfolded exactly \lMuCount{\tau} times.

\subsection{Dynamic Semantics}\label{sec:sem}

All our languages are given a small-step, contextual, call-by-value, operational semantics.
We highlight primitive reductions as $\redp$ and non-primitive ones as $\red$.
We indicate the capture-avoiding substitution of variable (or type variable) \com{x} in \com{t} with value (or type) \com{v} as $\com{t}\subc{v}{x}$.
Note that since \stlcem has no peculiar syntactic construct, it also has no specific reduction rule.
\begin{center}
\mytoprule{\com{t \red t'} \quad \text{ and } \quad \com{t \redp t'}}\\%
  \typerule{Eval-ctx}{
    \com{t \redp t'}
  }{
    \com{\evalctxhc{t}\red\evalctxhc{t'}}
  }{e-ctx}
  \typerule{Eval-beta}{}{
    \com{(\lam{x:\tau}{t})~v}\redp \com{t}\subc{v}{x}
  }{e-beta}
  \typerule{Eval-pi}{
    \com{i}\in\com{1..2}
  }{
    \com{\proji{\pair{v_1,v_2}}}\redp\com{v_i}
  }{e-p1}

  \medskip
  
  \typerule{Eval-seq}{
  }{
    \com{\unitc;t \redp t}
  }{e-seq}
  \typerule{Eval-inl}{}{
    \com{
      \casefoldedc{\inl{v}}{t}{t'}
    }\redp\com{t}\subc{v}{x_1}
  }{e-inl}

  \medskip
  
  \typerule{Eval-inr}{}{
    \com{\casefoldedc{\inr{v}}{t}{t'}}\redp\com{t'}\subc{v}{x_2}
  }{e-inr}
  \typerule{Eval-if}{
    \com{v}=\com{\truec}\vee\com{\falsec}
  }{
    \com{\ifte{v}{t_{\truec}}{t_{\falsec}} \redp \com{t_v}}
  }{e-if}

  \medskip
  
  \typerule{\stlcf-Eval-fix}{
  }{
    \src{\fix{\tau \to \tau}\ (\lam{x:\tau}{t}) \redp t\ \subs{\fix{\tau\to\tau} (\lam{x:\tau}{t}) }{x}}
  }{stlcf-eval-fix}
  \typerule{\stlcim-Eval-fold}{}{
    \trg{\unfold{\matt}~(\fold{\matt}~v) \redp v}
  }{e-i-fold}

  \smallskip

\botrule
\end{center}

\subsection{Notions of Termination}\label{sec:termination}
For technical reasons, we need to define two notions of termination for our languages.
To define contextual equivalence (which is required for fully-abstract compilation), we rely on the canonical definition of termination, which tells that a term eventually reduces to a value in some number of steps.
\begin{definition}[Termination]\label{def:term}
	\[
		\com{t}\term \isdef \exists n\in\mb{N}, \com{v} \ldotp \com{t \btermc{n} v}
	\]
\end{definition}

We rely on the auxiliary judgement for bounded termination in order to say that a term \com{t} reduces to a value \com{v} in $n$ steps.
\begin{center}
  \typerule{Bounded termination-value}{
  }{
    \com{v}\btermc{0}\com{v}
  }{bter-v}
  \typerule{Bounded termination-term}{
    \com{t}\red\com{t'}
    &
    \com{t'}\btermc{n}\com{v}
  }{
    \com{t}\btermc{n+1}\com{v}
  }{bter-t}
\end{center}

As mentioned in \Cref{sec:intro}, to make the logical relation more precise, we need another notion of bounded termination that not only bounds the number of steps needed for reaching a value but also the size of intermediate terms encountered during these steps. 
\begin{center}
  \typerule{\newterm-value}{
    \size{v}\leq n
  }{
    \com{v}\shrinkc{n}\com{v}
  }{shr-v}
  \typerule{\newterm-term}{
    \com{t}\red\com{t'}
    &
    \com{t'}\shrinkc{n}\com{v}
    &
    \size{t}\leq n
  }{
    \com{t}\shrinkc{n+1}\com{v}
  }{shr-t}
\end{center}
It is worth noting that this definition does not apply the same size bound to all terms encountered during execution, but the bound decreases as execution progresses.
This approach has minor technical benefits in the definition, but we think a definition with a single bound on all terms would work as well. 

The two termination notions are related by \Cref{thm:term-def-rel} below.
For any term \com{t} that terminates there exists a \com{n} such that \newterm holds for \com{t} in \com{n} steps.
Conversely, if \newterm holds for a term then it also terminates.
\begin{theorem}[Relation between Termination and \NewTerm]\label{thm:term-def-rel}
	\begin{align*}
		\text{ if } 
			&
			\com{t}\term
			\text{ then }
			\exists \com{n}\in\mb{N}, \com{v} \ldotp \com{t \shrinkc{n} v}
		\\
		\text{ if }
			&
			\com{t}\shrinkc{\_}
			\text{ then }
			\com{t}\term
	\end{align*}
\end{theorem}

Although this theorem is quite easy to prove, it does capture a non-trivial property of the programming language, namely the fact that it only contains finite values.
If we would define a variant of the language with infinite values (e.g.\ if we had interpreted $\mu$ as producing a coinductive fixpoint rather than an inductive one, perhaps with a call-by-need semantics), then the property would no longer hold.

\subsection{Logical Relations Between Our Languages}\label{sec:logrel-main}
As mentioned in \Cref{sec:intro}, we need cross-language relations that indicate when related source and target terms approximate each other.
Intuitively, one such relation is needed by each one of the compilers we define later.
Thus, we need to define three logical relations: 
\begin{enumerate}
  \item[A] one between \stlcf and \stlcim, which we dub \lrfi;
  \item[B] one between \stlcim and \stlcem, which we dub \lrie;
  \item[C] one between \stlcf and \stlcem, which we dub \lrfe.
\end{enumerate}
These relations are all indexed by a step and then by the source type, so logical relations (A) and (C) look the same. For brevity we present only one of them.
Additionally, given that \stlcim has the same types of \stlcf plus recursive types, we only show that case for logical relation (B).
Ours are Kripke, step-indexed logical relations that are based on those of \citet{popl-journal,Hur:2011:KLR:1926385.1926402}.
The step-indexing is not inherently needed for relations (A) and (C), which could be defined just by induction on \stlcf types (since they do not include recursive types).
However, all of our relations are step-indexed anyway because the steps also determine for how many steps one term should approximate the other and this detail is key for the backtranslation proof technique.

Before presenting the details, note that the relations we show here are \emph{not} complete.
Specifically they only talk about the terms needed to conclude reflection of fully-abstract compilation but not preservation (admittedly, the most interesting part).
Completing the logical relations relies on technical insights regarding the backtranslations, so we do this later in \Cref{sec:rel-bt}.
The goal of this section is to provide an understanding of what it means for two terms to approximate each other.

\begin{figure}[!htb]
  \begin{gather*}
    \begin{aligned}
    \W \isdef
      &\
      n \in\mb{N}
    & 
    \stepsfun{n} =
      &\
      n
    &
    \laterfun{0}=
      &\ 
      0
    &
    \laterfun{n+1} =
      &\
      n
    \end{aligned}
    \\
    \begin{aligned}
    \W\futw\W' =
      &\
      \stepsfun{\W}\leq\stepsfun{\W'}
    &
    \W\strfutw\W' =
      &\
      \stepsfun{\W}<\stepsfun{\W'}
    \end{aligned}
    \\
    \begin{aligned}
    \obsfun{\W}{\underlogrel}\isdef
      &\
      \myset{(\src{t},\trg{t})}{ \text{if } \stepsfun{\W}>n \text{ and } \src{t\shrinks{n}v} \text{ then } \exists \trg{k}, \trg{v}.~ \trg{t\btermt{k}v} }
    \\
    \obsfun{\W}{\overlogrel}\isdef
      &\
      \myset{(\src{t},\trg{t})}{ \text{if } \stepsfun{\W}>n \text{ and } \trg{t\shrinkt{n}v} \text{ then } \exists \src{k}, \src{v}.~ \src{t\bterms{k}v} }
    \\
    \obsfun{\W}{\bothlogrel}\isdef
      &\
      \obsfun{\W}{\underlogrel}\cap\obsfun{\W}{\overlogrel}
    \end{aligned}
  \end{gather*}
  \caption{
    Worlds, observations and related technicalities. 
    These are typeset for the relation between \stlcf and \stlcim but the other ones do not change.
    \label{fig:logrels-worlds}
  }
\end{figure}
All three relations rely on the same notion of very simple Kripke worlds \W (\cref{fig:logrels-worlds}). 
Worlds consist of just a step-index $k$ that is accessed via function \stepsfun{\W}.
The use of this function is intended to facilitate future extensions of the Kripke worlds with additional information, but we do not currently make use of this extra generality.
The $\later{}$ modality and future world relation $\futw$ express that future worlds allow programs to take fewer reduction steps.
We define two different observation relations, one for each direction of the approximations we are interested in: 
$\obsfun{\W}{\lesssim}$ and 
$\obsfun{\W}{\gtrsim}$ while 
$\obsfun{\W}{\bothlogrel}$ indicates the intersection of those approximations.
The former defines that a source term approximates a target term if shrinking of the first in $\stepsfun{\W}$ steps or less implies termination of the second (in any number of steps).
The latter requires the reverse.
All of our logical relations will be defined in terms of either $\obsfun{\W}{\lesssim}$ or $\obsfun{\W}{\gtrsim}$. 
For definitions and lemmas or theorems that apply for both instantiations, we use the symbol $\anylogrel$ as a metavariable that can be instantiated to either $\lesssim$ or $\gtrsim$.

Note that our logical relations are not indexed by source types, but by \emph{pseudo-types} $\com{\psd{\tau}}$. 
Pseudo-types contain all the constructs of source types, plus an additional type which we indicate for now as $\com{\emuldvtext}$.
This type is not a source type; it is needed because of the approximate backtranslation, so we defer explaining its details until \Cref{sec:rel-bt}.
Function $\emtotau{\cdot}$ converts a pseudo-type to an actual source type by replacing all occurrences of $\com{\emuldvtext}$ with a concrete source type.%
\footnote{
  As a convention, superscripts of these auxiliary functions indicate the initials of the two languages involved.
}
We will sometimes silently use a normal source type where a pseudo-type is expected; this makes sense since the syntax for the latter is a superset of the former.
Function \srctotrgty{\cdot} converts a \stlcf pseudo-type into its \stlcim correspondent; this is needed because unlike the previous work of \citet{popl-journal}, all of our target languages are typed.
The formal details of both these functions are deferred until \com{\emuldvtext} is defined (\Cref{sec:rel-bt}) but we report their types below for clarity.
Finally, function \oftypefi{\cdot} checks that terms have the correct form according to the rules of syntactic typing (\Cref{sec:typ}).
\begin{gather*}
\begin{aligned}
  \psd{\tau} \bnfdef&\ \Units \mid \Bools \mid \src{\psd{\tau}\to\psd{\tau}} \mid \src{\psd{\tau}\times\psd{\tau}} \mid \src{\psd{\tau}\uplus\psd{\tau}} \mid \com{\emuldvtext}
    \text{ (to be defined in \Cref{sec:rel-bt})}
  \\
  \oftypefi{\psd{\tau}} \isdef
      &\
      \myset{(\src{v},\trg{v})}{ \src{v}\in\oftypes{\emtotau{\psd{\tau}}} \text{ and } \trg{v}\in\oftypet{\srctotrgty{\psd{\tau}}} }
\end{aligned}
\\
\begin{aligned}
	\oftypes{\src{\tau}} \isdef
		&\
		\myset{\src{v}}{\srce\vdash\src{v}:\src{\tau}}
	&
	\oftypet{{\tat}} \isdef
		&\
		\myset{\trg{v}}{\trge\vdash\trg{v}:\tat}
	\\
	\emtotau{\cdot} :&\ \psd{\tau} \to \src{\tau}
	\text{ (see \Cref{sec:rel-bt})}
		&
		\srctotrgty{\cdot} :&\ \psd{\tau} \to \iso{\tat}
		\text{ (see \Cref{sec:rel-bt})}
\end{aligned}
\end{gather*}
These definitions are used in the \lrfi relation and similar ones are used in the other ones, so we report their definitions and signatures below.
Function \oftypeie{\cdot} does the analogous syntactic typecheck but for terms of \stlcim and \stlcem and \oftypefe{\cdot} does it for terms of \stlcf and \stlcem.
Functions \emtotaufe{\cdot} and \emtotauic{\cdot} do the analogous conversion from pseudo types to actual types.
Function \srctoothty{\cdot} and \srctotrgtyic{\cdot} do the analogous conversion from source pseudo types to target actual types.
As we clarify later, \com{\emuldvtext} is indexed by target types, so essentially we have a set of pseudo types for the \stlcf to \stlcim compilation and a different set for the \stlcf to \stlcem compilation, and thus we need two different conversion functions (whose signatures look the same for now).
\begin{gather*}
  \begin{aligned}
    \oftypeie{\psdic{\tat}} \isdef
        &\
        \myset{(\iso{v},\equi{v})}{ \iso{v}\in\oftypet{\emtotauic{\psdic{\tat}}} \text{ and } \equi{v}\in\oftypeo{\srctotrgtyic{\psdic{\tat}}} }
    \\
    \oftypefe{\psd{\tau}} \isdef
    	&\
    	\myset{(\src{v},\equi{v})}{ \src{v}\in\oftypes{\emtotau{\psd{\tau}}} \text{ and } \equi{v}\in\oftypeo{\srctoothty{\psd{\tau}}} }
  \end{aligned}
  \\
  \begin{aligned}
    \oftypeo{{\tau}} \isdef
      &\
      \myset{\equi{v}}{\othe\vdash\equi{v}:\equi{\tau}}
  \end{aligned}
  \\
  \begin{aligned}
  \emtotaufe{\cdot} :&\ \psdfe{\tau}\to\src{\tau}
    &&&
    \emtotauic{\cdot} :&\ \psdic{\tat}\to\tat
    &&&
    \text{ (see \Cref{sec:rel-bt})}
  \\
  \srctoothty{\cdot} :&\ \src{\psdfe{\tau}} \to \equi{\tau}
    &&&
    \srctotrgtyic{\cdot} :&\ \trg{\psdic{\tat}} \to \equi{\tau}
    &&&
    \text{ (see \Cref{sec:rel-bt})}
\end{aligned}
\end{gather*}

\begin{figure}[!t]\small
  \centering
  \begin{align*}
    \later R \isdef
        &\
        \myset{ (\W,\src{v},\trg{v}) }{ \text{if } \stepsfun{\W}>0 \text{ then } (\laterfun{\W},\src{v},\trg{v}) \in R }
      \\
    \valrel{\Units} \isdef
      &\
      \myset{ (\W,\src{v},\trg{v}) }{ \src{v}=\units \text{ and } \trg{v}=\unitt }
    \\
    \valrel{\Bools} \isdef
      &\
      \myset{ (\W,\src{v},\trg{v}) }{ (\src{v}=\trues \text{ and } \trg{v}=\truet) \text{ or } (\src{v}=\falses \text{ and } \trg{v}=\falset) }
    \\
    \valrel{\psd{\tau}\to\psd{\tau'}} \isdef
      &\
      \myset{ (\W,\src{v},\trg{v}) }{
        \begin{aligned}
          &
          (\src{v},\trg{v})\in\oftypefi{\src{\psd{\tau}\to\psd{\tau'}}} \text{ and }
          \\
          &
          \exists \src{t},\trg{t}.~ \src{v}=\src{\lam{x:\emtotau{\psd{\tau}}}{t}}, \trg{v}=\trg{\lamt{x:\srctotrgty{\psd{\tau}}}{t}} \text{ and }
          \\
          &
          \forall \W', \src{v'}, \trg{v'}. \text{ if } \W'\strfutw\W \text{ and } (\W',\src{v'},\trg{v'})\in\valrel{\psd{\tau}} \text{ and } 
          \\
          & 
          (\text{if } \genlogrel = \overlogrel \text{ then }\size{\trg{v'}}\leq\stepsfun{\W'}) \text{ then }
          \\
          &
          (\W', \src{t}\subs{v'}{x}, \trg{t}\subt{v'}{x})\in\termrel{\psd{\tau'}}
        \end{aligned}
      }
    \\
    \valrel{\psd{\tau}\times\psd{\tau'}} \isdef
      &\
      \myset{ (\W,\src{v},\trg{v}) }{
        \begin{aligned}
          &
          (\src{v},\trg{v})\in\oftypefi{\src{\psd{\tau}\times\psd{\tau'}}} \text{ and }
          \\
          &
          \exists \src{v_1},\src{v_2},\trg{v_1},\trg{v_2}.~ \src{v}=\src{\pair{v_1,v_2}}, \trg{v}=\trg{\pair{v_1,v_2}} \text{ and }
          \\
          &
          (\W,\src{v_1},\trg{v_1})\in\later\valrel{\psd{\tau}} \text{ and } (\W,\src{v_2},\trg{v_2})\in\later\valrel{\psd{\tau'}} 
        \end{aligned}
      }
    \\
    \valrel{\psd{\tau}\uplus\psd{\tau'}} \isdef
      &\
      \myset{ (\W,\src{v},\trg{v}) }{
        \begin{aligned}
          &
          (\src{v},\trg{v})\in\oftypefi{\src{\psd{\tau}\uplus\psd{\tau'}}} \text{ and either }
          \\
          &
          \exists \src{v'},\trg{v'}.~ (\W,\src{v'},\trg{v'})\in\later\valrel{\psd{\tau}} \text{ and } \src{v}=\src{\inl{v'}}, \trg{v}=\trg{\inl{v'}} \text{ or}
          \\
          &
          \exists \src{v'},\trg{v'}.~ (\W,\src{v'},\trg{v'})\in\later\valrel{\psd{\tau'}} \text{ and } \src{v}=\src{\inr{v'}}, \trg{v}=\trg{\inr{v'}}
        \end{aligned}
      }
    \\
    \valrel{\com{\emuldvtext}} \isdef
      &\
      \text{ to be defined in \Cref{sec:rel-bt}}
    \\
    \contrel{\psd{\tau}} \isdef
      &\
      \myset{ (\W,\evalctxs,\evalctxt) }{
        \begin{aligned}
          &
          \forall \W',\src{v},\trg{v}.~ \text{if } \W'\futw\W \text{ and } (\W',\src{v},\trg{v})\in\valrel{\psd{\tau}} \text{ then }
          \\
          &
          (\evalctxhs{v},\evalctxht{v})\in\obsfun{\W'}{\anylogrel}
        \end{aligned}
       }
    \\
    \termrel{\psd{\tau}} \isdef
      &\
      \myset{ (\W,\src{t},\trg{t}) }{ \forall\evalctxs,\evalctxt.~ \text{if } (\W,\evalctxs,\evalctxt)\in\contrel{\psd{\tau}} \text{ then } (\evalctxhs{t},\evalctxht{t})\in\obsfun{\W}{\anylogrel} }
    \\
    \envrel{\srce} \isdef
      &\
      \{ (\W,\srce,\trge) \}
    \\
    \envrel{\src{\psd{\Gamma},x:\psd{\tau}}} \isdef
      &\
      \myset{ (\W,\src{\gamma\subs{v}{x}},\trgb{\gamma\subt{v}{x}}) }{ (\W,\src{\gamma},\trgb{\gamma})\in\envrel{\psd{\Gamma}} \text{ and } (\W,\src{v},\trg{v})\in\valrel{\psd{\tau}} }
    \\
    \cline{1-2}
    \valrele{\psdic{\matt}} \isdef
      &\
      \myset{
        (\W,\trg{v},\equi{v})
      }{
        \begin{aligned}
          &
          (\trg{v},\equi{v})\in\oftypeie{\trg{\psdic{\matt}}} \text{ and }
          \\
          &
          \exists \trg{v'}.~(\W,\trg{v'},\equi{v})\in\valrele{\psdic{\tat\subt{\matt}{\alpt}}} \text{ and } 
          \trg{v}=\trg{\fold{\matt}~v'}
        \end{aligned}
      }
    \\
    &
    \text{ The rest of \valrel{\psdic{\tat}} is analogous to the cases presented for \valrel{\psd{\tau}} }
    \\
    &
    \text{ The \contrel{\psdic{\tat}}, \termrel{\psdic{\tat}}, and \envrel{{\psdic{\Gat}}} relations are analogous to the presented ones }
    \\
    \cline{1-2}
    &
    \text{ The \valrel{\psd{\tau}}, \contrel{\psd{\tau}}, \termrel{\psd{\tau}}, and \envrel{{\psd{\Gamma}}} relations for \lrfe are}
    \\
    &
    \text{ analogous to the presented ones }
  \end{align*}
  \caption{
    Part of the three cross-language logical relations we rely on (classical bits) and its auxiliary functions.
    \label{fig:logrel-main}
    }
\end{figure}
The value relation $\valrel{\psd{\tau}}$ (\Cref{fig:logrel-main}) is defined inductively on source pseudo-types and it is quite standard save for an additional premise in the value relation for function types. 
\com{\Unitc} and \com{\Boolc} values are related in any world so long as they are the same value. 
Function values are related if they are well-typed, if both are lambdas, and if substituting related values in the bodies yields related terms in any strictly-future world.
Additionally, when the approximation direction is $\overlogrel$, we require that the world $\W'$ contains enough steps to bound the size of the target argument \trg{v'}.
This is a technicality that is required to complete the proof of \Cref{thm:inj-ext-sem-pres}, as we explain at the end of \Cref{sec:injext}.
Pair values are related if both are pairs and each projection is related in strictly-future worlds and sum values are related if they have the same tag ($\inl{}$ or $\inr{}$) and the tagged values are related in strictly-future worlds.
Finally, the value relation for recursive types used by \lrie is not defined on strictly-future worlds because in an equi-recursive language, values of recursive type can be inspected without consuming a step.
However, this does not compromise well-foundedness of the relation because our recursive types \equi{\mat} are contractive, so the recursion variable \equi{\alpha} in \equi{\tau} must occur under a type constructor such as \equi{\to} and the relation for these constructors recurses only at strictly-future worlds.

The value, evaluation context and term relations are defined by mutual recursion, using a technique called biorthogonality (see, e.g., \citep{bistcc}).
Evaluation contexts $\contrel{\psd{\tau}}$ are related in a world if plugging in related values in any future world yields terms that are related according to the observation relation of the world.
Similarly, terms are related $\termrel{\psd{\tau}}$ if plugging the terms in related evaluation contexts yields terms related according to the observation relation of the world.
Relation $\envrel{\psd{\Gamma}}$ relates substitutions; this simply requires that substitutions for all variables in the context are for related values.

We indicate open terms to be logically related according to the three relations as follows (\Cref{def:logrel}, \Cref{def:logrelie,def:logrelfe}).
Those definitions rely on terms being related up to $n$ steps (\Cref{def:logrel-n-steps}) which we present for \lrfi only since the other definitions are analogous.
Here, when we apply \srctotrgty{\cdot} to typing contexts, we mean the application of \srctotrgty{\cdot} to all bindings in the context.
\begin{definition}[Logical relation up to $n$ steps for \lrfi]\label{def:logrel-n-steps}
  \begin{align*}
    \psd{\Gamma}\vdash\src{t}\anylogreln{n}\trg{t}:\psd{\tau} \isdef
      &\
      \emtotau{\psd{\Gamma}}\vdash\src{t}:\emtotau{\psd{\tau}}
    \\
    \text{and }
      &\
      \srctotrgty{\psd{\Gamma}}\vdash\trg{t}:\srctotrgty{\psd{\tau}}
    \\
    \text{and }
      &\
      \forall \W.~
    \\
    \text{if }
      &\
      \stepsfun{\W}\leq n
    \\
    \text{then }
      &\
      \forall \src{\gamma},\trgb{\gamma}.~
      (\W,\src{\gamma},\trgb{\gamma})\in\envrel{\psd{\Gamma}}, 
    \\
      &\
      (\W,\src{t\gamma},\trg{t\trgb{\gamma}})\in\termrel{\psd{\tau}}
  \end{align*}
\end{definition}

\begin{definition}[\lrfi Logical relation]\label{def:logrel}
  \begin{align*}
    \psd{\Gamma}\vdash\src{t}\anylogrel\trg{t}:\psd{\tau} \isdef
    &\
    \forall n\in\mb{N}.~ \psd{\Gamma}\vdash\src{t}\anylogreln{n}\trg{t}:\psd{\tau} 
  \end{align*}
\end{definition}

\begin{definition}[\lrie Logical relation]\label{def:logrelie}
  \begin{align*}
    \psdic{\Gat}\vdash\iso{t}\anylogrel\equi{t}:\psdic{\tat} \isdef
    &\
    \forall n\in\mb{N}.~ \psdic{\Gat}\vdash\iso{t}\anylogreln{n}\equi{t}:\psdic{\tat} 
  \end{align*}
\end{definition}

\begin{definition}[\lrfe Logical relation]\label{def:logrelfe}
  \begin{align*}
    \psd{\Gamma}\vdash\src{t}\anylogrel\equi{t}:\psd{\tau} \isdef
    &\
    \forall n\in\mb{N}.~ \psd{\Gamma}\vdash\src{t}\anylogreln{n}\equi{t}:\psd{\tau} 
  \end{align*}
\end{definition}

An open source term is related up to \com{n} steps at pseudo-type \src{\psd{\tau}} in pseudo-context \src{\psd{\Gamma}} to a target open term if both are well-typed and closing both terms with substitutions related in \src{\psd{\Gamma}} produces terms related at \src{\psd{\tau}} in any world that has at least \com{n} steps.
If terms are related for any number of steps, we simply omit the \com{n} index and write $\src{\psd{\Gamma}}\vdash\src{t}\anylogrel\iso{t}:\src{\psd{\tau}}$.
Since we have to also relate program contexts across languages, we define what it means for them to be related as follows.
\begin{definition}[\lrfi Logical relation for program contexts]\label{def:logrel-ctx}
  \begin{align*}
    \vdash\ctxs\anylogrel\ctxt:\psd{\Gamma},\psd{\tau}\to\psd{\Gamma'},\psd{\tau'} \isdef
    &\
    \vdash\ctxs:\psd{\Gamma},\psd{\tau}\to\psd{\Gamma'},\psd{\tau'} 
    \\
    \text{and }
    &\
    \vdash\ctxt:\srctotrgty{\psd{\Gamma}},\srctotrgty{\psd{\tau}}\to
    \\
    &\
    \qquad\qquad\qquad\srctotrgty{\psd{\Gamma'}},\srctotrgty{\psd{\tau'}} 
    \\
    \text{and }
    &\
    \forall \src{t},\trg{t}
    \\
    \text{if }
    &\
    \psd{\Gamma}\vdash\src{t}\anylogrel\trg{t}:\psd{\tau} 
    \\
    \text{then }
    &\
    \psd{\Gamma'}\vdash\ctxhs{t}\anylogrel\ctxht{t}:\psd{\tau'} 
  \end{align*}
\end{definition}

\begin{definition}[\lrie Logical relation for program contexts]\label{def:logrel-ctx-ei}
  \begin{align*}
    \vdash\ctxt\anylogrel\ctxo:\trg{\Gamma},{\tat}\to\trg{\Gamma'},\trg{\tat'} \isdef
    &\
    \vdash\ctxt:\trg{\Gamma},{\tat}\to\trg{\Gamma'},\trg{\tat'} 
    \\
    \text{and }
    &\
    \vdash\ctxo:\srctotrgtyic{\oth{\psdic{\Gamma}}},\srctotrgtyic{\oth{\psdic{\tau}}}\to
    \\
    &\
    \qquad\qquad\qquad\srctotrgtyic{\oth{\psdic{\Gamma'}}},\srctotrgtyic{\oth{\psdic{\tau'}}} 
    \\
    \text{and }
    &\
    \forall \trg{t},\oth{t}
    \\
    \text{if }
    &\
    \psdic{\Gamma}\vdash\trg{t}\anylogrel\oth{t}:\psdic{\tat} 
    \\
    \text{then }
    &\
    \psdic{\Gamma'}\vdash\ctxht{t}\anylogrel\ctxho{t}:\psdic{\tat'} 
  \end{align*}
\end{definition}

\begin{definition}[\lrfe Logical relation for program contexts]\label{def:fe-logrel-ctx}
  \begin{align*}
    \vdash\ctxs\anylogrel\ctxo:\psdfe{\Gamma},\psdfe{\tau}\to\psdfe{\Gamma'},\psdfe{\tau'} \isdef
    &\
    \vdash\ctxs:\psdfe{\Gamma},\psdfe{\tau}\to\psdfe{\Gamma'},\psdfe{\tau'} 
    \\
    \text{and }
    &\
    \vdash\ctxo:\srctoothty{\psdfe{\Gamma}},\srctoothty{\psdfe{\tau}}\to
    \\
    &\
    \qquad\qquad\qquad\srctoothty{\psdfe{\Gamma'}},\srctoothty{\psdfe{\tau'}} 
    \\
    \text{and }
    &\
    \forall \src{t},\oth{t}
    \\
    \text{if }
    &\
    \psdfe{\Gamma}\vdash\src{t}\anylogrel\oth{t}:\psdfe{\tau} 
    \\
    \text{then }
    &\
    \psdfe{\Gamma'}\vdash\ctxhs{t}\anylogrel\ctxho{t}:\psdfe{\tau'} 
  \end{align*}
\end{definition}

Program contexts are related if they are well-typed and if plugging terms related at the pseudo-type of the hole (\com{\psd{\tau}}) in each of them produces terms related at the pseudo-type of the result (\com{\psd{\tau'}}).

All our logical relations are constructed so that for related terms, termination of one term implies termination of the other according to the direction of the approximation ($\lesssim$ or $\gtrsim$) (\Cref{thm:log-rel-adeq-both}).
\begin{lemma}[Adequacy for $\bothlogrel$ for \lrfi]\label{thm:log-rel-adeq-both}
  \begin{align*}
    \text{if }
      &
      \srce \vdash \src{t} \underlogreln{n} \trg{t} : \src{\tau}
    \text{ and }
        \src{t \shrinks{m} v} \text{ with } n \geq m
    \text{ then }
      \trg{t \termt}
  \\
    \text{if }
      &
      \srce \vdash \src{t} \overlogreln{n} \trg{t} : \src{\tau}
    \text{ and }
        \trg{t \shrinkt{m} v} \text{ with } n \geq m
      \text{ then }
        \src{t \termsl}
  \end{align*}
\end{lemma}

\section{Fully-abstract compilation and Approximate Backtranslations}\label{sec:fac}
This section provides an overview of fully-abstract compilation and of the approximate backtranslation proof technique that we use (\Cref{sec:fac-primer}).
The approximate backtranslation requires defining the backtranslation type, i.e., the type that represents backtranslated values (\Cref{sec:bt-type}).
This type provides the insights needed to complete the definitions of our logical relations and to understand how to reason about backtranslated terms cross-languages (\Cref{sec:rel-bt}).

\subsection{A Primer on Fully-Abstract Compilation and Approximate Backtranslations}\label{sec:fac-primer}
A compiler is fully abstract if it preserves and reflects contextual equivalence between source and target language~\citep{DBLP:conf/icalp/Abadi98}.
Many compiler passes have been proven to satisfy this criterion~\citep{fstar2js,Ahmed:2008:TCC:1411203.1411227,ahmedCPS,max-embed,popl-journal,scoo-j,skorstengaard-stktokens:2019,van_strydonck_linear_2019}, we refer the interested reader to the survey of \citet{scsurvey}.

Two programs are contextually equivalent if they produce the same behaviour no matter the larger program (i.e., program context) they interact with~\citep{lcfConsidered}.
As commonly done, we define ``producing the same behaviour'' as equi-termination (one terminates iff the other does).
We use a complete formulation of contextual equivalence for typed programs, which enforces that contexts are well-typed and their types match that of the terms considered.
\begin{definition}[Contextual Equivalence]\label{def:ceq}\hfill
	\begin{align*}
		\com{\Gamma}\vdash\com{t_1\ceqc t_2} : \com{\tau} \isdef 
			&\
			\com{\Gamma}\vdash\com{t_1}:\com{\tau} \text{ and } \com{\Gamma}\vdash\com{t_2}:\com{\tau} \text{ and } 
		\\
			&\
			\forall \ctxc\ldotp \ctxc:\com{\Gamma},\com{\tau}\to\come,\com{\tau'} \ldotp \ctxhc{t_1}\term\iff\ctxhc{t_2}\term
	\end{align*}
\end{definition}
Quantifying over all contexts in \Cref{def:ceq} ensures that contextually-equivalent terms do not just equi-terminate, but that any value the context can obtain from them is indistinguishable.

For a compiler \comp{\cdot} from language \langS to \langT, we define full abstraction as follows:
\begin{definition}[Fully-abstract compilation]\label{def:fac}\hfill
    \[
    	\vdash\comp{\cdot}:\fa \isdef \forall\src{t_1},\src{t_2}\in\langS\ldotp \srce\vdash\src{t_1\ceqs t_2}:\src{\tau} \iff \trge\vdash\trg{\comp{\src{t_1}}\ceqt\comp{\src{t_2}}}:\trg{\comp{\src{\tau}}}
    \]
\end{definition}
For simplicity, we instantiate \Cref{def:fac} for closed terms only (i.e., well-typed under empty environments).
Opening the environment to a non-empty set of term variables is straightforward and therefore omitted~\citep{popl-journal}.

\begin{figure}[!ht]
\centering
\begin{tikzpicture}[scale=0.8,every node/.style={scale=.9}]
    \node at (5,4.7) { $\src{t_1\mathrel{\overset{\bl{?}}{\ceqs}} t_2}$ };

    \node at (3.4,4) { $\src{\ctxs\hole{t_1} \termsl}$ };
    \node at (5,4) { $\xRightarrow{\phantom{ooo}?\phantom{ooo}}$ };
    \node at (6.6,4) { $\src{\ctxs\hole{t_2} \termsl}$ };

    \node at (4.05,3) { (1) };
    \node at (5,2.4) { (2) };
    \node at (5.95,3) { (3) };

    \draw[out=260,in=100,double,-implies,double equal sign distance] (3.7,3.4) to (3.7,2.6);

    \draw[out=80,in=280,double,-implies,double equal sign distance] (6.3,2.6) to (6.3,3.4);

    \node[align=left] at (7.5,3) { $ \ctxs \bothlogrel \comp{\ctxs}$ \\
      $ \src{t_2} \bothlogrel \comp{\src{t_2}}$};
    \node[align=left] at (2.5,3) { $ \ctxs \logrel \comp{\ctxs}$ \\
      $ \src{t_1} \logrel \comp{\src{t_1}}$};

    \node at (3.4,2) { $\trg{\comp{\ctxs}\hole{\comp{\src{t_1}}}} \termt$ };
    \node at (5,2.1) { $\xRightarrow{\phantom{ooooo}}$ };
    \node at (6.6,2) { $\trg{\comp{\ctxs}\hole{\comp{\src{t_2}}}} \termt$ };

    \node at (5,1.3) { $\trg{\comp{\src{t_1}}\ceqt\comp{\src{t_2}}}$ };

    \draw[out=90,in=-90,double,-implies,double equal sign distance] (1.5,1) to node[sloped, yshift =.7em]{\small reflection direction} (1.5,5);
  \end{tikzpicture}
  \hspace{1em}
  \begin{tikzpicture}[scale=0.84,every node/.style={scale=.9}]
    \node at (5,4.7) { $\src{t_1\ceqs t_2}$ };

    \node at (3,4) { $\src{\backtr{\ctxt{}}_{\com{n}} \hole{ \src{t_1}} \termsl_{\_}}$ };
    \node at (5,4.2) { $\xRightarrow{\phantom{ooooooo}}$ };
    \node[] at (7,4)(ss2) { $\src{\backtr{\ctxt{}}_{\com{n}} \hole{\src{t_2}} \termsl_{\_}}$ };

    \node[] at (7,3.3)(ss1) { $\src{\backtr{\ctxt{}}_{\com{n}} \hole{\src{t_2}} \shrinks{\_}}$ };

    \node at (4.05,3) { (1) };
    \node at (5,3.6) { (2) };
    \node at (5.95, 2.3) { (3) };

    \draw[out=100,in=260,double,-implies,double equal sign distance] (3.7,2.6) to (3.7,3.4);

    \draw[out=280,in=80,double,-implies,double equal sign distance] (6.3,2.7) to (6.3,1.9);

    \node[align=left] at (7.8, 2.3) { $ \src{\backtr{\ctxt{}}_{\com{n}}} \lesssim_{\_} \trg{\ctxt{}}$ \\ $ \src{t_2} \lesssim_{\_} \comp{\src{t_2}}$};
    \node[align=left] at (2.2, 3) { $ \src{\backtr{\ctxt{}}_{\com{n}}} \gtrsim_{\com{n}} \trg{\ctxt{}}$ \\ $ \src{t_1} \gtrsim_{\_} \comp{\src{t_1}}$};

    \node[]at (3,1.3)(st2) { $\trg{\ctxt{} \hole{\comp{\src{t_1}}}  \termt_{\com{j}}}$ };

    \node[]at (3,2)(st1) { $\trg{\ctxt{} \hole{\comp{\src{t_1}}}  \shrinkt{{n}}}$ };

    \node at (5,1.4) { $\xRightarrow{\phantom{ooo}?\phantom{ooo}}$ };
    \node at (7,1.3) { $\trg{\ctxt{} \hole{\comp{\src{t_2}}} \termt_{\_}}$ };

    \node at (5,.7) { $\comp{\src{t_1}}\mathrel{\overset{?}{\ceqt}}\comp{\src{t_2}}$ };

    \draw[out=100,in=260,double,-implies,double equal sign distance] (st2.west) to node[midway,left](){\scriptsize{Thm~\ref{thm:term-def-rel}}} (st1.west);
    \draw[out=280,in=80,double,-implies,double equal sign distance] (ss2.east) to node[midway,right](){\scriptsize{Thm~\ref{thm:term-def-rel}}} (ss1.east);

    \draw[out=-90,in=90,double,-implies,double equal sign distance] (0,5) to node[sloped, yshift =.7em]{\small preservation direction} (0,0.3);
  \end{tikzpicture}
  \caption{Diagram breakdown of the reflection (left) and preservation (right) proofs of fully-abstract compilation.\label{fig:fac-dia}\label{fig:fa-refl}}
\end{figure}
\subsubsection{Proving Fully-Abstract Compilation: Reflection (or, the Easy Part)}\label{sec:fac-refl}

The reflection part of fully-abstract compilation requires that the compiler produces equivalent target programs only if their source counterparts were equivalent.
Contrapositively, inequivalent source programs must be compiled to inequivalent target programs.
This proof can often be derived as a corollary of standard compiler correctness (i.e., refinement)~\citep{scsurvey}.

As mentioned, we prove the reflection direction by relying on the cross-language logical relations.
Our logical relations are compiler-agnostic---they simply state when terms approximate each other (recall that $\bothlogrel$ is the intersection of both approximations $\lesssim$ and $\gtrsim$).
However, we use them to show that any term (and program context) is related to its compilation.
With this fact, by relying on the adequacy of logical relations (\Cref{thm:log-rel-adeq-both}), we know that related terms equi-terminate.
Thus, we can apply the reasoning depicted in \Cref{fig:fa-refl} (left) to conclude this part of fully-abstract compilation.

\subsubsection{Proving Fully-Abstract Compilation: Preservation (or, the Hard Part)}\label{sec:fac-pres}
Fully-abstract compilation proofs are notorious and their complexity resides in the \emph{preservation} direction.
That is, starting from contextually-equivalent programs in the source, prove that their compiled counterparts are contextually-equivalent in the target.
For our three fully-abstract compilation results we rely on the approximate backtranslation proof technique~\citep{popl-journal}, depicted in \Cref{fig:fac-dia} (right).

We rely on both directions of the cross-language approximation relating terms for this proof.
Recall that $\src{t}\gtrsim_n\trg{t}$ is used to know that if \trg{t} shrinks in \com{n} steps in the target, then \src{t} also terminates (in arbitrary steps) in the source.
The converse, $\src{t}\lesssim_n\trg{t}$ is used to know that if \src{t} shrinks in \com{n} steps in the source, then \trg{t} also terminates (again in arbitrary steps) in the target.
We start with source term \src{t} approximating (in both directions) its compilation \comp{\src{t}}.
Then, to prove target contextual equivalence (the ?-decorated equivalence), we start by assuming that a target context \ctxt linked with \comp{\src{t_1}} terminates in some steps ($\termt_\com{n}$).
By relying on \Cref{thm:term-def-rel}, we know that \ctxt linked with \comp{\src{t_1}} \newterms in some steps (\shrinkt{n'}).
Eventually, we need to show that the same target context linked with \comp{\src{t_2}} also terminates in any steps ($\termt_{\com{\_}}$).
This is the ?-decorated implication, the reverse direction holds by symmetry.
To progress, we construct a \emph{backtranslation} $\backtr{\cdot}_\com{n}$, i.e., a function that takes a target context \ctxt{} and returns a source context that approximates \ctxt in both directions.
With the backtranslation and this direction of the approximation $\gtrsim_{\com{n}}$, we prove implication (1): the backtranslated context $\backtr{\ctxt}_\com{n}$ linked with \src{t_1} terminates in the source.
At this point, the assumption of source contextual equivalence yields implication (2): the same backtranslated context $\backtr{\ctxt}_\com{n}$ linked with \src{t_2} also terminates (\termsl{}).
Here we apply again \Cref{thm:term-def-rel} to know that $\backtr{\ctxt}_\com{n}$ linked with \src{t_2} \newterms (\shrinks{\_}).
Now we rely on the another direction of the approximation between the target context and its backtranslation (as well as between source terms and their compilation): $\lesssim_{\_}$.
This other approximation lets us conclude implication (3): the original target context \ctxt linked with \comp{\src{t_2}} terminates in the target.
This is what we prove for a compiler to be fully abstract.

\subsection{A Family of Backtranslation Types}\label{sec:bt-type}
Backtranslated contexts must be valid source contexts, i.e., they need to be well typed in the source.
However, \stlcf does not have recursive types, so what is the source-level correspondent of \iso{\matt}?

We adapt the same intuition of previous work \citep{Devriese:2016:FCA:2837614.2837618,popl-journal} in our setting too: it is not necessary to precisely embed target types into the source language in order to backtranslate terms.
In fact, we need to reason for \emph{up to \com{n} steps}, which means that we can approximate target types \emph{n-levels deep}.
Thus, concretely, we do not need recursive types in \stlcf.
Given a target recursive type, we unfold it \com{n} times and backtranslate its unfolding to model the \com{n} target reductions required.

\begin{figure}
  \begin{align*}
    \uval{0;\trg{\tat}} \isdef
      &\
      \Units
    \\
    \uval{n+1;\trg{\tat}} \isdef
      &\
      \begin{cases}
        \src{\Units\uplus\Units}
        &
        \text{ if } \trg{\tat}=\Unitt
        \\
        \src{\Bools\uplus\Units}
        &
        \text{ if } \trg{\tat}=\Boolt
        \\
        \src{(\uval{n;\trg{\tat}}\to\uval{n;\trg{\tat'}})\uplus\Units}
        &
        \text{ if } \trg{\tat}=\trgb{\tat\to\tat'}
        \\
        \src{(\uval{n;\trg{\tat}}\times\uval{n;\trg{\tat'}})\uplus\Units}
        &
        \text{ if } \trg{\tat}=\trgb{\tat\times\tat'}
        \\
        \src{(\uval{n;\trg{\tat}}\uplus\uval{n;\trg{\tat'}})\uplus\Units}
        &
        \text{ if } \trg{\tat}=\trgb{\tat\uplus\tat'}
        \\
        \src{\uval{n;\trg{\tat'\subt{\matgent{\alpt}{\tat'}}{\alpt}}}\uplus\Units}
        &
        \text{ if } \trg{\tat}=\trg{\matgent{\alpt}{\tat'}}
      \end{cases}
  \\
  \cline{1-2}
  \uvalic{n;\oth{\tau}} \isdef
    &\
    \text{ as } \uvalfe{n;\oth{\tau}}
  \\
  \cline{1-2}
    \uvalfe{n+1;\oth{\tau}} \isdef
      &\
      \begin{cases}
        \text{omitted cases are as above}
      \\
        \src{\uvalfe{n+1;\oth{\tau'\subo{\matgen{\alpha}{\tau'}}{\alpha}}}}
        &
        \text{ if } \oth{\tau}=\oth{\matgen{\alpha}{\tau'}}
      \end{cases}
  \end{align*}
  \caption{\label{fig:uval}The type of backtranslated terms.}
\end{figure}
According to this strategy, the backtranslation of a term of type $\tat$ should have type \emph{unfold \tat\ \com{n} times}.
During this unfolding, however, things can go wrong.
Specifically, the backtranslated code does not know at runtime the level of unfolding we are dealing with, i.e., it cannot inspect \com{n} at runtime.
Thus, we need a way to model the term reaching more than \com{n} unfoldings, because in that case the backtranslated code needs to diverge.
Recall in fact that one of the two terms (\comp{\src{t_1}} and \comp{\src{t_2}}) is guaranteed to terminate within \com{n} steps.
Therefore, if that termination does not happen, the backtranslated code to diverge; this ensures that contextually-equivalent terms remain equivalent, i.e., they equi-terminate.
Thus at each level of unfolding, we backtranslate \tat\ into ``\src{\tat\uplus\Unit}'' (we will make this formal below), where the right \Units models failure.
Then any time the backtranslation code receives a value which inhabits the `right \Units' type of the backtranlation type, it will diverge, knowing that it is not dealing with the term that had to terminate within the \com{n} unfoldings.

We make these intuitions concrete and formalise the type for \stlcim values backtranslated into \stlcf as \uval{n;\tat} in \Cref{fig:uval} (for \src{B}ack\src{t}ranslation \src{T}ype; the superscript indicates the languages involved, the subscripts are effectively parameters of this type).
Type \uval{n;\tat} is defined inductively on \src{n} and it backtranslates the structure of \tat\ in the source type it creates.
At no steps (\src{n}=\src{0}), the backtranslation is not needed any more because intuitively we already performed the \com{n} steps, so the only type is \Units.
Otherwise, the backtranslated type maintains the same structure of the target type.
In the case for \matt, the backtranslated type is the unfolding of \matt, but at a decremented index (\src{n}).
Intuitively, this is to match the reduction step that will happen in the target for eliminating \iso{\unfold{\matt}\ \fold{\matt}} annotations.

The type of \stlcem terms backtranslated in \stlcf (\uvalfe{n;\oth{\tau}}, still in \Cref{fig:uval}) has an important difference.
The case for \oth{\mat} does not lose a step in the index and simply performs the unfolding of the recursive type without an additional \src{\uplus\Units}.
This difference matches the fact that in \stlcem there is no additional reduction rule in the semantics.
Additionally, this difference affects the helper functions needed to deal with values of backtranslation type, as we discuss later.

Intuitively, the fact that the backtranslation of a recursive type is its \com{n}-level deep unfolding is possible because \equi{\mat} is contractive in \equi{\alpha}.
This is sufficient because we need to only replicate \com{n} steps in order to differentiate terms, so a \com{n}-level deep unfolding of the type suffices in order to reach the differentiation.
For example, let us take the type of list of booleans in \stlcem: 
  \[
    \equi{\matgen{\alpha}{\Unito\uplus(\Boolo\times\alpha)}} \text{ (which we dub \equi{List_B})}
  \] 
and its first unfolding:
  \[
    \equi{\Unito\uplus(\Boolo\times List_B)} \text{ (which we dub \equi{List_B^1})}
  \]
the backtranslation (for \com{n=3}) for this type is:
\begin{align*}
\uvalfe{3;\equi{List_B}} 
  &= 
  \uvalfe{3;\equi{\Unito\uplus(\Boolo\times List_B)}} 
\\
  &= 
  \equi{((\uvalfe{2;\Unito})\uplus\uvalfe{2;\equi{\Boolo\times List_B}})\uplus\Unito} 
\\
  &= 
  \equi{((\Unito\uplus\Unito)\uplus(((\uvalfe{1;\Boolo})\times\uvalfe{1;\equi{List_B}})\uplus\Unito))\uplus\Unito} 
\\
  &=
  \equi{((\Unito\uplus\Unito)\uplus(((\Boolo\uplus\Unito)\times\uvalfe{0;\equi{List_B^1}})\uplus\Unito))\uplus\Unito} 
\\
  &= 
  \equi{((\Unito\uplus\Unito)\uplus(((\Boolo\uplus\Unito)\times\Unito)\uplus\Unito))\uplus\Unito}
\end{align*}
Formally, the measure that ensures that this type is well founded is the precision $n$ together with \lMuCount{\equi{\mat}} i.e., the number of leading \oth{\mu}s in type $\equi{\tau}$, for reasons analogous to those discussed in \Cref{sec:typ}.
The type of \stlcem terms backtranslated in \stlcim (\uvalic{n;\oth{\tau}}) is the same as the one just presented (\uvalfe{n;\oth{\tau}}).
Intuitively, this is because the \com{n}-level deep unfolding of \oth{\tau} in the backtranslation type does not rely on recursive types in \stlcim.

\subsubsection{Working with the Backtranslation Type}
In order to work with values of backtranslated type, we need a way to create and destruct them.
Additionally, we need a way to increase and decrease the approximation level (the \com{n} index), for reasons we explain below.
This is what we present now mainly for terms of type \uval{n;\tat}, though we report the most interesting cases for the other backtranslation types too.
Recall that the definitions of the other two backtranslation types are the same, so these helpers are also the same and we report only one.

Given a target value \iso{v} of type \tat, in order to \emph{create} a source term of type \uval{n;\tat} it suffices to create \src{\inl{v}} (informally).
However, in order to \emph{use} a source term of type \uval{n;\tat} at the expected type \tat, we need to destroy it according to \tat: this is done by the family of source functions \casetag{n;\tat}.
\begin{align*}
  \casetag{n;\trg{\tat}} =
    &\
    \src{
      \lam{x:\uval{n+1;\trg{\tat}}}{
        \caseof{x}{x_1}{\myomega_{\uval{n;\trg{\tat}}}}
      }
    }
\end{align*}
Intuitively, all these functions strip the value of type \uval{n+1;\tat} they take in input of the \src{\inl{}} tag and return the underlying value.
Thus, at arrow type, the returned value has type \src{ (\uval{n;\tat}\to\uval{n;\trgb{\tau'}})} while at recursive type it has type \uval{n;\trg{\tat\subt{\matt}{\alpt}}}.
In case the wrong value is passed in (i.e., it is an \src{\inr{}}), these functions diverge via term $\src{\myomega_{\uval{n;\tat}}}$, which is easily encodable in \stlcf.

Recall that the \uvalfe{n;\equi{\tau}} for $\equi{\tau}=\equi{\mat}$ is different: it is just \uvalfe{n;\equi{\tau\subo{\mat}{\alpha}}} so the type is unfolded and the index is the same.
The destructor used for this backtranslation type (\casetagfe{n;\oth{\mat}}{}{}) is therefore different than the one above.
Specifically, we do not need to destruct a backtranslated type indexed with \equi{\mat} because that never arises (i.e., the type is unfolded).
Consider type \uvalfe{3;\equi{List_B}} from before: at index \com{3} the backtranslation does not handle values of recursive type but of type \uvalfe{3;\equi{List_B^1}}.
That is, it handles values whose top-level connector is the \equi{\uplus} of \equi{List_B}.
Finally, the destructor used for \uvalic{n;\equi{\mat}} (\casetagic{n;\equi{\mat}}{}{}) is analogous to this last one (\casetagfe{n;\equi{\mat}}).
\begin{align*}
	\casetagfe{n;\equi{\tau}} =
	&\
		\src{
		\lam{x:\uvalfe{n+1;\equi{\tau}}}{
		\caseof{x}{x_1}{\myomega_{\uvalfe{n;\equi{\tau}}}}
		}
		}
	&
	\equi{\tau} \neq \equi{\mat}
	\\
	\casetagic{n;\equi{\tau}} =
	&\
		\iso{
		\lam{x:\uvalic{n+1;\equi{\tau}}}{
		\caseof{x}{x_1}{\myomega_{\uvalic{n;\equi{\tau}}}}
		}
		}
	&
	\equi{\tau} \neq \equi{\mat}
\end{align*}

\begin{figure}[!t]\small
\mytoprule{ 
	\upgrade{n;\tat} :
    $\src{\uval{n;\tat} \to \uval{n+1;\tat}}$
}

\begin{align*}
	\upgrade{0;d;\trgb{\tau}} =
		&\
		\src{\lam{x:\uval{0;\trgb{\tau}}}{\unk_d}}
	\\
	\upgrade{n+1;d;\Unitt} =
		&\
		\src{\lam{x:\Units\uplus\Units}{x}}
	\qquad
	\upgrade{n+1;d;\Boolt}
		=
		\src{\lam{x:\Bools\uplus\Units}{x}}
	\\
	\upgrade{n+1;d;\trgb{\tau\times\tau'}} =
		&\
		\src{
			\begin{aligned}[t]
				&
				\lam{
					\src{x:\uval{n+1;\trgb{\tau\times\tau'}}}
				}{
				\\
				&\
				\casefoldeds{\src{x}}{\src{\inl{\pair{\upgrade{n;d;\trgb{\tau}}\projone{x_1},\upgrade{n;d;\trgb{\tau'}}\projtwo{x_1}}}}}{\src{\inr{x_2}}}
				}
			\end{aligned}
		}
	\\
	\upgrade{n+1;d;\trgb{\tau\uplus\tau'}} =
		&\
		\src{
			\begin{aligned}[t]
				&
				\lam{
					\src{x:\uval{n+1;\trgb{\tau\uplus\tau'}}}
					}{
				\\
				&
				\casefoldeds{\src{x}}{
					\src{\inl{
						\casefoldeds{x_1}{
							\src{\inl{(\upgrade{n;d;\trgb{\tau}} x_1)}}
						}{
							\src{\inr{(\upgrade{n;d;\trgb{\tau'}} x_2)}}
						}
					}}
				}{\src{\inr{x_2}}}}
			\end{aligned}
		}
	\\
	\upgrade{n+1;d;\trgb{\tau\to\tau'}} =
		&\
		\src{
			\begin{aligned}[t]
				&
				\lam{
					\src{x:\uval{n+1;\trgb{\tau\to\tau'}}}
					}{
				\\
				&
				\casefoldeds{\src{x}}{
					\src{\inl{
						\lam{z:\uval{n+1;\trgb{\tau}}}{
							\upgradepar{n;d;\trgb{\tau'}}{x_1~ (\downgrade{n;d;\trgb{\tau}} z) }
						}
					}}
				}{\src{\inr{x_2}}}
			}
			\end{aligned}
		}
	\\
	\upgrade{n+1;d\trgb{\matgen{\alpha}{\tau'}}} =
		&\
		\src{
			\begin{aligned}[t]
				&
				\lam{
					\src{x:\uval{n+1;\trgb{\matgen{\alpha}{\tau'}}}}
					}{
				\casefoldeds{\src{x}}{
					\src{\inl{
						(\upgrade{n;d;\trgb{\tau'\subt{\matgen{\alpha}{\tau'}}{\alpha}}} x_1)
					}}
				}{\src{\inr{x_2}}}}
			\end{aligned}
		}
\end{align*}

\hrule
  \begin{align*}
    \upgradeic{n;\equi{\tau}} =&\ \text{ as } \upgradefe{n;\equi{\tau}}
  \end{align*}  
\hrule
  \begin{align*}
    \upgradefe{n+1;\equi{\mat}} =&\
      \upgradefe{n+1;\equi{\tau\subo{\mat}{\alpha}}}
    &
    \upgradefe{n;\equi{\tau}} =&\ \text{ as above}
  \end{align*}
  \caption{\label{fig:updn-up}Definition of the \src{upgrade} function.}
\end{figure}

\begin{figure}[!t]\small
\mytoprule{ 
  \downgrade{n;\tat} :
    $\src{\uval{n+1;\tat} \to \uval{n;\tat}}$
}

\begin{align*}
	\downgrade{0;d;\trgb{\tau}} =
		&\
		\src{\lam{x:\uval{d;\trgb{\tau}}}{ \units }}
	\\
	\downgrade{n+1;d;\Unitt} =
		&\
		\src{\lam{x:\Units\uplus\Units}{x}}
	\qquad
	\downgrade{n+1;d;\Boolt}
		=
		\src{\lam{x:\Bools\uplus\Units}{x}}
	\\
	\downgrade{n+1;d;\trgb{\tau\times\tau'}} =
		&\
		\src{
			\begin{aligned}[t]
				&
				\lam{
					\src{x:\uval{n+1+d;\trgb{\tau\times\tau'}}}
				}{
				\\
				&\
				\casefoldeds{\src{x}}{\src{\inl{\pair{\downgrade{n;d;\trgb{\tau}}\projone{x_1},\downgrade{n;d;\trgb{\tau'}}\projtwo{x_1}}}}}{\src{\inr{x_2}}}
				}
			\end{aligned}
		}
	\\
	\downgrade{n+1;d;\trgb{\tau\uplus\tau'}} =
		&\
		\src{
			\begin{aligned}[t]
				&
				\lam{
					\src{x:\uval{n+1+d;\trgb{\tau\uplus\tau'}}}
					}{
				\\
				&
				\casefoldeds{\src{x}}{
					\src{\inl{
						\casefoldeds{x_1}{
							\src{\inl{(\downgrade{n;d;\trgb{\tau}} x_1)}}
						}{
							\src{\inr{(\downgrade{n;d;\trgb{\tau'}} x_2)}}
						}
					}}
				}{\src{\inr{x_2}}}}
			\end{aligned}
		}
	\\
	\downgrade{n+1;d;\trgb{\tau\to\tau'}} =
		&\
		\src{
			\begin{aligned}[t]
				&
				\lam{
					\src{x:\uval{n+1+d;\trgb{\tau\to\tau'}}}
					}{
				\\
				&
				\casefoldeds{\src{x}}{
					\src{\inl{
						\lam{z:\uval{n;\trgb{\tau}}}{
							\downgradepar{n;d;\trgb{\tau'}}{x_1~ (\upgrade{n;d;\trgb{\tau}} z) }
						}
					}}
				}{\src{\inr{x_2}}}
			}
			\end{aligned}
		}
	\\
	\downgrade{n+1;d;\trgb{\matgen{\alpha}{\tau'}}} =
		&\
		\src{
			\begin{aligned}[t]
				&
				\lam{
					\src{x:\uval{n+1+d;\trgb{\matgen{\alpha}{\tau'}}}}
					}{
				\casefoldeds{\src{x}}{
					\src{\inl{
						(\downgrade{n;d;\trgb{\tau'\subt{\matgen{\alpha}{\tau'}}{\alpha}}} x_1)
					}}
				}{\src{\inr{x_2}}}}
			\end{aligned}
		}
\end{align*}
\hrule
  \begin{align*}
      \downgradeic{n;\equi{\tau}} =&\ \text{ as } \downgradefe{n;\equi{\tau}}
  \end{align*}  
\hrule
  \begin{align*}
    \downgradefe{n+1;\equi{\mat}} =&\
      \downgradefe{n+1;\equi{\tau\subo{\mat}{\alpha}}}
    &
      \downgradefe{n;\equi{\tau}} =&\ \text{ as above}
  \end{align*}
  \caption{\label{fig:updn-dn}Definition of the \src{downgrade} function.}
\end{figure}

The second piece of formalism that we need is functions to increase or decrease the approximation level of backtranslated terms.
We exemplify their necessity with an example from \citet{Devriese:2016:FCA:2837614.2837618}.
\begin{example}[The need for \src{downgrade}]\label{ex:down}
  
Consider \stlcim term \iso{\lamt{x:\tat}{\inr{x}}}, intuitively its backtranslation (for a sufficiently-large \com{n}) is: \src{\inl{\lam{x:\uval{n-1;\tat}}{\inl{\inr{x}}}}}
If we try to typecheck this, though, we see that \src{x} has type \uval{n-1;\tat} while it is expected to have type \uval{n-2;\tat}, i.e., its index should be lower.
This concern is about well-typedness, not precision of the backtranslation.
Since \src{x} is inside an \src{\inr{}}, inspecting it for any number of steps requires at least an additional step, to `case' \src{x} out of the \src{\inr{}}.
In other words, for the \src{\inr{}} to be a precise approximation up to \com{n-1} steps, \src{x} needs to only be precise up to \com{n-2} steps.
Thus, it is safe to throw away one level of precision and \emph{downgrade} \src{x} from type \uval{n-1;\tat} to \uval{n-2;\tat}.
\end{example}

However, downgrading is not sufficient, as demonstrated by the next example regarding function types.
\begin{example}[The need for \src{upgrade}]\label{ex:up}
Consider how we can downgrade a value of type $\uval{n+1;\trgb{\tau\to\tau'}}$ to one of type $\uval{n;\trgb{\tau\to\tau'}}$.
We need to convert a function of type $\src{\uval{n+1;\trgb{\tau}}\to}$ $\src{\uval{n+1;\trgb{\tau'}}}$ into one of type $\src{\uval{n;\trgb{\tau}}\to\uval{n;\trgb{\tau'}}}$.
To do this, we need to upgrade the argument value of type $\uval{n;\trgb{\tau}}$ into one of type $\uval{n+1;\trgb{\tau}}$.
Fortunately, this does not mean we need to magically improve the approximation precision of the value concerned.
Type $\uval{n;\trgb{\tau}}$ has an ``error box'' ($\src{\cdots \uplus\Units}$) at every level so we can simply construct the value such that it simply does not use the additional level of precision in $\uval{n;\trgb{\tau}}$.
\end{example}

Finally, another reason we need to upgrade and downgrade a value is that type \uval{n;\tat} must be sufficiently large to contain approximations of target values \emph{up to less than \com{n} steps}.
In fact, for a term to be well-typed the accuracy of the approximation can be less than \com{n}.
In these cases (i.e, for $m < n$), values of type \uval{n;\tat} will be downgraded to type \uval{m;\tat}. 
Dually, there will be cases where some values need to be upgraded.

Functions \upgrade{\cdot} and \downgrade{\cdot} perform what we just discussed; their types and formalisation is presented in \Cref{fig:updn-up,fig:updn-dn}.
Their definition closely follows the structure of the type approximations $\uval{n;\trgb{\tau}}$ and essentially just transfers an approximated value to the corresponding value in a deeper or shallower approximation of type $\trgb{\tau}$.
The cases for \Units and \Bools are optimised based on the fact that $\uval{n;\Unitt}=\uval{m;\Unitt}$ (resp. $\uval{n;\Boolt}=\uval{m;\Boolt}$) so long as $n,m>0$.
As mentioned, downgrade `forgets' information about the approximation, effectively dropping \com{1} level of precision in the backtranslation.
Dually, upgrade adds \com{1} level of information in the approximation.
Adding this information is, however, not precise, because those additional levels are unknown (\src{\unk}).
Effectively, while \src{\downgrade{n;\tat}(\upgrade{n;\tat}\ t)} reduces to \src{t}, term \src{\upgrade{n;\tat}(\downgrade{n;\tat}\ t)} does not reduce to \src{t} because information was lost (\Cref{ex:updn-forg}).
\begin{example}[Upgrading after downgrading forgets information]\label{ex:updn-forg}
  Consider the following term: $\src{\downgrade{0;\Boolt}\ \inl{\trues}}$, which reduces to \src{\units}.
  If we apply \src{\upgrade{0;\Boolt}} to it, we do not obtain back \src{\inl{\trues}} but \src{\unk}, which is \src{\inr{\units}}.
  That is because downgrade forgets the shape of the value it received (\src{\inl{\trues}}) and upgrade cannot possibly recover that information.
\end{example}

Finally, we need to define these functions for the other backtranslations that rely on the other backtranslation types \uvalfe{} and \uvalic{}.
As mentioned, the main difference between these last two backtranslation types and \uval{} is the case for target recursive types.
Recall that these last two backtranslation types for recursive types perform the unfolding of the type without decrementing the index.
This affects these functions too: upgrading or downgrading a term at a recursive type is like upgrading or downgrading at the unfolding of that type but at the same index.

In the backtranslation, we generally use creation of a backtranslated value together with a \downgrade{}, while we use destruction of backtranslated values together with an \upgrade{}.
Thus, we provide compacted functions that do exactly this, $\indn{n;\tat}$ and \caseup{n;\tat} (\Cref{fig:compact}).
Note that the arguments to the first function is not ill-typeset: they indeed take a parameter whose type is the \com{\inl{}} projection of type \uval{n;\Unitt}.
As for the previous helpers, the compacted versions that operate on terms of type \uvalfe{n;\equi{\mat}} (and \uvalic{n;\equi{\mat}}) are different.
Since there is no destructor for \uvalfe{n;\equi{\mat}}, there also is no need for a compacted version.
\begin{figure}[!t]
\small
\mytoprule{\indn{n;\tat} \quad \text{ and } \quad \caseup{n;\tat}}

\begin{gather*}
\begin{aligned}
  \indn{n;\trg{ \Unitt }}
    =&\
    \src{
      \lam{x: \Units }{
        \downgradepar{n;\trg{\Unitt}}{\inl{x}}
      }
    }
  &
  \indn{n;\trg{ \Boolt }}
    =&\
    \src{
      \lam{x: \Bools }{
        \downgradepar{n;\trg{\Boolt}}{\inl{x}}
      }
    }
  \\
  \indn{n;\trgb{\tau\to\tau' }}
    =&\
    \src{
      \begin{aligned}
      &
      \lam{\src{x: \uval{n;\tat}\to\uval{n;\trgb{\tau'}} }}{
      \\
      &\
        \src{\downgradepar{n;\trgb{\tau\to\tau'}}{\inl{x}}}
      }
      \end{aligned}
    }
  &
  \indn{n;\trgb{\tau\times\tau' }}
    =&\
    \src{
      \begin{aligned}
      &
      \lam{\src{x: \uval{n;\tat}\times\uval{n;\trgb{\tau'}} }}{
        \\
        &\
        \src{\downgradepar{n;\trgb{\tau\times\tau'}}{\inl{x}}}
      }
      \end{aligned}
    }
  \\
  \indn{n;\trgb{\tau\uplus\tau' }}
    =&\
    \src{
      \begin{aligned}
        &
        \lam{\src{x: \uval{n;\tat}\uplus\uval{n;\trgb{\tau'}} }}{
          \\
          &\
          \src{\downgradepar{n;\trgb{\tau\uplus\tau'}}{\inl{x}}}
        }
      \end{aligned}
    }
  &
  \indn{n;\trgb{\mat }}
    =&\
    \src{
      \begin{aligned}
        &
      \lam{\src{x: \uval{n;\trg{\tat\subt{\matt}{\alpt}}} }}{
        \\
        &\
        \src{\downgradepar{n;\trgb{\mat}}{\inl{x}}}
      }
      \end{aligned}
    }
  \end{aligned}
  \\
  \begin{aligned}
  \caseup{n;\trgb{\tat }}
    =&\
    \src{
      \lam{x:\uval{n;\trgb{\tat}}}{
        \casetagpar{ n;\trgb{\tat} }{\upgradepar{n;\trgb{\tat}}{x}}
      }
    }
\end{aligned}
\end{gather*}
\hrule
\begin{align*}
  \indnic{n;\equi{\tau }}
    &\
    \text{ and }
  \caseupic{n;\equi{\tau }}
    =
    \
    \text{ as above, without a case for }\equi{\tau}=\equi{\mat}
\end{align*}
\hrule
\begin{align*}
  \indnfe{n;\equi{\tau }}
    &\
    \text{ and }
  \caseupfe{n;\equi{\tau }}
    =
    \
    \text{ as above, without a case for }\equi{\tau}=\equi{\mat}
\end{align*}
\caption{\label{fig:compact}Compacted functions used to manipulate backtranslated values.}
\end{figure}

\smallskip

At this point we may ask ourselves: how can we reason about these functions, as well as about backtranslated terms?
This is what we explain next.

\subsection{Relating Backtranslated Terms}\label{sec:rel-bt}
If we were to use the logical relations of \Cref{fig:logrel-main} to relate a term and its backtranslation, this would simply not work.
Consider \stlcim type \iso{\Unit}, that is backtranslated (at any approximation \com{n>0}) into \uval{n;\Unitt}, i.e., \src{\Units \uplus \Units}.
Value \iso{\unitt} should normally be backtranslated to \src{\inl{\units}}.
Following the value relation in \lrfi for \src{\uplus} types, both terms need to have an \com{\inl{}} tag, so this does not work.
More importantly, it \emph{should not} work: we are not relating terms of \src{\uplus} type, we are relating backtranslated terms, where the backtranslation performs a modification on the type (and thus the term) by inserting the \com{\inl{}}.

\begin{figure}[!ht]
  \small
    \begin{align*}
    &
    \valrel{\emuldv{0;\imprecise;\tat}} \isdef
    \
      \myset{ (\W,\src{v},\trg{v}) }{ \src{v}=\units }
    \qquad\qquad\qquad
    \valrel{\emuldv{0;\precise;\tat}} \isdef
    \
      \emptyset
    \\
    &
    \valrel{\emuldv{n+1;p;\tat}} \isdef \{ (\W,\src{v},\trg{v}) ~|~ \src{v}\in\oftypes{\emuldv{n+1;p;\tat}} \text{ and } \trg{v}\in\oftypet{\tat}\text{ and }
      \\
      &\ \ \
        \left.\begin{aligned}
          \text{either }
          &
          \cdot
            \src{v}=\src{\inr{\units}} \text{ and } \src{p}=\src{imprecise}
          \\
          \text{or }
          &\cdot
          \begin{cases}
            \cdot
              &
              \tat=\Unitt \text{ and }
              \exists\src{v'}.~ \src{v}=\src{\inl{v'}} \text{ and } 
              (\W,\src{v'},\trg{v})\in\valrel{\Units}
            \\
            \cdot
              &
              \tat=\Boolt \text{ and }
              \exists\src{v'}.~ \src{v}=\src{\inl{v'}} \text{ and } 
              (\W,\src{v'},\trg{v})\in\valrel{\Bools}
            \\
            \cdot
              &
              \begin{aligned}[t]
                &
                \tat=\trgb{\tau_1\to\tau_2} \text{ and }
                \exists\src{v'}.~ \src{v}=\src{\inl{v'}} \text{ and } 
                (\W,\src{v'},\trg{v})\in\valrel{\emuldv{n;p;\trgb{\tau_1}}\to\emuldv{n;p;\trgb{\tau_2}}}
              \end{aligned}
            \\
            \cdot
              &
              \begin{aligned}[t]
                &
                \tat=\trgb{\tau_1\times\tau_2} \text{ and }
                \exists\src{v'}.~ \src{v}=\src{\inl{v'}} \text{ and } 
                (\W,\src{v'},\trg{v})\in\valrel{\emuldv{n;p;\trgb{\tau_1}}\times\emuldv{n;p;\trgb{\tau_2}}}
              \end{aligned}
            \\
            \cdot
              &
              \begin{aligned}[t]
                &
                \tat=\trgb{\tau_1\uplus\tau_2} \text{ and }
                \exists\src{v'}.~ \src{v}=\src{\inl{v'}} \text{ and } 
                (\W,\src{v'},\trg{v})\in\valrel{\emuldv{n;p;\trgb{\tau_1}}\uplus\emuldv{n;p;\trgb{\tau_2}}}
              \end{aligned}
            \\
            \cdot
              &
              \begin{aligned}[t]
                &
                \tat=\trgb{\mat} \text{ and }
                \exists\src{v'}.~ \src{v}=\src{\inl{v'}} \text{ and } 
                \\
                &\
                \exists\trg{v'}.~ \trg{v}=\trg{\fold{\matt}~v'}
                (\W,\src{v'},\trg{v'})\in\later\valrel{\emuldv{n;p;\trg{\tat\subt{\matt}{\alpt}}}}
              \end{aligned}
          \end{cases}
        \end{aligned}\hspace*{-10pt}\right\}
  \end{align*}  
  \hrule
  \begin{align*}
    \valrel{\emuldvic{n;p;\equi{\tau}}} &\text{ is defined analogously to  } \valrel{\emuldvfe{n;p;\equi{\tau}}}
  \end{align*}
  \hrule
  \begin{align*}
    &
    \valrel{\emuldvfe{0;\imprecise;\oth{\tau}}} \isdef
    \
      \myset{ (\W,\src{v},\oth{v}) }{ \src{v}=\units }
    \qquad\qquad\qquad
    \valrel{\emuldvfe{0;\precise;\oth{\tau}}} \isdef
    \
      \emptyset
    \\
    &
    \valrel{\emuldvfe{n+1;p;\oth{\tau}}} \isdef \{ (\W,\src{v},\equi{v}) ~|~ \src{v}\in\oftypes{\emuldvfe{n+1;p;\oth{\tau}}} \text{ and } \equi{v}\in\oftypeo{\tau}\text{ and }
    \\
      &\
      \left.\begin{aligned}
        \text{either }
        &
            \cdot
              \src{v}=\src{\inr{\units}} \text{ and } \src{p}=\src{imprecise}
            \\
          \text{or }
            &
            \cdot
            \begin{cases}
              \cdot
              &
              \text{omitted parts are as above}
              \\
              \cdot
              &
              \oth{\tau}=\oth{\mat} \text{ and } \equi{\tau} \text{ contractive in }\equi{\alpha} \text{ and }
              (\W,\src{v},\oth{v})\in\valrel{\emuldv{n+1;p;\oth{\tau\subo{\mat}{\alpha}}}}
        \end{cases}
      \end{aligned}\right\}
  \end{align*}
  \caption{
    Missing bits of the logical relation: value relation for backtranslation type (excerpts).
    Note that \src{p} can be either \src{\precise} or \src{\imprecise} in the second clause (the 'or') of the \src{n+1} case.
    }
  \label{fig:emuldv1}
\end{figure}
This is the reason we have pseudotypes and, in particular, the reason we have \com{\emuldvtext}.
We have three \com{\emuldvtext}s---one per backtranslation---and each follows the same intuition, which we explain starting with \emuldv{n;p;\tat}, the type of backtranslated \stlcim terms into \stlcf (top of \Cref{fig:emuldv1}).
\src{\emuldv{n;p;\tat}} is indexed by a non-negative number \src{n}, a value $\src{p} ::= \precise \mid \imprecise$ and the original target type \tat.
The number tracks the depth of type that are being related, index \src{p} tracks the precision of the approximation (as explained below) and the original type carries precise information of the type to expect in the backtranslation.
As seen, sometimes we have \src{\unk} values (i.e., \src{\inr{\units}}) in the backtranslation, the intuition behind their meaning is presented in \Cref{ex:unk}
\begin{example}[Approximate values \src{\unk}]\label{ex:unk}
Consider the $\trg{\uval{6,\Boolt}}$ value: \src{\inl{~\left\langle\inl{~(\inl{\unk_4})},\!\!\right.}}$ \linebreak $\src{\left.\unk_5\right\rangle}.
This value might be used by the approximate back-translation to represent the term $\trg{\pair{\inl{\pair{\unitt,\truet}},\lamt{x:\Boolt}{x}}}$.
Our $\valrel{\emuldv{\cdot}}_\square$ specification will enforce that terms of the form $\src{\inl{\pair{\cdot,\cdot}}}$ or $\src{\inl{(\inl{\cdot}})}$ represent the corresponding target constructs, but terms $\src{\unk_{4}}$ and $\src{\unk_{5}}$ can represent arbitrary terms (in this case: a pair of base values and a lambda).
\end{example}
Thus, $\valrel{\emuldv{n;p;\tat}}$ regulates how these \src{\unk} values occur depending on the precision index.
$\src{p} = \imprecise$ will only be used in the $\lesssim$ direction of the approximation, i.e., we have that source termination in \emph{any} number of steps implies target termination.
Here, $\valrel{\emuldv{n;p;\tat}}$ allows $\src{\unk}$ values to occur anywhere in a backtranslated term, and they can correspond to arbitrary target terms. 
These constraints are simple to enforce because with $\lesssim$ we can achieve this by making backtranslated terms diverge whenever they try to use a \src{\unk} value. 
This is sufficient because the $\lesssim$ approximation trivially holds when the source term diverges.

On the other hand, $\src{p} = \precise$ will be used for the other direction of approximation: $\gtrsim$.
Recall that for this direction, termination of target terms in less than $n$ steps implies termination of source terms. 
In this case, the requirements on backtranslated terms are stronger: \src{\unk} is ruled out by the definition of $\valrel{\emuldv{n;p;\tat}}$ within depth \com{n}, i.e., we cannot reach \src{\unk} in the steps of the world.
\begin{example}[Relatedness with \imprecise]\label{ex:imprecise}
Consider the term $\src{t}\equiv\src{\inl{\pair{\unk_{42},\unk_{42}}}}$.
This term will be related to $\trg{\pair{t_1,t_2}}$ at pseudo-type $\src{\emuldv{43;\imprecise;\trgb{\tau_1\times\tau_2}}}$ for any terms $\trg{t_1}$ and $\trg{t_2}$ and in any world.
\end{example}

\begin{example}[Relatedness with \precise]\label{ex:precise}
Consider again the term $\src{t} \isdef \src{\inl{\pair{\unk_{42},\unk_{42}}}}$.
This term will still be related by $\src{\emuldv{43;\precise;\trgb{\tau\times\tau'}}}$ to $\trg{t} \isdef \trg{\pair{t_1,t_2}}$ for any terms $\trg{t_1}$ and $\trg{t_2}$, but only in worlds $\W$ such that $\stepsfun{\W} = 0$. 
More precisely, our specification will state that $(\W,\src{t},\trg{t}) \in \valrel{\emuldv{43;\precise;\trgb{\tau_1\times\tau_2}}}$ iff
\begin{equation*}
  (\W,\src{\pair{\unk_{42},\unk_{42}}},\trg{\pair{t_1,t_2}}) \in \valrel{\emuldv{42;\precise;\trgb{\tau_1}}\times \emuldv{42;\precise;\trgb{\tau_2}}}
\end{equation*}
By the definition of the logical relation, this requires in turn that $(\W,\src{\unk_{42}},\trg{t_1})$ and $(\W,\src{\unk_{42}},\trg{t_2})$ are in $\later\valrel{\emuldv{42;\precise;\trgb{\tau_1}}}$ and in $\later\valrel{\emuldv{42;\precise;\trgb{\tau_2}}}$ respectively.
However if $\stepsfun{\W} = 0$, then this is vacuously true by definition of the $\later$ operator, independent of the requirements of $\valrel{\emuldv{42;\precise;\cdot}}$.
\end{example}

The pseudotype for the \stlcem to \stlcf backtranslation (\emuldvfe{\cdot}) follows the same pattern as \uvalfe{\cdot}: it does not lose a step in the \equi{\mat} case (\Cref{fig:emuldv1}). 
At a cursory glance, it appears that a non-contractive \equi{\mat} ruins the well-foundedness of our induction as without decrementing our step index, a non-contractive type seems to infinitely recurse under this definition. 
Fortunately, however, the condition $\equi{v}\in\oftypeo{\tau}$, which with the fact that no values exist of non-contractive types prevents this concern from arising.
As before, the pseudotype for the \stlcem to \stlcim backtranslation (\emuldvic{\cdot}) follows the same approach as \emuldvfe{\cdot}.

\begin{figure}[!t]
	
  \begin{gather*}
\small
\begin{aligned}
  \emtotau{\emuldv{n;p;\tat}} =
    &\
    \uval{n;\tat}
  &
  \emtotau{\src{\psd{\tau_1}\to\psd{\tau_2}}} =
    &\
    \src{\emtotau{\src{\psd{\tau_1}}}\to\emtotau{\psd{\tau_2}}}
  \\
  \emtotau{\src{\Bools}} =
    &\
    \src{\Bools}
    &
    \emtotau{\src{\psd{\tau_1}\times\psd{\tau_2}}} =
      &\
      \src{\emtotau{\src{\psd{\tau_1}}}\times\emtotau{\psd{\tau_2}}}
  \\
  \emtotau{\src{\Units}} =
    &\
    \src{\Units}
    &
    \emtotau{\src{\psd{\tau_1}\uplus\psd{\tau_2}}} =
      &\
      \src{\emtotau{\src{\psd{\tau_1}}}\uplus\emtotau{\psd{\tau_2}}}
  \\ 
  \srctotrgty{\emuldv{n;p;\tat}} =
    &\
    \tat
  &
  \srctotrgty{\src{\psd{\tau_1}\to\psd{\tau_2}}} =
    &\
    \trg{
        \srctotrgty{\psd{\tau_1}}
        \tot
        \srctotrgty{\psd{\tau_2}}
    }
  \\
  \srctotrgty{\Units} =
    &\
    \Unitt
    &
    \srctotrgty{\src{\psd{\tau_1}\times\psd{\tau_2}}} =
      &\
      \trg{
          \srctotrgty{\psd{\tau_1}}
          \timest
          \srctotrgty{\psd{\tau_2}}
      }
  \\
  \srctotrgty{\Bools} =
    &\
    \Boolt
  &
  \srctotrgty{\src{\psd{\tau_1}\uplus\psd{\tau_2}}} =
    &\
    \trg{
        \srctotrgty{\psd{\tau_1}}
        \uplust
        \srctotrgty{\psd{\tau_2}}
    }
\end{aligned}
  \\
  \small
\begin{aligned}
  \emtotaufe{\emuldv{n;p;\equi{\tau}}} =
    &\
    \uvalfe{n;\equi{\tau}} 
    &&&
	\emtotaufe{\cdots} =
		&\
		\text{ as the other cases for }\emtotau{\cdot}
	\\
    \emtotauic{\emuldvic{n,\equi{\tau}}} =
	    &\
	    \uvalic{n,\equi{\tau}}
	&&&
		\emtotauic{\cdots} =
		    &\
		    \text{ as the other cases for }\emtotau{\cdot}
    \\
  	\srctoothty{\emuldv{n;p;\equi{\tau}}} =
  		&\ 
  		\equi{\tau}
    &&&
    \srctoothty{\cdots} =
    	&\
    	\text{ as the other cases for }\srctotrgty{\cdot}
    \\
    \srctotrgtyic{\emuldvic{n;p;\equi{\tau}}} =
    	&\
    	\equi{\tau}
    &&&
    	\srctotrgtyic{\cdots}
    	&\
    	\text{ as the other cases for }\srctotrgty{\cdot}
\end{aligned}
\end{gather*} 
\caption{%
    Missing auxiliary functions of the logical relation.
    }
  \label{fig:emuldv-aux}
\end{figure}
Finally, we can define function \emtotau{\cdot} that translate from source pseudo-types into plain source types and function \srctotrgty{\cdot}, that translates source pseudotypes into target types (\Cref{fig:emuldv-aux}).
As expected, these functions exists for all backtranslations and they follow the same pattern presented here; for the sake of brevity, we only report the names and types of the omitted ones.

\section{The Three Compilers and Their Backtranslations}\label{sec:comp} 
Our compilers (\Cref{sec:comp-def}) and backtranslations (\Cref{sec:backtr-def}) translate between languages as depicted in \Cref{fig:dia}.
After showing their formalisation and proving that they relate terms cross-language, this section proves the compilers are fully abstract (\Cref{sec:proofs-fa}).

\subsection{Compilers and Reflection of Fully-Abstract Compilation}\label{sec:comp-def}
The compilers (\Cref{fig:comp}) are all mostly homomorphic apart from what we describe below.
We overload the compilation notation and express the compiler for types and terms in the same way (we omit the compiler for types since it is the identity).
Compiler \compstlcfi{\cdot} translates \src{\fix{\cdot}} into the Z-combinator annotated with \iso{\fold{}} and \iso{\unfold{}} for \stlcim.
We cannot use the Y combinator since it does not work in call-by-value~\citep{max-embed,popl-journal}, but fortunately the Z-combinator does~\cite[Sec. 5]{bookpierce}.
Compiler \compstlcic{\cdot} erases \iso{\fold{}} and \iso{\unfold{}} annotations since \stlcem does not have them.
Compiler \compstlce{\cdot} is just the composition of the previous two.
\begin{figure}[!t]
\begin{center}
\small
\mytoprule{ \compstlcfi{\cdot} : \src{t} \to \iso{t} \quad \text{ and } \quad \compstlcie{\cdot} : \iso{t} \to \equi{t} \quad \text{ and } \quad \compstlcfe{\cdot} : \src{t} \to \equi{t}}

\begin{gather*}
\begin{aligned}
  \compstlcfi{\units} 
    =&\ 
    \unitt
      &
      \compstlcfi{\src{\lam{x:\tau}{t}}} 
        =&\ 
        \trg{\lamt{x:\compstlcfi{\tau}}{\compstlcfi{t}}}
          &
          \compstlcfi{\src{\projone{t}}} 
            =&\ 
            \trg{\projone{\compstlcfi{t}}}
              &
              \compstlcfi{\src{x}} 
                =&\ 
                \trg{x}           
  \\
  \compstlcfi{\trues} 
    =&\ 
    \truet
      &
      \compstlcfi{\src{t~t'}} 
        =&\ 
        \trg{\compstlcfi{t}~\compstlcfi{t'}}
          &
          \compstlcfi{\src{\projtwo{t}}} 
            =&\ 
            \trg{\projtwo{\compstlcfi{t}}}
              &
              \compstlcfi{\src{\inl{t}}} 
                =&\ 
                \trg{\inl{\compstlcfi{t}}}
  \\
  \compstlcfi{\falses} 
    =&\ 
    \falset
      &
      \compstlcfi{\src{\pair{t,t'}}} 
        =&\ 
        \trg{\pair{\compstlcfi{t},\compstlcfi{t'}}}
          &
          \compstlcfi{\src{t;t'}} 
            =&\ 
            \trg{\compstlcfi{t};\compstlcfi{t'}}
            &
            \compstlcfi{\src{\inr{t}}} 
              =&\ 
              \trg{\inr{\compstlcfi{t}}}
\end{aligned}
\\
\begin{aligned}
  \compstlcfi{  
    \src{\caseof{t}{\src{t'}}{\src{t''}}}
  } 
    =&\ 
    \trg{
      \caseof{\compstlcfi{t}}{\compstlcfi{t'}}{\compstlcfi{t''}}
    }
  \\
  \compstlcfi{\src{\iftes{t}{t'}{t''}}} 
    =&\ 
    \trg{\iftet{\compstlcfi{t}}{\compstlcfi{t'}}{\compstlcfi{t''}}}
\end{aligned}
\\
\begin{aligned}
  \compstlcfi{\src{\fix{\tau_1\to\tau_2} t}} 
    =&\ \\
    & \hspace*{-40pt}
    \trg{ 
      \left(\begin{aligned}
        &
        \lamt{\trg{f:\compstlcfi{(\tau_1\to\tau_2)\to\tau_1\to\tau_2}}}{ 
        \\
        &\
          \left(\begin{aligned}
            &
            \lamt{\trg{x:\trg{\matgent{\alpt}{\alpt\to\unboldmath\compstlcfi{\tau_1\to\tau_2}}}}}{ 
            \trg{f~(\lamt{y:\compstlcfi{\tau_1}}{( (\unfold{\matgent{\alpt}{\alpt\to\compstlcfi{\tau_1\to\tau_2}}}~x)~x )~y} } )}
          \end{aligned}\right)
        \\
        &\
          \trg{\fold{\matgent{\alpt}{\alpt\to\compstlcfi{\tau_1\to\tau_2}}}}~
        \\
        &\ \ 
          \left(\begin{aligned}
            &
            \lamt{\trg{x:\trg{\matgent{\alpt}{\alpt\to\compstlcfi{\tau_1\to\tau_2}}}}}{ 
            \trg{f~(\lamt{y:\compstlcfi{\tau_1}}{( (\unfold{\matgent{\alpt}{\alpt\to\compstlcfi{\tau_1\to\tau_2}}}~x)~x )~y} } )}
          \end{aligned}\right)
        }   
      \end{aligned}\right ) \compstlcfi{t}  
    }
\end{aligned}
\end{gather*}
\hrule
\begin{align*}
	\compstlcic{\cdots} =
		&
		\begin{aligned}
		 	&
		 	\text{omitted rules are}
		 	\\
		 	&
		 	\text{ as above}
		 \end{aligned} 
  &
  \compstlcic{\trg{\fold{\matt}~t}} 
    =&\ 
    \oth{\compstlcic{t}}
  &
  \compstlcic{\trg{\unfold{\matt}~t}} 
    =&\ 
    \oth{\compstlcic{t}}
\end{align*}
\hrule
\begin{align*}
  \compstlce{ t } = \compstlcic{\compstlcfi{ t } }, \text{ i.e., as above, without \iso{fold}/\iso{unfold} annotations in the compilation of \src{fix}} 
\end{align*}
\vspace{-1em}
\end{center}
\caption{Definition of our compilers.\label{fig:comp}}
\end{figure}

Correctness of the compilation (\Cref{thm:compstlc-sem-pres,thm:ic-compstlc-sem-pres,thm:compstlce-sem-pres-ic} below) is proven via a series of standard compatibility lemmas (\Cref{thm:compat-lem-lam}, we report just the case for lambda since the others follow the same structure).
These, in turn, rely on a series of standard results for these kinds of logical relations such as the fact that related terms plugged in related contexts are still related and antireduction (i.e., if two terms step to related terms, then they are themselves related).
\begin{lemma}[Compatibility for $\lambda$]\label{thm:compat-lem-lam}
    \[
    \text{if }
      \src{\Gamma},\src{x:\tau'} \vdash \src{t} \anylogreln{n} \trg{t} : \src{\tau}
    \text{ then }
      \src{\Gamma} \vdash \src{\lam{x:\tau'}{t}} \anylogreln{n} \trg{\lamt{x:\tat'}{t}} : \src{\tau'\to\tau}
    \]
\end{lemma}
\begin{lemma}[\compstlc{\cdot} is semantics preserving]\label{thm:compstlc-sem-pres}
    \[
    \text{if }
    \src{\Gamma}\vdash\src{t}:\src{\tau}
    \text{ then }
    \src{\Gamma}\vdash\src{t}\anylogreln{n}\compstlc{\src{t}}:\src{\tau}
    \]
\end{lemma}
\begin{lemma}[\compstlcic{\cdot} is semantics preserving]\label{thm:compstlce-sem-pres-ic}
    \[
    \text{if }
    \trgb{\Gamma}\vdash\trg{t}:\trg{\tat}
    \text{ then }
    \trgb{\Gamma}\vdash\trg{t}\anylogreln{n}\compstlcic{\trg{t}}:\trg{\tat}
    \]
\end{lemma}
\begin{lemma}[\compstlcfe{\cdot} is semantics preserving]\label{thm:ic-compstlc-sem-pres}
    \[
    \text{if }
    \src{\Gamma}\vdash\src{t}:\src{\tau}
    \text{ then }
    \src{\Gamma}\vdash\src{t}\anylogreln{n}\compstlcfe{\src{t}}:\src{\tau}
    \]
\end{lemma}

Since fully-abstract compilation requires reasoning about program contexts, we extend the compiler to operate on them too.
This follows the same structure of the compilers above and therefore we omit this definition.
Correctness of the compiler scales to contexts too (\Cref{thm:compstlc-sem-pres-ctx,thm:compstlc-sem-pres-ctx-ie,thm:compstlc-sem-pres-ctx-fe}).
\begin{lemma}[\compstlc{\cdot} is semantics preserving for contexts]\label{thm:compstlc-sem-pres-ctx}\hfill

  \[
      \text{if }
    \vdash\ctxs:\src{\Gamma,\tau}\to\src{\Gamma',\tau'}
    \text{ then }
    \vdash\ctxs\anylogreln{n}\compstlc{\ctxs}:\src{\Gamma,\tau}\to\src{\Gamma',\tau'}
  \]
\end{lemma}

\begin{lemma}[\compstlcie{\cdot} is semantics preserving for contexts]\label{thm:compstlc-sem-pres-ctx-ie}\hfill

  \[
      \text{if }
    \vdash\ctxt:\iso{\Gamma,\tat}\to\iso{\Gamma',\tat'}
    \text{ then }
    \vdash\ctxt\anylogreln{n}\compstlcie{\ctxt}:\iso{\Gamma,\tat}\to\iso{\Gamma',\tat'}
  \]
\end{lemma}

\begin{lemma}[\compstlcfe{\cdot} is semantics preserving for contexts]\label{thm:compstlc-sem-pres-ctx-fe}\hfill

  \[
      \text{if }
    \vdash\ctxs:\src{\Gamma,\tau}\to\src{\Gamma',\tau'}
    \text{ then }
    \vdash\ctxs\anylogreln{n}\compstlcfe{\ctxs}:\src{\Gamma,\tau}\to\src{\Gamma',\tau'}
  \]
\end{lemma}

With these results, we can already prove the reflection direction of fully-abstract compilation (\Cref{thm:stlc-refl,thm:stlc-refl-fe,thm:stlc-refl-ic}).
The proof follows the structure depicted in the left part of \Cref{fig:fa-refl}.
\begin{theorem}[\compstlc{\cdot} reflects equivalence]\label{thm:stlc-refl}
  \[
    \text{If }
        \trge \vdash \compstlc{t_1} \ceqt \compstlc{t_2} : \compstlc{\tau}
      \text{ then }
        \srce \vdash \src{t_1} \ceqs \src{t_2} : \src{\tau}
  \]
\end{theorem}
\begin{theorem}[\compstlcic{\cdot} reflects equivalence]\label{thm:stlc-refl-ic}
  \[
    \text{If }
        \othe \vdash \compstlcic{t_1} \ceqo \compstlcic{t_2} : \compstlcic{\tau}
      \text{ then }
        \trge \vdash \trg{t_1} \ceqt \trg{t_2} : \trg{\tat}
    \]
\end{theorem}
\begin{theorem}[\compstlcfe{\cdot} reflects equivalence]\label{thm:stlc-refl-fe}
  \[
    \text{If }
        \othe \vdash \compstlcfe{t_1} \ceqo \compstlcfe{t_2} : \compstlcfe{\tau}
      \text{ then }
        \srce \vdash \src{t_1} \ceqs \src{t_2} : \src{\tau}
     \]
\end{theorem}
Since this last compiler is the composition of the other two, the proof of \Cref{thm:stlc-refl-fe} trivially follows from composing the proofs of the other two compilers.

\subsection{Backtranslations and Preservation of Fully-Abstract Compilation}\label{sec:backtr-def}
Function \emulate{}{\cdot} is responsible for translating a target term of type \tat\ into a source one of type \uval{n;\tat} (\Cref{sec:emul}) by relying on the machinery needed for working with \uval{} terms from \Cref{sec:bt-type}.
This function is easily extended to work with program contexts, producing contexts with hole of type \uval{n;\tat}.
However, recall that the goal of the backtranslation is generating a source context whose hole can be filled with source terms $\src{t_1}$ and $\src{t_2}$ and their type is not \uval{n;\tat} but \src{\tau}.
Thus, there is a mismatch between the type of the hole of the emulated context and that of the terms to be plugged there.
Since emulated contexts work with \uval{} values, we need a function that wraps terms of an arbitrary type \src{\tau} into a value of type \uval{n;\tat}.
This function is called \inject{} (\Cref{sec:injext}) and it is the last addition we need before the backtranslations (\Cref{sec:backtranslations}).

\subsubsection{Emulation of Terms and Contexts}\label{sec:emul}
Like the compiler, the emulation must not just operate on types and terms, but also on program contexts.
Unlike the compiler, the emulation operates on \emph{type derivations} for terms and contexts since all our target languages are typed.
Thus, the emulation of a lambda would look like the following (using \iso{D} as a metavariable to range over derivations and omitting functions to work with \uval{}).
\begin{align*}
  \emulate{}{
    \AxiomC{\iso{D}}
    \UnaryInfC{$\iso{\Gat,x:\tat}\vdash\iso{t}:\iso{\tat'}$}
    \UnaryInfC{$\iso{\Gat}\vdash\iso{\lamt{x:\tat}{t}}:\iso{\tat\to\tat'}$}
    \DisplayProof
  } 
  =&\
  \src{\lam{\src{x:}\uval{n;\tat}}{
    \emulate{}{
      \AxiomC{\iso{D}}
      \UnaryInfC{$\iso{\Gat,x:\tat}\vdash\iso{t}:\iso{\tat'}$}
      \DisplayProof
    }
  }}
\end{align*}
However, note that each judgement uniquely identifies which typing rule is being applied and the underlying derivation.
Thus, for compactness, we only write the judgement in the emulation and implicitly apply the related typing rule to obtain the underlying judgements for recursive calls.

\begin{figure}[!ht]
\small
\mytoprule{
  \emulate{n}{\cdot} : 
  \trg{\Gat}\vdash\trg{t}:\trgb{\tau} \to \src{t} %
}

\begin{gather*}
\begin{aligned}
  \emulate{n}{
      \trg{\Gat}\vdash
      \unitt
      : \Unitt
  } \isdef
    &\
      \src{\indn{n;\Unitt}\units}
  &&&
  \emulate{n}{
      \trg{\Gat}\vdash
      \truet
      : \Boolt
  } \isdef
    &\
      \src{\indn{n;\Boolt}\trues}
  \\
  \emulate{n}{
      \trg{\Gat}\vdash
      \falset
      : \Boolt
  } \isdef
    &\
      \src{\indn{n;\Boolt}\falses}
  &&&
  \emulate{n}{
      \trg{\Gat}\vdash
      \trg{x}
      \trgb{: \tau}
  } \isdef
    &\
    \src{x}
\end{aligned}
\\
\begin{aligned}
  \emulate{n}{
      \trg{\Gat}\vdash
      \trg{\lamt{x:\trgb{\tau}}{t}}
      \trgb{: \tau\to\tau'}
  } \isdef
    &\
      \src{
          \indn{n;\trgb{\tau\to\tau'}}
            \left(\lam{\src{x:\uval{n;\trgb{\tau}}}}{
              \emulate{n}{
                  \trg{\Gat,x:\trgb{\tau}}\vdash
                  \trg{t}
                  \trgb{: \tau'}
              }
            }\right)
      }
  \\
  \emulate{n}{
      \trg{\Gat}\vdash
      \trg{t~t'}
      \trgb{: \tau}
  } \isdef
  \hspace*{10pt}&\hspace*{-10pt}\
      \src{
          \left( 
            \caseup{n;\trgb{\tau'\to\tau}}
            \emulate{n}{
                \trg{\Gat}\vdash
                \trg{t}
                \trgb{: \tau'\to\tau}
            } 
          \right)
          \left(\emulate{n}{
            \trg{\Gat}\vdash
            \trg{t': \trgb{\tau'}}
          }\right)
      }
  \\
  \emulate{n}{
      \trg{\Gat}\vdash
      \trg{\pair{t,t'}}
      \trgb{: \tau\times\tau'}
  } \isdef
    &\
    \src{
        \indn{n;\trgb{\tau\times\tau'}}
        \pair{
          \emulate{n}{
            \trg{\Gat}\vdash
            \trg{t: \trgb{\tau}}
          }
        ,
          \emulate{n}{
            \trg{\Gat}\vdash
            \trg{t': \trgb{\tau'}}
          }
        }
    }
  \\
  \emulate{n}{
      \trg{\Gat}\vdash
      \trg{\projone{t}}
      \trgb{: \tau}
  } \isdef
    &\
      \src{
        \projone{\left(
          \caseup{n;\trgb{\tau\times\tau'}}
          \emulate{n}{
            \trg{\Gat}\vdash
            \trg{t: \trgb{\tau\times\tau'}}
          }
        \right)}
      }
  \\
  \emulate{n}{
      \trg{\Gat}\vdash
      \trg{\projtwo{t}}
      \trgb{: \tau}
  } \isdef
    &\
      \src{
        \projtwo{\left(
          \caseup{n;\trgb{\tau'\times\tau}}
          \emulate{n}{
            \trg{\Gat}\vdash
            \trg{t: \trgb{\tau'\times\tau}}
          }
        \right)}
      }
  \\
  \emulate{n}{
    \begin{aligned}
      &
      \trg{\Gat}\vdash
      \trg{\case~t~\of}
      \\
      &\ 
      \left|\begin{aligned}
        &
        \trg{\inl{x_1}\mapsto t'}
        \\
        &
        \trg{\inr{x_2}\mapsto t''}
      \end{aligned}\right.
      \trgb{: \tau}
    \end{aligned}
  } 
  \isdef 
    &\
      \src{
      \begin{aligned}
        &
        \src{\case}
        \left(
          \caseup{n;\trgb{\tau_1\uplus\tau_2}}
          \emulate{n}{
            \trg{\Gat}\vdash
            \trg{t: \trgb{\tau_1\uplus\tau_2}}
          }
        \right)
      \\&\
        \src{\of\left|\begin{aligned}
          &
          \src{\inl{x_1}\mapsto}
          \emulate{n}{
            \trg{\Gat,(x_1:\trgb{\tau_1})}\vdash
            \trg{t': \trgb{\tau}}
          }
          \\
          &
          \src{\inr{x_2}\mapsto}
          \emulate{n}{
            \trg{\Gat,(x_2:\trgb{\tau_2})}\vdash
            \trg{t'': \trgb{\tau}}
          }
        \end{aligned}\right.
        }
      \end{aligned}
    }
  \\
  \emulate{n}{
      \trg{\Gat}\vdash
      \trg{\inl{t}}
      \trgb{: \tau\uplus\tau'}
  } \isdef
    &\
      \src{
        \indn{n;\trgb{\tau\uplus\tau'}}
        \left(\inl{
          \emulate{n}{
            \trg{\Gat}\vdash
            \trg{t: \trgb{\tau}}
          }\right)
        }
      }
  \\
  \emulate{n}{
      \trg{\Gat}\vdash
      \trg{\inr{t}}
      \trgb{: \tau\uplus\tau'}
  } \isdef
    &\
      \src{
        \indn{n;\trgb{\tau\uplus\tau'}}
        \left(\inr{
          \emulate{n}{
            \trg{\Gat}\vdash
            \trg{t: \trgb{\tau'}}
          }\right)
        }
      }
  \\
  \emulate{n}{\begin{aligned}
      &
      \trg{\Gat}\vdash
      \trg{
        \begin{aligned}
          &
          \iftet{\trg{t}}{\trg{t1}
          \\
          &
          }{\trg{t2}}
        \end{aligned}
      }
      \trgb{: \tau}
    \end{aligned}\hspace*{-12pt}} 
    \isdef
    &\
      \src{
        \begin{aligned}
        &
        \iftes{
          \left(
            \caseup{n;\Boolt}
            \emulate{n}{
              \trg{\Gat}\vdash
              \trg{t: \Boolt}
            }
          \right)
        \\
        &
        }{
        \emulate{n}{
          \trg{\Gat}\vdash
          \trg{t1: \trgb{\tau}}
          }
        }{
        \emulate{n}{
          \trg{\Gat}\vdash
          \trg{t2: \trgb{\tau}}
          }
        }
        \end{aligned}
      }
  \\
  \emulate{n}{
      \trg{\Gat}\vdash
      \trg{t;t'}
      \trgb{: \tau}
    } 
    \isdef
    &\
      \src{
        \left(
        \caseup{n;\Unitt}
        \emulate{n}{
          \trg{\Gat}\vdash
          \trg{t: \Unitt}
        }
        \right);
        \emulate{n}{
          \trg{\Gat}\vdash
          \trg{t': \tau}
        }
      }
  \\
  \emulate{n}{
      \trg{\Gat}\vdash
      \trg{\fold{\matt}t}
      \trgb{: \mat}
  } \isdef
    &\
      \src{
        \indn{n;\trg{\tat\subt{\matt}{\alpt}}}
        \emulate{n}{
          \trg{\Gat}\vdash
          \trg{t: \trg{\tat\subt{\matt}{\alpt}}}
        }
      }
  \\
  \emulate{n}{
    \begin{aligned}
      &
      \trg{\Gat}\vdash
      \trg{\unfold{\matt}\ t}
      \\
      &\ \
      \trg{: \tat\subt{\matt}{\alpt}}
    \end{aligned}
  } \isdef
    &\  
      \src{
        \caseup{n;\trgb{\mat}}
        \emulate{n}{
          \trg{\Gat}\vdash
          \trg{t: \trgb{\mat}}
        }
      }
\end{aligned}
\\
\cline{1-2}
\begin{aligned}
    \emulateic{n}{\cdots}
      \isdef
      &\
      \text{ as } \emulateef{n}{\cdots}
\end{aligned}
\\
\cline{1-2}
\begin{aligned}
  \emulateef{n}{
        \AxiomC{$\oth{\Gamma}\vdash\oth{t:\othb{\tau}}$}
        \AxiomC{$\tyeqbin{\tau}{\sigma}$}
        \BinaryInfC{$\oth{\Gamma}\vdash\oth{t:\othb{\sigma}}$}
        \DisplayProof
  } 
  \isdef
    &\
        \emulateef{n}{
          \AxiomC{$\oth{\Gamma}\vdash\oth{t:\othb{\tau}}$}
          \DisplayProof
        }
    &
    \emulateef{n}{ \cdots }
      \isdef&\
      \begin{aligned}
        &
      \text{other cases}
      \\
      &
      \text{ are as above}
      \end{aligned}
\end{aligned}
\end{gather*}
\caption{\label{fig:emulat}Emulation of target terms into source ones.}
\end{figure}

\begin{figure}[!t]
  \small
  \mytoprule{
  \emulate{n}{\cdot} : 
    \
    (\vdash\ctxt:\trgb{\Gamma},\tat\to\trgb{\Gamma'},\trgb{\tau'}) \to  \ctxs
  }
  \scalebox{0.95}{%
  \begin{minipage}{\textwidth}
  \begin{gather*}
    \begin{aligned}
  \emulate{n}{\hole{\cdot}} \isdef&\ \src{\hole{\cdot}}
  \\
  \emulate{n}{
    \begin{aligned}
      &
      \vdashbl \trg{\lamt{x:\trgb\tau'}{\ctxt}} : 
      \\
      &
      \trgb{\Gamma''},\trgb{\tau''} \tobl \trgb{\Gamma},\trgb{\tau' \to \tau}
    \end{aligned}
  } \isdef
    &\
    \src{
      \begin{aligned}
        &
        \src{\indn{n;\trgb{\tau\to\tau'}}}
        \
        \\
        &\
        \src{
	        \left(\lam{\src{x:\uval{n;\trgb{\tau}}}}{
	          \emulate{n}{
	              \vdashbl \ctxt : 
	              \trgb{\Gamma''},\trgb{\tau''} \tobl \trg{\trgb{\Gamma}, x:\trgb{\tau'}},\trgb{\tau}
	          }
	        }\right)
        }
      \end{aligned}
    }
  \\
  \emulate{n}{
      \vdashbl \trg{\ctxt\ t_2} : 
      \trgb{\Gamma'},\trgb{\tau'} \tobl \trgb{\Gamma}, \trgb{\tau_2}
  } \isdef
    &\
    \src{
        \begin{aligned}
          &
          \left( 
            \caseup{n;\trgb{\tau'\to\tau}}
            \emulate{n}{
                \vdashbl \ctxt : 
                \trgb{\Gamma'},\trgb{\tau'} \tobl \trgb{\Gamma}, \trgb{\tau_1\to\tau_2}
            } 
          \right)
          ~
          \\
          &\ 
          \left(\emulate{n}{
            \trgb{\Gamma}\vdash
            \trg{t_2:\trgb{\tau_1}}
          }\right)
        \end{aligned}
      }
  \\
  \emulate{n}{
    \begin{aligned}
      &
      \vdash \trg{\projone{\ctxt}} : 
      \\
      &
      \trg{\Gamma'},\trgb{\tau'} \to \trg{\Gamma}, \trgb{\tau_1}
    \end{aligned}
  } \isdef
    &\
    \src{
        \projtwo{\left(
          \caseup{n;\trgb{\tau\times\tau'}}
          \emulate{n}{
            \vdash \ctxt : 
            \trg{\Gamma'},\trgb{\tau'} \to \trg{\Gamma}, \trgb{\tau_1\times\tau_2}
          }
        \right)}
      }
  \\
  \emulate{n}{
    \begin{aligned}
      &
      \vdash \trg{\projtwo{\ctxt}} : 
      \\
      &
      \trg{\Gamma'},\trgb{\tau'} \to \trg{\Gamma}, \trgb{\tau_2}
    \end{aligned}
  } \isdef
    &\
    \src{
        \projone{\left(
          \caseup{n;\trgb{\tau\times\tau'}}
          \emulate{n}{
            \vdash \ctxt : 
            \trg{\Gamma'},\trgb{\tau'} \to \trg{\Gamma}, \trgb{\tau_1\times\tau_2}
          }
        \right)}
      }
  \\
  \emulate{n}{
    \begin{aligned}
      &
      \vdash \trg{\pair{\ctxt,t_2}} :
      \\
      &
      \trg{\Gamma'},\trgb{\tau'}\to\trg{\Gamma},\trgb{\tau_1\times\tau_2}
    \end{aligned}
  } \isdef
    &\
    \src{
      \begin{aligned}
        &
        \indn{n;\trgb{\tau_1\times\tau_2}}
        \\
        &\
        \pair{
          \emulate{n}{
            \vdash \ctxt :
            \trg{\Gamma'},\trgb{\tau'}\to\trg{\Gamma},\trgb{\tau_1}
          }
        ,
          \emulate{n}{
            \trg{\Gamma}\vdash
            \trg{t_2:\trgb{\tau_2}}
          }
        }
      \end{aligned}
    }
  \\
  \emulate{n}{
    \begin{aligned}
      &
      \vdash \trg{\inl{\ctxt}} : 
      \\
      &
      \trg{\Gamma''},\trgb{\tau''}\to \trg{\Gamma},\trgb{\tau\uplus\tau'}
    \end{aligned}
  } \isdef
    &\
    \src{
        \indn{n;\trgb{\tau\uplus\tau'}}
        \left(\inl{
          \emulate{n}{
            \vdash \ctxt : 
            \trg{\Gamma''},\trgb{\tau''}\to \trg{\Gamma},\trgb{\tau}
          }\right)
        }
      }
  \\
  \emulate{n}{
    \begin{aligned}
      &
      \vdash \trg{\inr{\ctxt}} : 
      \\
      &
      \trg{\Gamma''},\trgb{\tau''}\to \trg{\Gamma},\trgb{\tau\uplus\tau'}
    \end{aligned}
  } \isdef
    &\
    \src{
        \indn{n;\trgb{\tau\uplus\tau'}}
        \left(\inr{
          \emulate{n}{
            \vdash \ctxt : 
            \trg{\Gamma''},\trgb{\tau''}\to \trg{\Gamma},\trgb{\tau'}
          }\right)
        }
      }
  \\
  \emulate{n}{
    \begin{aligned}
      &
      \vdash \trg{\casefoldedt{\ctxt}{\trg{t_1}}{\trg{t_2}}} : 
      \\
      &
      \trg{\Gamma'},\trgb{\tau'} \to \trg{\Gamma}, \trgb{\tau_3}
    \end{aligned}
  } \isdef
    &\
    \begin{aligned}
      &
      \src{
        \case
        \left(
          \caseup{n;\trgb{\tau_1\uplus\tau_2}}
          \emulate{n}{
            \vdash \ctxt : 
            \trg{\Gamma'},\trgb{\tau'} \to \trg{\Gamma}, \trgb{\tau_1\uplus\tau_2}
          }
        \right)
      }
      \\&\ \
        \src{\of\left|\begin{aligned}
          &
          \src{\inl{x_1}\mapsto}
          \emulate{n}{
            \trg{\Gamma,(x_1:\trgb{\tau_1})}\vdash
            \trg{t_1: \trgb{\tau_3}}
          }
          \\
          &
          \src{\inr{x_2}\mapsto}
          \emulate{n}{
            \trg{\Gamma,(x_2:\trgb{\tau_2})}\vdash
            \trg{t_2: \trgb{\tau_3}}
          }
        \end{aligned}\right.
        }
    \end{aligned}
  \\
  \emulate{n}{
    \begin{aligned}
      &
      \vdash \trg{\ctxt;t} : 
      \\
      &
      \trg{\Gamma},\trgb{\tau}\to\trg{\Gamma'},\trgb{\tau''}
    \end{aligned}
  } \isdef
    &\
    \src{
      \begin{aligned}
        &
        \left(
        \caseup{n;\Unitt}
        \emulate{n}{
          \vdash \ctxt : 
          \trg{\Gamma},\trgb{\tau}\to\trg{\Gamma'},\trg{\Unitt}
        }
        \right);
        \\
        &\
        \emulate{n}{
          \trg{\Gamma}\vdash
          \trg{t': \trgb{\tau}}
        }
      \end{aligned}
    }
  \\
  \end{aligned}
  \\
  \begin{aligned}
  \emulate{n}{
      \vdashbl \trg{\fold{\matt}\ctxt} : 
      \trgb{\Gamma'},\trgb{\tau'}\tobl\trgb{\Gamma},\trgb{\mat}
  } \isdef
    &\
    \src{
      \begin{aligned}
        &
        \indn{n;\trg{\tat\subt{\matt}{\trgb{\alpha}}}}
        \\
        &\
        \emulate{n}{
          \vdashbl \ctxt : 
          \trgb{\Gamma'},\trgb{\tau'}\tobl\trgb{\Gamma},\trg{\tat\subt{\matt}{\trgb{\alpha}}}
        }
      \end{aligned}
    }
  \\
  \emulate{n}{
      \vdashbl \trg{\unfold{\matt}\ctxt} : 
      \trgb{\Gamma'},\trgb{\tau'}\tobl\trgb{\Gamma},\trg{\tat\subt{\matt}{\trgb{\alpha}}}
  } \isdef
    &\
    \src{
        \caseup{n;\trgb{\mat}}
        \emulate{n}{
          \vdashbl \ctxt : 
          \trgb{\Gamma'},\trgb{\tau'}\tobl\trgb{\Gamma},\trgb{\mat}
        }
    }
    \end{aligned}
  \end{gather*}
  \end{minipage}
  }
  \smallskip

  \mytoprule{
  \emulateic{n}{\cdot} : 
    \
    (\vdash\ctxo:\equi{\Gamma},\equi{\tau}\to\equi{\Gamma'},\equi{\tau'}) \to  \ctxt
  }
  \begin{center}
  	\text{ Analogous to the case above since \ctxo are a subset of \ctxt }
  \end{center}

  \mytoprule{
  \emulateef{n}{\cdot} : 
    \
    (\vdash\ctxo:\equi{\Gamma},\equi{\tau}\to\equi{\Gamma'},\equi{\tau'}) \to  \ctxs
  }

  \begin{center}
  	\text{ Analogous to the case above since \ctxo are a subset of \ctxt }
  \end{center}

  \caption{\label{fig:emul-ctx}Emulation of target contexts into source ones (excerpts).}
\end{figure}
Function \emulate{n}{\cdot} (\Cref{fig:emulat,fig:emul-ctx}) is indexed by the approximation index \com{n} in order to know which \uval{}-helper functions to use.
There are few interesting bits in the emulation of terms (and of contexts).
When emulating constructors for terms of type \tat, we create a value of the corresponding backtranslation type \uval{n;\tat} and, in order to be well-typed, we \downgrade{} that value by 1.
Dually, emulating destructors for terms of type \tat\ requires upgrading the term for 1 level of precision because they are then destructed to access the underlying type.
When emulating \stlcim derivations into \stlcf, we need to consider the case when \iso{\fold{\matt}} and \iso{\unfold{\matt}} annotations are encountered.
There, we know that the backtranslation will work with terms typed at the unfolding of \matt, so we simply perform the recursive call and insert the appropriate helper function to ensure the resulting term is well-typed.
Concretely, \Cref{ex:emu} shows what the emulation of a simple term is.
\begin{example}[Emulating a term]\label{ex:emu}
  Consider the term \trg{\trge\vdash{\truet}:\Boolt}, its emulation is:
  \begin{align*}
    &\
    \src{
          \indn{n;\trgb{\Boolt}} \trues
      }
    \\
    &\text{then by unfolding the definition of \indn{\cdot}}
    \\
    =
    &\
    \src{
          (\lam{y: \Bools }{ \downgradepar{n;\trg{\Boolt}}{\inl{y}} })
          ~\trues
      }
    \\
    &\text{then by unfolding the definition of \downgradepar{\cdot}{} }
    \\
    =
    &\
    \src{
          (\lam{y: \Bools }{ (\lam{z:\Bools\uplus\Units}{z}){\inl{y}} })
          ~\trues
      }
  \end{align*}
  Which eventually reduces to value \src{\inl{\trues}}, as expected.
\end{example}

When emulating \stlcem derivations (in the other two emulates in \Cref{fig:emulat}), we need to consider the case when term \equi{t} is given type \equi{\tau} knowing it had type \equi{\sigma} and that \tyeqbin{\sigma}{\tau} (\Cref{tr:t-eq}).
Here we rely on a crucial observation: given two equivalent types, their backtranslation types are \emph{the same} (\Cref{thm:tyeq-same-uval}).
To understand why this is the case, consider how the definition of $\uval{n;\othb{\tau}}$ simply unfolds recursive types without losing precision, i.e.\ it essentially only looks at the depth-$n$ unfolding of type $\othb{\tau}$ and these unfoldings are equal for equal types $\tyeqbin{\tau}{\sigma}$.
With this fact, we can get away with just performing the recursive call on the sub-derivation for \equi{t} at type \equi{\sigma}.
\begin{theorem}[Equivalent types are backtranslated to the same type]\label{thm:tyeq-same-uval}
  \[
  \text{If } \tyeqbin{\tau}{\sigma} \text{ then } \uvalfe{n;\equi{\tau}}=\uvalfe{n;\equi{\sigma}}
  \]
\end{theorem}

Finally, consider \emulateic{\cdot}{\cdot}, i.e., the emulation of \stlcem terms into \stlcim: there is no construct that adds \iso{\fold{}}/\iso{\unfold{}} annotations.
This is due to the same intuition presented before regarding the unfolding of the backtranslation type \uvalic{n;\equi{\mat}}, which is \uvalic{n;\equi{\tau\subo{\mat}{\alpha}}} i.e, the indexing type is unfolded but the step is not decreased.
Intuitively, the backtranslation performs an \com{n}-level deep unfolding of the recursive types and operates on those.
Thus, backtranslated contexts do not use recursive types but just their \com{n}-level deep unfolding, so their annotations are not needed.

In order to state that \emulate{}{\cdot} is correct, we rely on compatibility lemmas akin to those used for compiler correctness (recall \Cref{thm:compat-lem-lam}).
First, note that all our logical relations relate a source and target term at a source pseudo-type.
We have extended the logical relation to express the relation between a source and target term at pseudotype \emuldv{}, so we should use this to relate a target term and its backtranslation.
Second, all logical relations require a source environment to relate terms, and in this case we are given a target environment (the one for the typing of the backtranslated term).
To create a source environment starting from this target environment, we take each bound variable and give it backtranslation type using function \toemul{\cdot}{}.
Finally, in these lemmas we need to account for the different directions of the approximation we have.
Thus, these compatibility lemmas require that either $n < m$ (so that the results only hold in worlds $\W$ with $\stepsfun{\W} \leq n < m$) and $p = \precise$ or $\genlogrel = \lesssim$ and $p = \imprecise$, for \com{m} being the approximation level of interest. 

The intuition behind these constraints is that when $p = \imprecise$, there is no lower bound on the emulation depth $m$.
However, in that case, we only get a left-to-right approximation $\underlogrel$, since the emulated term may have insufficient precision to emulate the original term accurately (as in \Cref{ex:imprecise}) and may diverge in cases where the emulation precision runs out.
In the case where $p = \precise$, the lemma requires that the emulation depth $m$ is sufficiently large.
Specifically, $m$ is required to be at least as large as the step bound $n$ up to which the approximation in both directions $\anylogrel{n}$ is guaranteed.
Intuitively, this covers the scenario where the depth of the approximation is larger than the amount of steps taken by a back-translated program.
In such a scenario, the back-translation is guaranteed to accurately emulate the behaviour of the target term and we get approximations in both directions, but only up to the amount of steps $n$.

Thus, a typical compatibility lemma for emulate looks like \Cref{thm:compat-lem-backtr-lam}.

\begin{lemma}[Compatibility for $\lambda$ Emulation]\label{thm:compat-lem-backtr-lam}
  \begin{align*}
    \text{if }
    & 
      (m > n \text{ and } p = \precise) \text{ or } (\anylogrel = \underlogrel \text{ and } p = \imprecise)
    \\
    \text{then }
    &
    \begin{aligned}[t]
      \text{if }
      &
        \toemul{\trg{\Gat,x:\trgb{\tau}}}{m;p}\vdash \src{t} \anylogreln{n} \trg{t}: \emuldv{m;p;\trgb{\tau'}} 
      \\
      \text{then }
      &
        \toemul{\trg{\Gat}}{m;p}\vdash 
          \src{ \indn{m;\trgb{\tau\to\tau'}} \left(\lam{\src{x:\uval{m;\trgb{\tau}}}}{ t } \right)}
        \anylogreln{n} 
          \trg{\lamt{x:\trgb{\tau}}{t}}
        : \emuldv{m;p;\trgb{\tau\to\tau'}} 
    \end{aligned}
  \end{align*}
\end{lemma}
The compatibility lemma for terms typed using type equality (\Cref{thm:fe-compat-lem-backtr-eq}) is the most interesting of these.
The proof of this lemma is surprisingly simple because most of the heavy lifting is done by a corollary of \Cref{thm:tyeq-same-uval}, which proves that equivalent types have not only the same backtranslation type but also the same term relation.
\begin{lemma}[Compatibility lemma for emulation of type equality]\label{thm:fe-compat-lem-backtr-eq}
  \begin{align*}
    \text{if }
    &
      (m > n \text{ and } p = \precise) \text{ or } (\anylogrel = \underlogrel \text{ and } p = \imprecise)
    \\
    \text{then }
    &
    \begin{aligned}[t]
      \text{if }
      &
        \toemulfe{\oth{\Gamma}}{m;p}\vdash \src{t} \anylogreln{n} \oth{t}: \emuldvfe{m;p;\othb{\tau}}
            \text{ and }
                \tyeqbin{\tau}{\sigma}
      \\
      \text{then }
      &
        \toemulfe{\oth{\Gamma}}{m;p}\vdash \src{t} \anylogreln{n} \oth{t}: \emuldvfe{m;p;\othb{\sigma}}
    \end{aligned}
  \end{align*}
\end{lemma}
\begin{corollary}[Equivalent types have the same term relation]\label{thm:fe-eq-type-eq-rels}
  \begin{align*}
    \text{if }
      &
      \
      \tyeqbin{\tau}{\sigma} 
        \text{ then } \forall n\ldotp
          \termrel{\emuldvfe{n;p;\tau}} = \termrel{\emuldvfe{n;p;\sigma}}
  \end{align*}
\end{corollary}

Given a series of these kinds of compatibility lemmas, we can state that emulate is correct.
\begin{lemma}[Emulate is semantics-preserving]\label{thm:emul-sem-pres}
\begin{align*}
  \text{if }  
    &\
    (m > n  \text{ and } p = \precise) \text{ or } (\anylogrel = \underlogrel \text{ and } p = \imprecise)
  \text{ and }
    \trg{\Gat} \vdash \trg{t} : \trgb{\tau}
  \\
  \text{then }
    &\
    \toemul{\trg{\Gat}}{m;p}\vdash \emulate{m}{\trg{\Gat}\vdash\trg{t}:\trgb{\tau}} \anylogreln{n} \trg{t}: \emuldv{m;p;\trgb{\tau}}
\end{align*}
\end{lemma}

The key property we rely on for fully-abstract compilation though, is that emulation of contexts is correct (this relies on correctness of emulation for terms though).
\begin{lemma}[Emulate is semantics preserving for contexts]\label{lem:emulate-works-ctx}\label{thm:emul-ctx-sem-pres}
\footnotesize
\begin{align*}
  \text{if }
    &\
    (m > n  \text{ and } p = \precise) \text{ or } (\anylogrel = \underlogrel \text{ and } p = \imprecise)
    \text{ and }
    \vdash \ctxt : \trg{\Gat'},\trgb{\tau'} \to \trg{\Gat},\trgb{\tau}
  \\
  \text{then }
    &\
    \vdash \emulate{m}{\vdash \ctxt : \trg{\Gat'},\trgb{\tau'} \to \trg{\Gat},\trgb{\tau}} \anylogreln{n} \ctxt :
      \toemul{\trg{\Gat'}}{m;p},\emuldv{m;p;\trgb{\tau'}} \to
      \toemul{\trg{\Gat}}{m;p}, \emuldv{m;p;\trgb{\tau}}
\end{align*}
\end{lemma}

\subsubsection{Inject and Extract}\label{sec:injext}
As mentioned, the backtranslated target context must be a valid source context in order to be linked with a source term.
Specifically, it must have a hole whose type is the compilation of some source type $\src{\tau}$. Backtranslated terms, however, have backtranslation type $\uvalfe{n;\tau}$, so we need to convert values of source type into values of backtranslation type (and back).
To do this conversion we rely on functions \inject{} and \extract{} whose types and definitions are in \Cref{fig:injext}.
Function \inject{} takes a source value of type \src{\tau} and converts it into ``the same'' value at the backtranslation type so that backtranslated terms can use that value.
Since the backtranslation type is indexed by target types, we use function \srctotrgty{\cdot} to generate the target type related to \src{\tau}.
Function \extract{} does the dual and takes a value of backtranslation type and converts it into a type of some source type.
These functions are defined mutually inductively in order to contravariantly convert function arguments to the appropriate type.

\begin{figure}[!ht]
\small
\mytoprule{
  \inject{n;\tau} : 
    \src{\src{\tau}\to\uval{n;\srctotrgty{\src{\tau}}}}
  \quad\text{ and }\quad
  \extract{n;\tau} : 
    \src{\uval{n;\srctotrgty{\src{\tau}}}\to\src{\tau}}
}

  \scalebox{0.95}{%
  \begin{minipage}{\textwidth}
  \begin{gather*}
  \begin{aligned}
    \inject{0;\tau} =
      &\
      \src{\lam{x:\tau}{\units}}
    &&&
    \inject{n+1;\Units} =
      &\
      \src{
        \lam{x:\Units}{\inl{x}}
      }
    &&&
    \inject{n+1;\Bools} =
      &\
      \src{
        \lam{x:\Bools}{\inl{x}}
      }
  \end{aligned}
  \\
  \begin{aligned}
    \inject{n+1;\tau\to\tau'} =
      &\
      \src{
        \begin{aligned}[t]
          &
          \src{\lam{x:\tau\to\tau'}{}}
            \src{\inl{}}
              \src{\lam{y:\uval{n;\srctotrgty{\src{\tau}}}}{}}
                \injectpar{n;\tau'}{
                  x~
                  (
                    \extract{n;\tau}{
                      y
                    }
                  )
                }
        \end{aligned}
      }
    \\
    \inject{n+1;\tau\times\tau'} =
      &\
      \src{ 
        \begin{aligned}[t]
          &
          \src{\lam{x:\tau\times\tau'}{}}
            \src{\inl{}}
              \pair{
                \injectpar{n;\tau}{\projone{x}}
                ,
                \injectpar{n;\tau'}{\projtwo{x}}
              }
        \end{aligned}
      }
    \\
    \inject{n+1;\tau\uplus\tau'} =
      &\
      \src{
        \begin{aligned}[t]
          &
          \src{\lam{\src{x:\tau\uplus\tau'}}{}}
            \src{\inl{}}
          \src{
              \caseof{
                x
              }{
                \src{\inl{(\inject{n;\tau} x_1)}}
              }{
                \src{\inr{(\inject{n;\tau'} x_2)}}
              }
          }
        \end{aligned}
      }
  \end{aligned}
  \\
  \begin{aligned}
    \extract{0;\tau} =
    &\
    \src{
      \lam{x:\uval{n;\srctotrgty{\src{\tau}}}}{\myomega_\tau}
    }
  \\
  \extract{n+1;\Units} =
    &\
    \src{
      \lam{x:\uval{n+1;\Unitt}}{ \casetag{n+1;\Unitt} x }
    }
  &&&
  \extract{n+1;\Bools} =
    &\
    \src{
      \lam{x:\uval{n+1;\Boolt}}{ \casetag{n+1;\Boolt} x }
    }
  \end{aligned}
  \\
  \begin{aligned}
  \extract{n+1;\tau\to\tau'} =
    &\
    \src{
      \begin{aligned}[t]
        &
        \lam{\src{x:\uval{n+1;\srctotrgty{\src{\tau\to\tau'}}}}}{
          \src{
            \lam{y:\tau}{
              \extract{n;\tau'} \left(
                \casetag{n;\srctotrgty{\src{\tau\to\tau'}}} x \left(\inject{n;\tau}~y \right) 
              \right)
            }
          }
        }
      \end{aligned}
    }
  \\
  \extract{n+1;\tau\times\tau'} =
    &\
    \src{
      \lam{x:\uval{n+1;\srctotrgty{\src{\tau\times\tau'}}}}{ 
        \pair{
          \begin{aligned}
            &
            \extract{n;\tau} \left(\projone{\casetag{n;\srctotrgty{\src{\tau}}} \src{x}}\right) , 
            \\
            &
            \extract{n;\tau'} \left(\projtwo{\casetag{n;\srctotrgty{\src{\tau'}}} \src{x}}\right) 
          \end{aligned}
        }}
    }
  \\
  \extract{n+1;\tau\uplus\tau'} =
    &\
    \src{
        \lam{\src{x:\uval{n+1;\srctotrgty{\src{\tau\uplus\tau'}}}}}{
          \src{
            \casefoldeds{
              \left(\casetag{n;\srctotrgty{\src{\tau\uplus\tau'}}} \src{x}\right)
            }{ 
              \src{\inl{\extract{n;\srctotrgty{\src{\tau}}} \src{x_1} }}
            }{ 
              \src{\inr{\extract{n;\srctotrgty{\src{\tau'}}} \src{x_2} }}
            }
          }
        }
    }
  \end{aligned}
  \end{gather*}
  \end{minipage}
  }
  \smallskip
\hrule
  \begin{align*}
    \injectic{n+1;\matt} =
      &\
      \trg{
        \begin{aligned}[t]
          &
          \lamt{
            \trg{x:\matt}
          }{
            \trg{\injectic{n+1;\tat\subt{\matt}{\alpt}}\ (\unfold{\matt}\ x)}
          }
        \end{aligned}
      }
    \\
    \extractic{n+1;\matt} =
      &\
      \trg{
        \begin{aligned}[t]
          &
          \lamt{
            \trg{x:\uvalic{n+1;\srctotrgtyic{\matt}}}
          }{
            \extractic{n+1; \matt } \trg{\fold{\matt}}\ \trg{(\casetagic{n+1;\srctotrgtyic{\matt}}\ x)}
          }
        \end{aligned}
      }
    \\
      &
      \text{omitted cases are as above}
  \end{align*}  
\hrule
  \begin{align*}
    \injectfe{n;\tau} \isdef&\ \text{ as above}
      &
      \extractfe{n;\tau} \isdef&\ \text{ as above}
  \end{align*}
  \caption{\label{fig:injext}Definition of the \src{inject} and \src{extract} functions.}
\end{figure}
For values of the base type, these functions use the already introduced constructors and destructors for backtranslation type to perform their conversion.
For pair and sum types, these functions operate recursively on the structure of the values they take in input.
For arrow type, these functions convert the argument contravariantly before converting the result after the application of the function.
When the size of the type is insufficient for these functions to behave as expected (i.e., when $n$ is $0$) it is sufficient for \inject{} to return \units and for \extract{} to just diverge.
\begin{example}[The need for \extract{}]\label{ex:inj}
  Consider the emulated term from \Cref{ex:emu}: \src{\inl{\trues}}, which is the result of emulating \trg{\trge\vdash\truet:\Boolt}.
  Ideally, we want to extract that term into type \src{\Bools} at index \src{1}, in order to strip the underlying \trues of the outer \src{\inl{\cdot}}.
  That is precisely what \extract{1;\Bools} does:
  \begin{align*}
    &\
    \src{(\extract{1;\Bools})}~ \src{\inl{\trues}}
    =
    \\
    &\
    \src{
      \lam{x:\uval{1;\Boolt}}{ \casetag{2;\Boolt} x }
    }~ \src{\inl{\trues}}
    \\
    &
    \text{ which by definition of \casetag{\cdot} becomes }
    \\
    &\
    \src{(
      \lam{y:\uval{1;\Boolt}}{ 
          (\lam{x:\uval{2;\Boolt}}{
        \caseof{x}{x_1}{\myomega_{\uval{2;\Boolt}}})
      }
       y }
    )}~ \src{\inl{\trues}}
  \end{align*}
  After two reduction steps, this term becomes the expected \trues, which can be used at the expected \Bools type.
\end{example}

Note that these functions are indexed by \emph{source} types since they convert between them and the backtranslation type.
Thus, while two of our compilers have the same source language (and therefore the same \src{inject}/\src{extract}), the third compiler has a different source language, with more types: \iso{\matt}.
Thus, for the third backtranslation, we have a different, extended version of \injectic{}/\extractic{} that converts values of recursive types into values of backtranslation type and back.
Additionally, the hole of the first two backtranslations cannot have a recursive type, since the source type for those backtranslations is \stlcf.

As for the emulation of terms, we prove that these functions are correct according to the logical relations.
Terms that are related at a source type are related at backtranslation type after an \inject{} while terms that are related at backtranslation type are related at source type after an \extract{}.
\begin{lemma}[Inject and extract are semantics preserving]\label{thm:inj-ext-sem-pres}
\begin{align*}
  \text{ If }
    &\
    (m \geq n \text{ and } p = \precise) \text{ or } (\anylogrel = \underlogrel \text{ and } p = \imprecise)
  \\
  \text{then }
    &\
    \begin{aligned}[t]
      \text{if }
        &\
        \src{\Gamma} \vdash \src{t} \anylogreln{n} \trg{t} : \src{\tau}
      \text{ then }
        \src{\Gamma} \vdash \inject{m;\tau}\ \src{t} \anylogreln{n} \trg{t} : \emuldv{m;p;\srctotrgty{\src{\tau}}}
    \end{aligned}
    \\
    &\
    \begin{aligned}[t]
      \text{if }
        &\
        \src{\Gamma} \vdash \src{t} \anylogreln{n} \trg{t} : \emuldv{m;p;\srctotrgty{\src{\tau}}}
      \text{ then }
        \src{\Gamma} \vdash \extract{m;\tau}\ \src{t} \anylogreln{n} \trg{t} : \src{\tau}
    \end{aligned}
\end{align*}
\end{lemma}

As mentioned in \Cref{sec:intro}, \Cref{thm:inj-ext-sem-pres} broke with the logical relation that does not define the observation relation $\obsfun{\W}{\bothlogrel}$ in terms of \newterm.
\Cref{ex:need-newterm} below argues why this technical change is needed and what the differences are in the technical development as opposed to the old one of \citet{popl-journal}.

\begin{example}[The Need for \NewTerm]\label{ex:need-newterm}

In this example, assume the logical relation does not rely on $\obsfun{\cdot}{}$, but on the equi-termination observation relation defined below ($\wobsfun{\cdot}{}$ for \mc{W}rong).
\begin{align*}
	\wobsfun{\W}{\underlogrel}\isdef
	&\
	\myset{(\src{t},\trg{t})}{ \text{if } \stepsfun{\W}>n \text{ and } \src{t\bterms{n}v} \text{ then } \exists \trg{k}.~ \trg{t\btermt{k}v} }
	\\
	\wobsfun{\W}{\overlogrel}\isdef
	&\
	\myset{(\src{t},\trg{t})}{ \text{if } \stepsfun{\W}>n \text{ and } \trg{t\btermt{n}v} \text{ then } \exists \src{k}.~ \src{t\bterms{k}v} }
	\\
	\wobsfun{\W}{\bothlogrel}\isdef
	&\
	\wobsfun{\W}{\underlogrel}\cap\wobsfun{\W}{\overlogrel}
\end{align*}
Now take the following two terms (for $m >= 1$):
\begin{align*}
  \tprobs :
  	&\ 
  		\uval{m; \iso{(\Boolt \uplust \Unitt)\uplust\Unitt}}
  	&
  		\tprobt :
  			&\  
  				\iso{(\Boolt \uplust \Unitt)\uplust\Unitt}
  	\\
	\tprobs =
		&\
		\src{\inl{(\inl{(\inl{(\inl{(\inr{\units})})})})}}
	&
	\tprobt =
		&\
		\iso{\inl{(\inl{\truet})}}
\end{align*} 
Intuitively, $\tprobs$ correctly emulates $\tprobt$ but only one level deep: it correctly emulates the outer two $\iso{\inl{}}$ constructors as two $\src{\inl{\inl{ \cdot }}}$ but then bails out by using $\src{\inr{\units}}$, i.e., the right branch of the $\src{\cdots \uplus \Unit}$ in the definition of $\uval{n;\tat}$, which models approximation failure.
For these two terms, we can fulfil the premise of \Cref{thm:inj-ext-sem-pres} for specific $n$ and $m$ and prove that \tprobs and \tprobt are related, but unfortunately we cannot prove the conclusion of the lemma, which amounts to proving that if \tprobt terminates, then \extract{\cdot}\tprobs terminates as well.

Let us first show that the premise of the lemma is satisfied for $\anylogreln{} = \gtrsim{}$ and $n=1$.
This amounts to proving that \tprobs and \tprobt are in the term relation for $\emuldv{m;p;\tat}$, where $\tat = \trg{(\Boolt\uplust\Unitt)\uplust\Unitt}$, $\src{m}=3$ and $\src{p}=\precise$.
For this, we have to prove that they are in the term relation for any world $\W$ whose level is at most $n$, i.e., $1$.
In the case where the level is \com{0} the relation is trivial, since any term is related in a world with no steps.
Since the term relation includes the value relation, it suffices to show that: $(\W, \tprobs, \tprobt)\in\valrel{\emuldv{3;p;\iso{(\Boolt\uplust\Unitt)\uplust\Unitt}}}$.
From the definition of that value relation (\com{n+1} case) it suffices to strip \tprobs of one \src{\inl{\cdot}} and show that the terms are in $\valrel{ \emuldv{2;p;\trg{\Boolt\uplust\Unitt}} \times \emuldv{2;p;\trg{\Unitt}}}$.
From the definition of the value relation for \src{\uplus} it suffices to strip each term of an \com{\inl{\cdot}}, decrease the level of \W by \com{1} (which becomes \com{0}) and show that the resulting terms (\src{\inl{\inl{\inr{\units}}}} and \iso{\inl{\truet}}) are in $\valrel{ \emuldv{2;p;\trg{\Boolt\uplust\Unitt}}}$.
Again from the definition of the value relation for \emuldv{\cdot} (\com{n+1} case) it suffices to strip \tprobs of one \src{\inl{\cdot}} and show that the terms are in $\valrel{ \emuldv{1;p;\trg{\Boolt}} \times \emuldv{1;p;\trg{\Unitt}}}$.
From the definition of the value relation for \src{\uplus} it suffices to strip each term of an \com{\inl{\cdot}} and prove that (\src{\inr{\units}} and \iso{\truet}) are in $\later\valrel{ \emuldv{1;p;\trg{\Boolt}}}$.
This is vacuously true from the definition of $\later\valrel{}$ since the world has \com{0} steps. 
It is worth noting that if we had taken $n > 1$, we would not be able to prove that $\tprobs$ and $\tprobt$ are related, since the premise of $\later\valrel{\cdot}$ would be true, but the conclusion would not be (i.e., \src{\inr{\units}} and \iso{\truet} are not in \valrel{ \emuldv{1;p;\trg{\Boolt}}} for any world).

We now focus on the reductions for the problematic case of extract, for which the conclusion of \Cref{thm:inj-ext-sem-pres} does not hold (note that $\src{\tau} = \src{(\Bools\uplus\Units)\uplus\Units}$).
{\small
\begin{align*}
	&
	\src{\extract{2,\tau}~\tprobs}
	\\
	\src{\redgen{3}\ }
	&
   \src{\inl{ (\extract{2,\Bools\uplus\Units}~(\inl{(\inl{(\inr{\units})})}) )}}
  	\\
  	{\isdef\ }
    &
    \src{\inl{\left(
      \left(\lam{\src{x:\uval{2;\tat}}}{
        \src{
            \casefoldeds{
              \left(\casetag{1;\tau} \src{x}\right)
            }{ 
              \src{\inl{(\extract{1;\Bools} \src{x_1} )}}
            }{ 
              \src{\inr{(\extract{1;\Units} \src{x_2} )}}
            }
          }
        }
      \right)
      ~
      \src{(\inl{(\inl{(\inr{\units})})})}
      \right)}}
  \\
  \src{\redgen{3}\ }
    &
      \src{\inl{\left(\inl{(\extract{1;\Bools} \src{\left(\inr{\units}  \right)} )}  \right)}}
  \\
  \isdef\ 
    &
    \src{\inl{\left(
      \inl{\left(
        \left(\lam{\src{x:\uval{1;\Bools}}}{
            \casetag{0;\Bools} \src{x}
          }
        \right)
        ~
      \src{\left(\inr{\units}  \right)}
      \right)}
      \right)}}
    \\
    \src{\redgen{3}\ }
      &
        \src{\inl{\left(\inl{\myomega_{\uval{0;\Boolt}}}  \right)}}  \text{ which diverges}
\end{align*}}%
This breaks \Cref{thm:inj-ext-sem-pres}, since our goal was to prove that \src{\extract{2,\tau}~\tprobs} terminates.

Intuitively, the problem here is that applying extract to a value like \tprobs will diverge whenever there is an approximation failure in the value, no matter how deep in the value.
This approximation failure is ruled out by the value relation, but only for worlds with a sufficiently large step index.
For smaller worlds, whose step index is not large enough to look at the full depth of the term, the lemma simply does not hold as demonstrated by our example.

Fortunately, the observation relation $\obsfun{\cdot}{\cdot}$ from \Cref{fig:logrels-worlds} resolves this issue, so that we can prove the conclusion of \Cref{thm:inj-ext-sem-pres}.
Specifically, given that \W has level $1$, by the definition of $\obsfun{\W}{\overlogrel}$, we need to show that if $\trg{\tprobt\shrinkt{0}v}$ then \tprobs terminates.
This holds vacuously since the premise of the implication is false: it is not true that $\trg{\tprobt\shrinkt{0}v}$ since $\size{\tprobt} = 2$ and $2\not\leq 0$.
In other words, the new observation relation simply rules out worlds whose step index is not large enough to look at the full depth of the term, leaving us with only larger step indices where the problem does not exist.
\end{example}

The \newterm hypothesis of $\obsfun{\cdot}{\cdot}$ shows up in the technical development in only a few places.
For the interested reader, we now give a brief, very technical and succinct overview of where the change impacts the technical development.
Readers who are not experts or not interested are encouraged to skip ahead to \Cref{sec:backtranslations}.

Concretely, \Cref{thm:inj-ext-sem-pres} relies on two auxiliary lemmas, one for \inject{\cdot} and one for \extract{\cdot}.
The latter is extended with an additional hypothesis that if $\anylogreln{} = \gtrsim{}$, then $\size{\iso{t}}\leq\stepsfun{\W}$, which comes in handy in all the cases for constructors.
For example, when proving relatedness of two terms, knowing $\size{\iso{\inl{t}}} \leq \stepsfun{\W}$ lets us rule out the case when $\stepsfun{\W}=0$.

Dually, in the case for \inject{\cdot} for function types, \extract{\cdot} is called on the argument of the function.
In that case we need to prove that the world under consideration has enough steps to ensure \newterm of the argument of the function.
This fact follows from the additional premise in the value relation for function types.

Finally, in the compatibility lemma for application, we have to fulfil this additional premise for function types and show that the \stlcim function argument \newterms.
We get this fact by unfolding a few definitions: from the definition of logical relation and term relation, in the lemma we have to prove that for any related context, the functions applied to the values are in the observation relation.
From the observation relation for $\gtrsim$ we obtain the assumption that the \stlcim function applied to its value \newterms.
From this fact we obtain that just the value \newterms.

\subsubsection{The Backtranslations}\label{sec:backtranslations}
The backtranslation of a target context based on its type derivation is defined as follows by relying on both \emulate{}{\cdot} and \inject{}.
All three backtranslations follow exactly the same pattern and enjoy the same properties.
As already shown, the only interesting changes are in the sub-parts of the backtranslation (e.g., in the different definitions of inject/extract).
Thus, we only show the backtranslation from \stlcim to \stlcf and we state properties only for this one.
\begin{definition}[Approximate backtranslation for \stlcim contexts into \stlcf]\label{def:ctx-approx-backtr}
  {\small
  \begin{align*}
      &
      \backtrstlc{\ctxt, n} \isdef \src{\emulate{n}{\vdash\ctxt:\trg{\Gat,\compstlc{\src{\tau}}}\to\trg{\Gat',\tat'}}\hole{\inject{n;\tau}\cdot}}
      ~~(\text{provided } \vdash\ctxt:\trg{\Gat,\compstlc{\src{\tau}}}\to\trg{\Gat',\tat'})
  \end{align*}}%
\end{definition}

\begin{wrapfigure}{R}{.55\textwidth}
  \centering
  \begin{tikzpicture}[scale=0.8,every node/.style={scale=.9}]

    \node at (0,2)[anchor = east] (a){ $\src{\emulate{n}{\ctxt}[ }$};
    \node[right =of a.east, xshift = 1.5em, anchor = east](b){ {\inject{n;\tau}}};
    \node[right =of b.east, xshift = -1.3em, anchor = east](c){\src{~(t)}};
    \node[right =of c.east, xshift = -1.4em, anchor = east](c1){\src{]}};

    \draw [decorate,decoration={brace,amplitude=10pt}]  ([yshift=.8em]a.west) -- ([xshift=-.1cm,yshift=.8em]b.east)   node [black,midway,yshift=2em] (em){$\src{\backtrstlcif{\ctxt,n}}$};

    \node [anchor = east, below =1.2 of a.south, xshift=2em] (d) { \ctxt\trg{\big[} };
    \node[right =of d.east, xshift = -2em, anchor = east](e){ };
    \node[right =of e.east, xshift = .2em, anchor = east](f){\trg{\compstlc{t}~}};
    \node[right =of f.east, xshift = -1.3em, anchor = east](f1){\trg{\big]}};

    \node[left =of a.west, anchor = east](t){$\src{\backtrstlcif{\ctxt,n}\hole{t}}$};
    \node[anchor = east, below = of t.north,yshift=-.5em] (gl) { $\anylogreln{n}$ }; %
    \node[anchor = east, below = of t.south] (tt) { \trg{\ctxt\hole{\compstlc{t}}} };

    \node[rounded corners,rounded corners, fill=blue!20,below = of tt.south,xshift = -2.5em,yshift=1.5em] (th1){ \phantom{a}};
    \node[,right = of th1.west,xshift = -1.5em] (th1t){ \cref{thm:compstlc-sem-pres}};
    \node[rounded corners,rounded corners, fill=red!20,right = of th1.east,xshift=3em] (th2){ \phantom{a}};
    \node[,right = of th2.west,xshift = -1.5em] (th2t){ \cref{thm:inj-ext-sem-pres}};
    \node[rounded corners,rounded corners, fill=green!20,right = of th2.east,xshift = 3em] (th3){ \phantom{a}};
    \node[,right = of th3.west,xshift = -1.5em] (th3t){ \cref{lem:emulate-works-ctx}};

    \draw[dashed] (-3.3,-0.5) to node[midway,sloped,above,font=\footnotesize](txt){\Cref{thm:backtr-corr}} (-3.3,3);
    \draw[dashed] (-3.3,-0.5) to node[midway,sloped,below,font=\footnotesize](txt){is expanded to this} (-3.3,3);
  \begin{pgfonlayer}{background}
    \node[rounded corners,fit=(c),rounded corners, fill=blue!20] (th1c){};
    \node[rounded corners,fit=(f),rounded corners, fill=blue!20] (th1f){};
    \draw[rounded corners=2em,line width=1em,blue!20,cap=round] (c.south) -- (f.north) node [black,midway] (r1){$\anylogreln{n}$};
  \end{pgfonlayer}
  \begin{pgfonlayer}{veryback}
    \node[rounded corners,fit=(b)(th1c),rounded corners, fill=red!20] (th1b){};
    \node[rounded corners,fit=(e)(th1f),rounded corners, fill=red!20] (th1e){};
    \draw[rounded corners,line width=1em,red!20,cap=round] (b.south) -- (e.north)  node [black,midway,yshift=.3em] (r2){$\anylogreln{n}$};
  \end{pgfonlayer}
  \begin{pgfonlayer}{veryback2}
    \node[rounded corners,fit=(a)(c1)(th1b),rounded corners, fill=green!20] (th1a){};
    \node[rounded corners,fit=(d)(f1)(th1e),rounded corners, fill=green!20] (th1d){};
    \draw[rounded corners=2em,line width=1em,green!20,cap=round] (a.south) -- (d.east) node [black,midway,yshift=.5em] (r3){$\anylogreln{n}$};
  \end{pgfonlayer}
  \end{tikzpicture}
  \caption{
  \label{fig:relat}Diagram representing the relatedness between different bits of the backtranslation and of the compiler.
  }
  \vspace{-2em}
\end{wrapfigure}
As for the compiler from \stlcf to \stlcem, we can derive the backtranslation from \stlcem to \stlcf by composing the backtranslations through \stlcim.
Thus, $\backtrstlcef{t} = \backtrstlcif{\backtrstlcei{t}}$.
Interestingly, this means that the type of \stlcem terms backtranslated into \stlcf is the same as the one for \stlcem terms backtranslated into \stlcim, i.e., the case for \uvalfe{} for \equi{\mat} should not lose precision (as shown in \Cref{fig:uval}).
Notice that the first backtranslation (\backtrstlcei{\cdot}) directs this, since \uvalic{} is simply a collection of \trg{\psdic{\tat}\uplust\psdic{\tat'}} pseudotypes, the second backtranslation (\backtrstlcif{\cdot}) simply relies on the case for \uval{n;\trgb{\tat\uplus\tat'}}.

Using the same approach for the correctness of emulate, we can state that the backtranslations are correct.
For simplicity, we provide a visual representation of this proof in \Cref{fig:relat} (adapted from the work of \citet{Devriese:2016:FCA:2837614.2837618} to our setting).
All of the infrastructure used by the backtranslation (i.e., \inject{}/\extract{} and the \uval{} helpers) have correctness lemmas that follow the same structure of the one for \emulate{}{\cdot}.
Specifically, they relate terms at \emuldv{}, they transform target environments into source ones via function \toemul{\cdot}{} and they have a condition on the different directions of the approximation (the first line in \Cref{thm:compat-lem-backtr-lam,thm:emul-ctx-sem-pres,thm:emul-sem-pres,thm:fe-compat-lem-backtr-eq}).
\begin{lemma}[Correctness of \backtrstlc{\cdot} ]\label{lem:correctness-back-translation}\label{thm:backtr-corr}
\begin{align*}
  \text{ If }
    &\
    (m \geq n \text{ and } p = \precise) \text{ or } (\anylogrel = \underlogrel \text{ and } p = \imprecise)
  \\
  \text{then }
    &\
    \begin{aligned}[t]
      \text{if }  
        \vdash \ctxt : \trge,\compstlc{\tau} \to \trge,\trgb{\tau}
      \text{ and }
        \srce \vdash\src{t} \anylogreln{n} \trg{t} :\src{\tau}
      \text{ then }
        \srce\vdash \src{\backtrstlc{\ctxt, m}\hole{t}} \anylogreln{n} \ctxht{t} : \emuldv{m;p;\trgb{\tau}}
    \end{aligned} 
\end{align*}
\end{lemma}

With correctness of the backtranslation we can prove the preservation direction of fully-abstract compilation for all compilers, following the proof structure of \Cref{fig:fac-dia}.
\begin{theorem}[\compstlc{\cdot} preserves equivalence]\label{thm:contextual-equivalence-preservation}\label{thm:stlc-pres}
  \[
    \text{If }
        \srce \vdash \src{t_1} \ceqs \src{t_2} : \src{\tau}
      \text{ then }
        \trge \vdash \compstlc{t_1} \ceqt \compstlc{t_2} : \compstlc{\tau}
  \]
\end{theorem}
\begin{proof}
  Take \ctxt such that $\vdash \ctxt : \trge,\compstlc{\tau} \to \trge,\trgb{\tau}$.
  We need to prove that $\ctxht{\compstlc{t_1}}\termt \iff \ctxht{\compstlc{t_2}}\termt$.
  By symmetry, we prove only that if $\ctxht{\compstlc{t_1}}\termt$ then $\ctxht{\compstlc{t_2}}\termt$ (HPTT).
  Take $n$ strictly larger than the steps needed for $\ctxht{\compstlc{t_1}}\termt$.
  By \Thmref{thm:compstlc-sem-pres} we have $\srce\vdash\src{t_1}\anylogreln{n}\compstlc{\src{t_1}}:\src{\tau}$.
  Take $m=n$, so we have $(m \geq n \text{ and } p = \precise)$ and therefore $(\anylogrel = \overlogrel)$.
  By \Thmref{thm:backtr-corr} we have $\srce\vdash \src{\backtrstlc{\ctxt, m}\hole{t_1}} \overlogreln{n} \ctxht{\compstlc{\src{t_1}}} : \emuldv{m;p;\trgb{\tau}}$.
  By \Thmref{thm:term-def-rel} with HPTT we have: $\ctxht{\compstlc{t_2}}\shrinkt{\_}$ (HPTS).
  By \Thmref{thm:log-rel-adeq-both} for $\overlogrel$ and HPTS we have: \src{\backtrstlc{\ctxt, m}\hole{t_1}\termsl},
  which by source contextual equivalence gives us $\src{\backtrstlc{\ctxt, m}\hole{t_2}\termsl}$ (HPTS2).
  Given $n'$ the number of steps for HPTS2, by \Thmref{thm:compstlc-sem-pres} we have: $\srce\vdash\src{t_2}\anylogreln{n'}\compstlc{\src{t_2}}:\src{\tau}$.
  So by definition: $\srce\vdash\src{t_2}\underlogreln{n'}\compstlc{\src{t_2}}:\src{\tau}$.
  By \Thmref{thm:backtr-corr} (with $n=n'$, $p=\imprecise$ and $\anylogrel=\underlogrel$) we can conclude $\srce\vdash \src{\backtrstlc{\ctxt, m}\hole{t_2}} \underlogreln{n} \ctxht{\compstlc{\src{t_2}}} : \emuldv{m;p;\trgb{\tau}}$.
  By \Thmref{thm:term-def-rel} with HPTS2 we have: $\src{\backtrstlc{\ctxt, m}\hole{t_2}\shrinks{\_}}$ (HPTT2).
  By \Thmref{thm:log-rel-adeq-both} for $\underlogrel$ with HPTT2 we conclude the thesis.
\end{proof}

\begin{theorem}[\compstlcic{\cdot} preserves equivalence]\label{thm:contextual-equivalence-preservation-ic}\label{thm:stlc-pres-ic}
    \[
    \text{If }
        \trge \vdash \trg{t_1} \ceqt \trg{t_2} : \trg{\tat}
      \text{ then }
        \othe \vdash \compstlcic{t_1} \ceqo \compstlcic{t_2} : \compstlcic{\tau}
    \]
\end{theorem}

\begin{theorem}[\compstlcfe{\cdot} preserves equivalence]\label{thm:fe-stlc-pres-fe}
    \[
    \text{If }
        \srce \vdash \src{t_1} \ceqs \src{t_2} : \src{\tau}
      \text{ then }
        \othe \vdash \compstlcfe{t_1} \ceqo \compstlcfe{t_2} : \compstlcfe{\tau}
    \]
\end{theorem}

\subsection{Full Abstraction for the Three Compilers}\label{sec:proofs-fa}
With the two directions of fully-abstract compilation already proved, we can easily show that all three compilers are fully abstract.
As before, full abstraction of \compstlcfe{\cdot} trivially follows from composing full abstraction for the other two compilers.
\begin{theorem}[\compstlc{\cdot} is fully abstract]\label{thm:fac-stlc}
  \[
        \srce \vdash \src{t_1} \ceqs \src{t_2} : \src{\tau} \iff \trge \vdash \compstlc{t_1} \ceqt \compstlc{t_2} : \compstlc{\tau}
      \]
\end{theorem}

\begin{theorem}[\compstlcic{\cdot} ms fully abstract]\label{thm:fac-stlc-ic}\label{thm:fac-stlcie}
  \[
        \trge \vdash \trg{t_1} \ceqt \trg{t_2} : \trg{\tat} \iff \othe \vdash \compstlcic{t_1} \ceqo \compstlcic{t_2} : \compstlcic{\tat}
    \]
\end{theorem}

\begin{theorem}[\compstlcfe{\cdot} is fully abstract]\label{thm:fe-fac-stlc-fe}
  \[
        \srce \vdash \src{t_1} \ceqs \src{t_2} : \src{\tau} \iff \othe \vdash \compstlcfe{t_1} \ceqo \compstlcfe{t_2} : \compstlcfe{\tau}
    \]
\end{theorem}

\section{Mechanisation of the Results}\label{sec:coq}
A full mechanization of all results in this paper in the Coq proof assistant is available at the following url: 
\begin{center}
\url{https://github.com/dominiquedevriese/fixismu-coq}
\end{center}
As the results of this paper are based on the earlier results of \citet{Devriese:2016:FCA:2837614.2837618,popl-journal}, the mechanization is based on the one of \citet{popl-journal}.
It was this mechanization effort which made us notice the errors in our earlier paper-only proofs \citep{isoequi-popl} and it is the mechanization which makes us confident in our current solution based on \NewTerm.
In fact, \NewTerm was first used in the Coq mechanization and subsequently backtranslated - \emph{cough} - to the paper proofs.

The mechanized proof corresponds quite closely to the proofs detailed in this paper, including the addition of \NewTerm.
The main technical challenge is that Coq requires us to be more specific about certain aspects that we gloss over informally on paper.
This includes specifically the fact that all types in \stlcim and \stlcem are closed and that all recursive types must be contractive.
Interestingly, this contractiveness requirement is necessary for our backtranslation from \stlcim to \stlcf to work, but not essential for the meta-theory of \stlcim itself, so we had initially not included the requirement in the definition of the language but treated it only as a precondition of the back-translation. 
This broke down because the meta-theory of \stlcem does not make sense without the contractiveness requirement and embedding potentially uncontractive \stlcim terms into contractive \stlcem terms does not work, so we ended up including the requirement in the definition of \stlcim as well.

\section{Discussion}\label{sec:disc}
At this point, it is useful to take a step back, and reflect on the meaning of our results.
As we have explained, our results demonstrate that iso- and equi-recursive types do not fundamentally alter the expressiveness of the simply typed lambda calculus with term-level recursion.
This result can appear contradictory, since recursive types certainly make it possible to define types and programs that do not exist in the unmodified simply typed lambda calculus.
A simple example is the type of boolean lists $\iso{\mathit{BoolList}} \isdef \iso{\matgen{X}{\Unitt \uplus (\Boolt \times X)}}$.
This type is inexpressible in the simply typed lambda calculus, as is, in fact, any type that can contain values of an a priori unbounded size.
Clearly, the ability to define such types and algorithms that work with it, is useful in a programming language.
But what then does it mean that recursive types do not increase the expressiveness of the language?

To understand this well, it is important to reflect on the meaning of programming language expressiveness.
As we have explained, we use a fully abstract embedding to express equi-expressiveness between the two languages.
Let us investigate the statement of, for example, \Cref{thm:fac-stlc} again, to reflect upon what it means:
\[
    \srce \vdash \src{t_1} \ceqs \src{t_2} : \src{\tau} \iff \trge \vdash \compstlc{t_1} \ceqt \compstlc{t_2} : \compstlc{\tau}
\]

The property states that if two  terms $\src{t_1}$ and $\src{t_2}$ are contextually equivalent in \stlcim{}, then they remain contextually equivalent in \stlcem{}.
To understand what this means for the relative expressiveness of \stlcim{} and \stlcem{}, one should regard the contextual equivalence $\src{t_1} \ceqs \src{t_2}$ as an expressiveness challenge for \stlcim{} contexts.
The property implies that no \stlcim{} context is sufficiently expressive to distinguish the two terms $\src{t_1}$ and $\src{t_2}$.
The fully abstract embedding of \Cref{thm:fac-stlc-ic} then, implies that if such a challenge is unsolvable by \stlcim{} contexts, then it is also unsolvable by \stlcem{} contexts.

It is not difficult to see that other language extensions of \stlcim{} do change the set of contextual equivalences.
For example, adding some form of mutable state would make it easy to distinguish $\src{\lambda f: \Units \to \Units\ldotp f~(f~\units)}$ from $\src{\lambda f:\Units \to \Units\ldotp f~\units}$.
Our results imply that no such expressiveness differences exist between \stlcf{}, \stlcim{} and \stlcem{}.

Essentially, our proof is based on considering how a \stlcem{} or \stlcim{} context solves one of the expressiveness challenges we consider.
Specifically, when a \stlcim{} or \stlcem{} context distinguishes two terms by terminating for one but diverging for another, we cannot simply replicate its behaviour in \stlcf{} because it may have used values of types that are unrepresentable in \stlcf{}.
However, the terminating execution will have taken only a finite amount of steps and in this finite amount of steps, it can only have inspected \stlcim{} or \stlcem{} values up to a finite depth.
Because of this, we can replicate the context's behaviour in \stlcf{} using only finite types by approximating potentially infinite recursive types up to a sufficiently large but finite depth.
It is precisely this approximation of infinite types that we define in our back-translation.

The usage of contextual equivalences as a challenge of expressiveness for program contexts allows us to (1) clarify how our fully abstract embeddings imply a form of equi-expressiveness and (2) understand the limitations of the presented results.
Particularly, the results are crucially based on the observation that the challenge only requires accurately emulating the behaviour of a two particular executions and only up to the point that one terminates while the other doesn't.
We could, for example, consider expressiveness challenges that involve not two programs, but an infinite sequence of programs, in which case, it might not be possible to determine a finite depth of emulation for the back-translation to work.

A well-known infinitary expressiveness challenge, for example, is to take the set of all Turing machines, encoded as integers, and require the context to terminate iff the corresponding Turing machine terminates.
Since \stlcf{} types can only represent finite data types (note the absence of an unbounded integer type), it is not obvious that such a context exists, as Turing machines may use unbounded amounts of memory.
Then again, in the absence of infinite types, it is also impossible to encode the infinite set of Turing machines.
If we did have a type of unbounded naturals or integers, there would automatically be ways to represent infinite memory, for example, as functions of type $\mathbb{N} \to \mathbb{N}$.
As such, it is natural to suspect that such a version of \stlcf{} would be able to semi-decide Turing machine termination, like \stlcim{} and \stlcem{}.

The expressiveness comparison might also yield different results in versions of \stlcf{}, \stlcim{} and \stlcem{} with external effects.
In such a setting, the observable behaviour of an expression might consist of a potentially infinite trace of events rather than termination after a finite amount of steps.
The infinite nature of observable behaviour in such a setting might also make it impossible to determine a bound on the required back-translation depth.
In such a setting, one could imagine an expressiveness challenge that requires the context to produce effectful behaviour that requires an unbounded amount of memory.
For example, we might consider a set of programs that invoke a function in the context, where the context needs to respond to each invocation by printing the full list of values received so far.
If there is an infinite amount of programs without a bound on the amount of values, then the finite memory of \stlcf{} contexts might not allow them to remember all booleans received, unlike \stlcim{} or \stlcem{} contexts.
In such an effectful setting, an infinitary expressiveness challenge might indeed demonstrate a way that recursive types increase the expressiveness of the language. 

In this paper, it is not our goal to investigate in detail such other notions of expressiveness, defined by infinitary expressiveness challenges and/or potentially infinite external effects.
However, it is important to understand that our results naturally pertain to forms of language expressiveness that are measured using finitary expressiveness challenges like full abstraction.
This corresponds to the intuitive understanding that recursive types allow defining potentially infinite types like lists and algorithms that work with them.
We consider it likely that the existence of such types and such algorithms can be detected using appropriately-chosen infinitary expressiveness challenges.
As such, the equi-expressiveness of our full abstraction results should not be taken to mean that recursive types are useless, just that they do not increase the ability of contexts to distinguish pairs of expressions.

\section{Related Work}\label{sec:rw}
Two alternative formulations of equi-recursive types exist: one based on an inductive type equality (which we dub \stlcemind in this section) and one based on a weak type equality (which we dub \stlcemsimp).%
\footnote{
  We typeset these languages in a \eqind{green}, \eqind{verbatim} font, though they appear in this section only.
}
\stlcemind defines an equality relation on types ($\tyeqind$) that, unlike ours, is inductively defined~\citep{syn-con-rec-ty}. %
Types are equal if they are the same (\Cref{tr:eq-t-base-ind,tr:eq-t-var-ind}), when their subparts are equal (\Cref{tr:eq-t-bi-ind,tr:eq-t-mu-ind}) or when one is the unfolding of the other (\Cref{tr:eq-t-fold-ind}).
To keep track of type variables, typing equality is defined with respect to an environment $\eqind{\Delta}\bnfdef\eqinde\mid\eqind{\Delta;\alpha}$.
\begin{center}
\small
\mytoprule{ \tyeqindbin{\Delta }{\tau}{\sigma}}

\typerule{Eq-type-Symmetric}{
  \eqind{\Delta}\vdash\eqind{\tau'}\tyeqind\eqind{\tau}
}{
  \eqind{\Delta}\vdash\eqind{\tau}\tyeqind\eqind{\tau'}
}{eq-t-sym-ind}\and
\typerule{Eq-type-Transitive}{
  \eqind{\Delta}\vdash\eqind{\tau}\tyeqind\eqind{\tau''}
  \\
  \eqind{\Delta}\vdash\eqind{\tau''}\tyeqind\eqind{\tau'}
}{
  \eqind{\Delta}\vdash\eqind{\tau}\tyeqind\eqind{\tau'}
}{eq-t-trans-ind}
\typerule{Eq-type-Bi}{
  \eqind{\star} \in \{\eqind{\to}, \eqind{\times}, \eqind{\uplus}\}
  \\
  \eqind{\Delta}\vdash\eqind{\tau_1}\tyeqind\eqind{\tau_1'}
  &
  \eqind{\Delta}\vdash\eqind{\tau_2}\tyeqind\eqind{\tau_2'}
}{
  \eqind{\Delta}\vdash\eqind{\tau_1\star\tau_2}\tyeqind\eqind{\tau_1'\star\tau_2'}
}{eq-t-bi-ind}\and
\typerule{Eq-type-Base}{
  \eqind{\iota} = \eqind{\Unit} ~\vee~
  \\
      \eqind{\iota} = \eqind{\Bool}
}{
  \eqind{\Delta}\vdash\eqind{\iota}\tyeqind\eqind{\iota}
}{eq-t-base-ind}\and
\typerule{Eq-type-Var}{
  \eqind{\alpha}\in\eqind{\Delta}
}{
  \eqind{\Delta}\vdash\eqind{\alpha}\tyeqind\eqind{\alpha}
}{eq-t-var-ind}\and
\typerule{Eq-type-Mu}{
  \eqind{\Delta,\alpha}\vdash\eqind{\tau}\tyeqind\eqind{\tau'}
}{
  \eqind{\Delta}\vdash\eqind{\mat}\tyeqind\eqind{\matgen{\alpha}{\tau'}}
}{eq-t-mu-ind}\and
\typerule{Eq-type-Unfold}{
  \eqind{\Delta}\vdash\eqind{\tau\subi{\mat}{\alpha}}\tyeqind\eqind{\tau'}
}{
  \eqind{\Delta}\vdash\eqind{\mat}\tyeqind\eqind{\tau'}
}{eq-t-fold-ind}

\end{center}
\citet{Cai:2016:SFE:2914770.2837660} explain that this notion of type equality is strictly weaker than the coinductive one we have used.
For example, they mention two type equalities that do not hold in \stlcemind:
\begin{align*}
  \eqinde \vdash&\ \eqind{\matgen{\alpha}{\alpha \to \Unit}} \ntyeqind \eqind{\matgen{\alpha}{(\alpha \to \Unit) \to \Unit}}
  &
  \eqinde \vdash&\ \eqind{\matgen{\alpha}{\matgen{\beta}{\alpha \to \beta}}} \ntyeqind \eqind{\matgen{\alpha}{\alpha \rightarrow \alpha}}
\end{align*}
To understand why these equalities do not hold in the inductive formulation, consider that no amount of unfolding of a recursive type $\mu$s will ever produce recursive types with a different body. 

\begin{wrapfigure}{R}{.5\textwidth}
\small
\centering
\typerule{Type-\stlcemsimp-fold}{
  \eqind{\Gamma}\vdash\eqind{t}:\eqind{\tau\subi{\mat}{\alpha}}
}{
  \eqind{\Gamma}\vdash\eqind{t}:\eqind{\mat}
}{t-e-fold-s}
\typerule{Type-\stlcemsimp-unfold}{
  \eqind{\Gamma}\vdash\eqind{t}:\eqind{\mat}
}{
  \eqind{\Gamma}\vdash\eqind{t}:\eqind{\tau\subi{\mat}{\alpha}}
}{t-e-unfold-s}
\end{wrapfigure}
\stlcemsimp instead enforces that just a recursive type and its unfolding are equivalent~\citep{ahmedphd,appel-equi,prta,10.1145/800017.800528}.
This leads to more compact typing rules and it does not require a type equivalence relation, effectively this is like \stlcim but without \iso{\fold{}}/\iso{\unfold{}} annotations.

The main difference is that in this last variant, unfoldings can only happen at the top-level of a type of a term (i.e., when terms are of a recursive type themselves).
In both \stlcemind and in our coinductive variant \stlcem, unfoldings can also happen inside the types.
For example, types such as $\equi{(\matgen{\alpha}{B\uplus\alpha})\to B}$ and $\equi{(B\uplus(\matgen{\alpha}{B\uplus\alpha}))\to B}$ are not equivalent in this last variant, because we can unfold $\equi{\matgen{\alpha}{B\uplus\alpha}}$ to $\equi{(B\uplus(\matgen{\alpha}{B\uplus\alpha}))}$ inside the domain of the function type.
These types are however equivalent in \stlcemind and in \stlcem.

Since terms of \stlcemind (or \stlcemsimp) can be typed in \stlcem and their semantics do not vary, our results show that all these different formulations of equi-recursive types are equally expressive. 
Since the approximate backtranslation is needed to deal with the coinductive derivations of \stlcem, we believe that a precise backtranslation akin to that of \citet{max-embed} can be used to prove full abstraction for the compiler from \stlcim to \stlcemind.
We leave investigating this for future work.

As mentioned in \Cref{sec:intro}, the closest work to ours is that of \citet{syn-con-rec-ty}.
Like us, they study the relation between iso- and equi-recursive types and prove that any term typed \stlcim can be typed in \stlcemind and vice versa.
For the backward direction, they insert cast functions which appropriately insert \iso{\fold{}} and \iso{\unfold{}} annotations to make terms typecheck.
Additionally, they use a logic to prove that the terms with the casts are equivalent to the original, but the logic does not come with a soundness proof.
\citeauthor{syn-con-rec-ty} do not connect their results to the operational semantics in any way, unlike ours, and their results cannot be used to derive fully-abstract compilation, as they relate one term and its compilation, not two terms and their compilation.
Finally, it is not clear if \citeauthor{syn-con-rec-ty}'s Theorem 6.8 can be interpreted to imply any form of equi-expressiveness of the two languages.
In fact, what \citeauthor{syn-con-rec-ty} prove is that an equi-recursive term is equal to a back-translated term under a certain equality that is (conjectured to be) almost (but not entirely) sound for observational equivalence in equi-recursive contexts. 
On the other hand, in our setting, the interaction of the same programs with arbitrary contexts provides a measure on the relative expressiveness of those contexts when interacting with the given programs.
This difference is key to make claims about the relative expressive power of languages, as we make.

\smallskip

Fully-abstract compilation derived from fully-abstract semantics models~\citep{faml}, and it has been initially devised to study the relative expressive power of programming languages~\citep{gorla-fa,Mitchell-expr-pow,Felleisen-expr-pow}.%
\footnote{
  Not all these works use the term ``fully-abstract compilation'' but their intuition is the same.
}
Fully-abstract compilation has been widely used to compare process algebras and their relative expressiveness, as surveyed by \citet{exprPA}.
Additionally, researchers have argued that fully-abstract compilation is a feasible criterion for secure compilation~\citep{DBLP:conf/icalp/Abadi98,Kennedy:2006:SNP:1226601.1226605}, as surveyed by~\citet{scsurvey}.

Proofs of fully-abstract compilation are notoriously complex and thus a large amount of work exists in devising proof techniques for it.
Most of these proof techniques require a form of backtranslation \citep{Ahmed:2008:TCC:1411203.1411227,ahmedCPS,nonintfree}.
Precise backtranslations generate source contexts that reproduce the behaviour of the target context faithfully, without any approximation~\citep{max-embed,van_strydonck_linear_2019}.
Approximate backtranslations, instead, generate source contexts that reproduce that behaviour up to a certain number of steps.
The approximate backtranslation proof technique we use was conjectured by \citet{obs-pc-corr-trans} and was used by \citet{popl-journal} to prove full abstraction for a compiler from \stlcf to the untyped lambda calculus (\stlcu).
Unlike these works, we deal with a family of backtranslation types that is indexed by target types.
Additionally, our compilers do not perform dynamic typechecks; they are simply the canonical translation of a term in the source language into the target.
Finally, we remark that our results cannot be derived from \citet{Devriese:2016:FCA:2837614.2837618} since the languages in that paper have no recursive types. 

Interestingly, our current result can be seen as factoring out the first phase of \citet{Devriese:2016:FCA:2837614.2837618}'s compiler;
their result could be seen as composing one of our current results with a second fully abstract compiler from \stlcim to \stlcu, which takes care of dynamic type enforcement.
The full abstraction proof for this second compiler could be a lot simpler with recursive types in the source language, as it would no longer require an approximate backtranslation.
In fact, we believe that reusable sub-results could be factored out from other full abstraction results in the literature too.
For example, we conjecture that one could separate closure conversion from purity enforcement in \citet{max-embed}'s compiler, or separate contract enforcement from universal contract erasure in \citet{van_strydonck_linear_2019}'s compiler.
We hope our experience can inspire other researchers to pay more attention to such factoring opportunities and strive to minimize compiler phases.
In other words, we believe the community could benefit from using a nanopass secure compilation mindset, in the spirit of \emph{nanopass} compilation \citep{sarkar_nanopass_2004}.
Even computationally-trivial nanopasses like ours can be useful as they enrich the power of contexts and simplify secure compilation proofs further downstream.

\section{Conclusion}\label{sec:conc}
This paper demonstrates that the simply typed lambda calculus with iso- and equi-recursive types has the same expressive power.
To do so, it presented three fully-abstract compilers in order to reason about iso- and equi-recursively typed terms interacting over a simply-typed interface and a recursively-typed one. 
The first compiler translates from a simply-typed lambda calculus with a fixpoint operator (\stlcf) to a simply-typed lambda calculus with iso-recursive types (\stlcim).
The second compiler translates from \stlcf to a simply-typed lambda calculus with coinductive equi-recursive types (\stlcem).
These two compilers demonstrate the same expressive power of iso- and equi-recursive types on a simply-typed interface.
The third compiler translates from \stlcim to \stlcem, demonstrating equal expressiveness of iso- and equi-recursive types on a recursively-typed interface.
All fully-abstract compilation proofs rely on a novel adaptation of the approximate backtranslation proof technique that works with families of target types-indexed backtranslation type.

\section*{Acknowledgements}
The authors thank the anonymous reviewers for detailed feedback on an earlier draft as well as Phil Wadler for interesting comments and suggestions and Steven Keuchel for Coq hints.
This work was partially supported: 
  by the German Federal Ministry of Education and Research (BMBF) through funding for the CISPA-Stanford Center for Cybersecurity (FKZ: 13N1S0762),
  by the Italian Ministry of Education through funding for the Rita Levi Montalcini grant (call of 2019);
  by the Air Force Office of Scientific Research under award number FA9550-21-1-0054, and
  by the Fund for Scientific Research - Flanders (FWO).

\newpage
\bibliographystyle{alphaurl}   %
\bibliography{./refs}

\end{document}